\documentclass[11pt]{article}%
\usepackage{amsfonts}
\usepackage{enumerate}
\usepackage{multicol,subcaption}
\usepackage{afterpage}
\usepackage{amsmath,bbm,dsfont,mathrsfs,mathtools,appendix}
\usepackage{amssymb}
\usepackage{graphicx}
\usepackage{fullpage}%
\usepackage{framed}
\usepackage[table]{xcolor}
\usepackage[colorlinks=true,linkcolor=blue,citecolor=green,plainpages=false,pdfpagelabels]%
{hyperref}%
\hypersetup{
	colorlinks,
	linkcolor={blue!60!green},
	citecolor={green!50!yellow!75!black},
	urlcolor={blue!80!black},
	linktoc=all
}

\usepackage{bm}
\usepackage{todonotes}
\providecommand{\U}[1]{\protect\rule{.1in}{.1in}}

\newtheorem{theorem}{Theorem}

\newtheorem{corollary}[theorem]{Corollary}

\newtheorem{lemma}[theorem]{Lemma}

\newtheorem{proposition}[theorem]{Proposition}
	
\newtheorem{remark}[theorem]{Remark}

\newenvironment{proof}[1][Proof]{\noindent\textbf{#1.} }{\ \rule{0.5em}{0.5em}}

\newcommand{\Par}{\Pi}
\newcommand{\cH}{\mathcal{H}}
\newcommand{\R}{\mathbbm{R}}

\newcommand{\C}{\mathbb{C}}

\newcommand{\N}{\mathbb{N}}

\newcommand{\nn}{\nonumber}

\newcommand{\cP}{\mathcal{P}}

\newcommand{\cS}{\mathcal{S}}

\newcommand{\cA}{\mathcal{A}}
\newcommand{\cE}{\mathcal{E}}

\newcommand{\cX}{{\cal X}}

\newcommand{\1}{\mathbbm{1}}

\def\>{{\rangle}}
\def\<{{\langle}}
\newcommand{\be}{\begin{equation}}
	\newcommand{\ee}{\end{equation}}
\newcommand{\bea}{\begin{eqnarray}}
	\newcommand{\eea}{\end{eqnarray}}

\newcommand{\eps}{\varepsilon}

\newcommand{\ceil}[1]{\left\lceil {#1} \right\rceil}
\newcommand{\floor}[1]{\left\lfloor {#1} \right\rfloor}

\newcommand{\ket}[1]{|#1\rangle} %ket
\newcommand{\bra}[1]{\langle#1|} %bra
\newcommand{\kb}[1]{|#1\rangle\!\langle#1|} %ketbra
\def\placeholder{\,\cdot\,}
\newcommand{\Tr}{\mathrm{Tr}}

\newcommand{\comment}[1]{}

\newcommand{\sgn}{\operatorname{sgn}}

\DeclareMathOperator{\KL}{KL}

\newcommand{\NW}{N^{\scalebox{0.65}{$W$}}_{\eps,\delta}}
\newcommand{\NP}{N^{\scalebox{0.65}{$P$}}_{\eps,\delta}}

\newcommand{\tr}{{\rm Tr}}

\numberwithin{equation}{section}
\numberwithin{theorem}{section}

\usepackage{color}
\definecolor{colorthree}{rgb}{0.01,0.51,0.93}

\usepackage{setspace}

\newcommand{\footremember}[2]{%
    \footnote{#2}
    \newcounter{#1}
    \setcounter{#1}{\value{footnote}}
}
\newcommand{\footrecall}[1]{%
    \footnotemark[\value{#1}]%
}

\allowdisplaybreaks
\usepackage{subfiles}

\title{Optimal Fidelity Estimation from Binary  Measurements for Discrete and Continuous Variable Systems}
\author{Omar Fawzi \footremember{Lyon}{Univ Lyon, Inria, ENS de Lyon, LIP, 46 Allee d'Italie, 69364  Lyon Cedex 07, France} \and Aadil Oufkir \footremember{Aachen}{Institute for Quantum Information, RWTH Aachen University, Aachen, Germany}   \and  Robert Salzmann \footrecall{Lyon}}

\date\today
\begin{document}
\maketitle
\begin{abstract}
         Estimating the fidelity between a desired target quantum state and an actual prepared state is essential for assessing the success of experiments. For pure target states, we use functional representations that can be measured directly and determine the number of copies of the prepared state needed for fidelity estimation. In continuous variable (CV) systems, we utilise the Wigner function, which can be measured via displaced parity measurements.  
         We provide upper and lower bounds on the sample complexity required for fidelity estimation, considering the worst-case scenario across all possible prepared states. For target states of particular interest, such as Fock and Gaussian states, we find that this sample complexity is characterised by the $L^1$-norm of the Wigner function, a measure of Wigner negativity widely studied in the literature, in particular in resource theories of quantum computation.
         For discrete variable systems consisting of $n$ qubits, we explore fidelity estimation protocols using Pauli string measurements. Similarly to the CV approach, the sample complexity is shown to be characterised by the $L^1$-norm of the characteristic function of the target state for both Haar random states and stabiliser states. Furthermore, in a general black box model, we prove that, for any target state, the optimal sample complexity for fidelity estimation is characterised by the smoothed $L^1$-norm of the target state. To the best of our knowledge, this is the first time the $L^1$-norm of the Wigner function provides a lower bound on the cost of some information processing task.
         %Even though having a large Wigner $L^1$-norm is known to be a necessary condition for hardness of some information processing tasks, this is the first time, to the best of our knowledge, that it provides also a sufficient condition for hardness.
        
    % punchy sentence in the following spirit:
    %Many previous results showed that large Wigner negativity is necessary for hardness. As far as we are aware, this is the first result showing that large Wigner negativity is sufficient for hardness.

         %The $L^1$-norm of the Wigner function is also called Wigner negativity and prominently appears in the resource theory of quantum computation.

          %For discrete variable systems consisting of $n$ qubits, we explore fidelity estimation protocols using Pauli string measurements.
\end{abstract} 
\tableofcontents

\section{Introduction}
With the possibility to produce a variety of complex and highly entangled quantum states in experiments, it is essential to develop reliable methods certifying that a desired target state has been prepared successfully. For a pure target state $\rho = \kb{\psi}$ and an actually prepared state $\sigma$, an operationally meaningful way to achieve such certification is by estimating the fidelity $F(\rho,\sigma)$ through measurements on the state $\sigma.$  For example, performing the projective measurement $\{\kb{\psi},\1-\kb{\psi}\}$, allows fidelity estimation using a number of copies of $\sigma$ that depends only on the desired accuracy and success probability, and is independent of the target state or system size. However, for most states $\kb{\psi}$ of interest, performing this measurement is not possible in practice in an efficient and reliable manner. Therefore, it is sensible to express the fidelity as the overlap of certain functions representing the quantum states, which can be directly measured with current experimental techniques \cite{Flammia_DirectFidelityEstimation_2011,Silva_PracticalFidel_2011}:

For  continuous variable (CV) quantum systems with $m\in\N$ degrees of freedom, e.g. an $m$-mode photonic system, we use the phase space representation of the fidelity between $\rho$ and $\sigma$ in terms of their \emph{Wigner functions} $W_\rho$ and $W_\sigma$ given by
\begin{align}
\label{eq:WignerFidelityIntro}
   F(\rho,\sigma) = \pi^m \int_{\C^m} W_\rho(\alpha)W_\sigma(\alpha) d\alpha.
\end{align}
Assuming full knowledge of the pure target state $\rho$, the fidelity can hence be estimated by directly measuring the Wigner function $W_\sigma(\alpha)$ at different points in phase space $\alpha\in \C^m$ through \emph{Wigner tomography.} This can be achieved by utilising 
that the Wigner function of $\sigma$ is given by the expectation value of the \emph{displaced parity operator} \cite{Royer_WignerAsExpOfParity_1977}, i.e.
\begin{align}
\label{eq:WignerParityIntro}
    W_\sigma(\alpha) = \left(\frac{2}{\pi}\right)^m\Tr\Big(D(\alpha)\Par D^*(\alpha)\ \sigma\Big),
\end{align}
    with $D(\alpha)$ denoting the \emph{displacement operator} and $\Par=(-1)^N$ the \emph{parity operator} with $N$ being the \emph{$m$-mode number operator} (c.f.~Section~\ref{sec:BasicFactsWignerFunction} for definition). The observable $\left(2/\pi\right)^mD(\alpha)\Par D^*(\alpha)$ has measurement outcomes $\{\left(2/\pi\right)^m,-\left(2/\pi\right)^m\},$ showing that we can estimate $    W_\sigma(\alpha)$ by some binary measurement on the state $\sigma.$ A practical scheme for performing these displaced parity measurements relying on Ramsey interferometry has been proposed in \cite{Lutterbach_DirectMeasWigner_1997}  and experimentally first implemented in \cite{Lutter_ExperimentWigner_2000,Bertet_ExpLutterDirectMeasWigner_2002,Vlastakis_ScienceWignerMeas_2013}. Directly measuring the Wigner function in this way and performing fidelity estimation using \eqref{eq:WignerFidelityIntro} has become conventional in practice, see  \cite{Vlast_Science_2013,ScienceWigner2_2016,reglade_quantum_2024,Antoni_Wigner_2024,Huard_Wignertomoagraphy_2024} and references therein. In the following, we refer to the above-described paradigm as the \emph{Wigner model of fidelity estimation}.

On the other hand, for discrete variable (DV) quantum systems consisting of $n$ qubits, we consider the representation of the fidelity given by
\begin{align}
\label{eq:FidelPauliIntro}
F(\rho,\sigma) = \frac{1}{d}\sum_{P\in\cP_n}\chi_\rho(P)\chi_\sigma(P),
\end{align}
with dimension of the overall system $d=2^n$, $\cP_n =\{\1,X,Y,Z\}^{\otimes n}$ denoting the set of Pauli strings of length $n$ and
\begin{align}
 \chi_\sigma(P)= \Tr(P \sigma)   
\end{align} the \emph{characteristic function} of the state $\sigma.$ Again, assuming full knowledge of the pure target state $\rho,$ we can estimate the fidelity by measuring Pauli strings with binary outcome in $\{1,-1\}$ on the state $\sigma.$
 We refer to this paradigm as the \emph{Pauli model of fidelity estimation} in the following. 

For both of these settings, efficient fidelity estimation algorithms are know for some special families of states, e.g., well-conditioned states including stabiliser states~\cite{Flammia_DirectFidelityEstimation_2011}. However, as far as we are aware, no lower bound for fidelity estimation is known. This raises the following natural question:
\begin{center}
\emph{Are there states for which fidelity estimation is hard?}
\end{center}
More specifically, in the Wigner model of fidelity estimation, are there states for which the sample complexity is arbitrarily large? In the Pauli model, are there states for which the number of samples required is exponential in $n$? Our results give an affirmative answer to both of these questions using Fock states (Eq.~\eqref{eq:SampleFockIntro}) and Haar-random $n$-qubit states (Eq.~\eqref{eq:ThetaRandomPauliIntro}). Going beyond such hard states, we can ask:
%bounds and understand which property of the state $\rho$ governs the sample complexity. 
 \begin{center}
 \emph{Which property of $\rho$ quantifies the sample complexity of fidelity estimation?}
 \end{center}
We identify the smoothed $L^1$-norm as quantifying the optimal sample complexity for fidelity estimation in a general black box model (Eq.~\eqref{eq:SampleCharBlackBoxIntro}).

 %\paragraph{Overview of results} 

\subsection{Main results}
We now describe more precisely the setup for our results. We analyse the \emph{sample complexity of fidelity estimation in the Wigner- and Pauli model}, respectively denoted by $\NW(\rho)$ and $\NP(\rho)$. They are defined as the minimal number of copies of the state $\sigma$ needed to estimate the fidelity $F(\rho,\sigma)$ with accuracy $\eps>0$ and  probability of failure at most $\delta>0$ in the worst case over all states $\sigma$ using the paradigms described around \eqref{eq:WignerFidelityIntro} and \eqref{eq:FidelPauliIntro} respectively. 

 For this, we unify the discussion for continuous and discrete variable systems by considering general overlaps of functions $f$ and $g$ on some measurable space $\Omega$ with measure $\mu.$ Inspired from \eqref{eq:WignerFidelityIntro} and \eqref{eq:FidelPauliIntro}, we choose for CV systems $\Omega_W = \C^m$ with rescaled Lebesgue measure $d\mu_W(\alpha) = \pi^m d\alpha$ and for DV systems $\Omega_P = \cP_n$\footnote{More precisely, we choose in Section~\ref{sec:PauliFidel} the set $\Omega_P=\cP_n\setminus\{\1\}$ instead of $\cP_n$ as on the identity string the characteristic function of any state $\sigma$ trivially satisfies $\chi_\sigma(\1)=1$ and hence no information is gained when measuring the observable $\1$ in $\sigma.
 $} with $\mu_P =\frac{1}{d}\mu_{\text{count}}$ and $\mu_{\text{count}}$ denoting the counting measure.

 The function $f$ corresponds to the Wigner- or characteristic function of the target state $\rho$ in the Wigner- or Pauli model respectively and hence full knowledge of this function is assumed. On the other hand, for the function $g,$ which plays the role of $W_\sigma$ or $\chi_\sigma$ in the above, we assume to have \emph{black box access:} That is for each $\lambda\in\Omega$ we can obtain samples of a random variable $Y_\lambda$ with binary outcome space and which satisfies $\mathbb{E}[Y_\lambda]=g(\lambda).$ 
 
 The objective is then to estimate the overlap 
\begin{align}
\label{eq:OverlapblackboxIntro}
   \int_{\Omega}f(\lambda)g(\lambda)d\mu(\lambda) = \int_{\Omega}f(\lambda)\, \mathbbm{E}[Y_\lambda] d\mu(\lambda)
\end{align}
from as few samples of $Y_\lambda$ at different points $\lambda$ as possible. We call this task \emph{black box overlap estimation} in the following.

Here, $f$ is assumed to be fixed and the function $g,$ for which black box access is provided, is allowed to be any element of some fixed set of functions $\cS$. The \emph{sample complexity of black box overlap estimation}, denoted by $N_{\eps,\delta}(f,\cS),$ is then given by the minimal number of samples of $Y_\lambda$ needed to estimate the overlap \eqref{eq:OverlapblackboxIntro} with accuracy $\eps>0$ and probability of failure at most $\delta>0$ in the worst case over all $g\in\cS.$ In particular, choosing $\Omega$ and $\mu$ for CV and DV quantum systems as outlined above and denoting the set of possible Wigner- and characteristic functions by $\cS_W$ and $\cS_P$ respectively, we can make the following identifications:

\begin{align}
     \nn\text{CV systems: }\quad \ \NW(\rho) &= N_{\eps,\delta}(W_\rho,\cS_W), \\ \text{DV systems: }\ \quad\NP(\rho) &= N_{\eps,\delta}(\chi_\sigma,\cS_P).
\end{align}

We provide upper and lower bounds on the sample complexity of black box overlap estimation $N_{\eps,\delta}(f,\cS)$ for a large class of sets of functions $\cS$: 
%In particular, we can deduce bounds on the sample complexities of fidelity estimation in the Wigner- and Pauli model $\NW(\rho)$ and $\NP(\rho).$ In many instances, we can prove that those upper and lower bounds match and obtain a characterisation of the corresponding sample complexities given in terms of the $L^1$-norm (or a smoothed version of it) of the target function $f$ (or $W_\rho$ or $\chi_\rho$ respectively).

\bigskip 

\noindent \textbf{Black box overlap estimation:}\\
For general measurable spaces $\Omega$ with measure $\mu$ we prove upper and lower bounds on the sample complexity $N_{\eps,\delta}(f,\cS)$ for $f\in L^2(\Omega,\mu)$ and all sets of functions $\cS\subset \cS_{\max}$ where for some fixed number $r>0$ the set $\cS_{\max}$ consists of all functions $g$ satisfying $\|g\|_{L^2(\Omega,\mu)}\le 1$ and $\|g\|_{L^\infty(\Omega,\mu)}\le r$. These conditions on the sets $\cS$ and $\cS_{\max}$ are motivated by the Wigner- and Pauli model as for all quantum states $\rho$ of a continuous (or discrete) variable system the corresponding Wigner (or characteristic) function satisfies $\|W_\rho\|_{L^2(\C^m,\mu_W)}\le 1$ and $\|W_\rho\|_{L^\infty(\C^m,\mu_W)}\le (2/\pi)^m$\ (or $\|\chi_\rho\|_{L^2(\cP_n,\mu_P)}\le 1$ and $\|\chi_\rho\|_{L^\infty(\cP_n,\mu_P)}\le 1$). Hence, under the appropriate choice of $\Omega$ and $\mu$ as above, the sets of Wigner- and characteristic functions for continuous and discrete variable quantum systems respectively provide examples of such sets $\cS,$ i.e. $\cS_W\subset \cS_{\max}$ and $\cS_P\subset\cS_{\max}.$

    For all such sets $\cS,$ we provide an explicit algorithm for estimating the overlap \eqref{eq:OverlapblackboxIntro} yielding an upper bound on $N_{\eps,\delta}(f,\cS)$ (c.f.~Theorem~\ref{thm:blackboxOverEstimation}) and furthermore, through some information theoretic argument, the corresponding lower bound (c.f.~Theorem~\ref{thm:LowerBoundblackbox}). For the particular choice of $\cS=\cS_{\max},$ these upper and lower bounds can be shown to be essentially matching in their scaling in the relevant variables $f,\eps,\delta$ since they are given as 
     \begin{align}
\label{eq:SampleCharBlackBoxIntro}
\left(\frac{\|f\|^{(2\eps)}_1}{2\eps}\right)^2\log\left(\frac{1}{3\delta}\right) \ \lesssim\, N_{\eps,\delta}(f,\cS_{\max}) \,\lesssim\inf_{\eps'\in [0,\eps)}\left(\frac{\|f\|^{(\eps')}_1}{\eps-\eps'}\right)^{2}\log\left(\frac{1}{\delta}\right).
 \end{align}
Here, we have denoted the \emph{smoothed $L^1$-norm} of the function $f$ by \begin{align}
\label{eq:SmoothedL1Intro}
 \|f\|^{(\eps)}_1=\inf_{\substack{\tilde f\in L^1(\Omega)\cap L^2(\Omega),\\\|f-\tilde f\|_2\le \eps}}\|\tilde f\|_1.  
\end{align}
Note that~\eqref{eq:SampleCharBlackBoxIntro} provides an \emph{instance-optimal} characterisation for the task of black box overlap estimation, i.e., it identifies precisely the property of $f$ that governs the sample complexity as the smoothed $L^1$-norm. The study of instance-optimal bounds is well-developed in the learning literature, for example for the problem of identity testing, one of the flagship results is that classical identity testing is characterised by a smoothed version of the $L^{2/3}$-quasinorm of the distribution to be tested~\cite{valiant2017automatic}; see also~\cite{chen2022toward} for results about quantum identity testing in this direction.

For functions $f$ with rapidly decaying tail as defined in Section~\ref{sec:BlackBox} and discussed in detail in Appendix~\ref{sec:Tail}, which include for example exponentially decaying functions on Euclidean space, we find that the sample complexity is characterised by the usual $L^1$-norm as 
\begin{align}
\label{eq:ThetaBlackBoxIntro}
    N_{\eps,\delta}(f,\cS_{\max}) = \Theta\left(\left(\frac{\|f\|_1}{\eps}\right)^2\log\left(\frac{1}{\delta}\right)\right),
\end{align}
with precise statement given in Corollary~\ref{cor:SmaxTheta}.

\bigskip

\noindent\textbf{Fidelity estimation in Wigner model:}
In Theorem~\ref{thm:WignerUpperBound} and~\ref{thm:LowerBoundWigner}
we give upper and lower bounds on the sample complexity of fidelity estimation in the Wigner model, $\NW(\rho),$ for pure states $\rho$ on $L^2(\R^m)$ from the ones established for the abstract framework of black box overlap estimation. In particular, the upper bound involves the smoothed $L^1$-norm of $W_\rho,$ giving that for every fixed pure state $\rho$ and in the worst case over all states $\sigma$, the fidelity $F(\rho,\sigma)$ can be estimated by performing displaced parity measurements on a finite number of copies of $\sigma$, i.e. $\NW(\rho)<\infty.$

Showing for general pure states $\rho$ that these upper and lower bounds on $\NW(\rho)$ are essentially matching similarly to \eqref{eq:SampleCharBlackBoxIntro} and \eqref{eq:ThetaBlackBoxIntro} is, however, hard as the set of Wigner functions, $\cS_W,$ is a proper and complicated subset of $\cS_{\max}=B^2_1(0)\cap B^{\infty}_{(2/\pi)^m}(0).$  Hence, we focus on interesting examples of states of continuous variable systems and establish for those that the sample complexity can, up to constants, be expressed in terms of the $L^1$-norm of the corresponding Wigner function:

Firstly, we consider the Fock states $\left(\kb{n}\right)_{n\in\N},$ i.e. the eigenstates of the harmonic oscillator or number operator. In this case, we find in Proposition~\ref{prop:FockStates} that their sample complexity of fidelity estimation in the Wigner model is fully characterised as
\begin{align}
\label{eq:SampleFockIntro}
     \NW(\kb{n}) = \Theta\left(\Bigg(\frac{\left\|W_{\kb{n}}\right\|_1}{\eps}\Bigg)^2\log\left(\frac{1}{\delta}\right)\right).
\end{align} 
 Here, the $L^1$-norm of the Wigner functions of the Fock states satisfies the scaling behaviour \begin{equation}
     \|W_{\kb{n}}\|_1\sim \sqrt{n}
 \end{equation}
 as shown in \eqref{eq:FockWignerL1} below by straightforwardly employing known asymptotic expressions of the $L^p$-norms of the Laguerre polynomials in \cite{Markett_ScalingLpLaguerre_1982}. Note that shows that there exist states with arbitrarily large sample complexity.

As another example, we construct in Section~\ref{sec:SpreadedSpikeStates} the so-called \emph{spike states} for any $m$-mode continuous variable quantum system. These are pure states $\left(\rho_n\right)_{n\in\N}=\left(\kb{\psi_n}\right)_{n\in\N}$ on the Hilbert space $L^2(\R^m)$ whose Wigner functions are uniformly vanishing as $n\to\infty,$ i.e.
\begin{align}
\label{eq:SpikeVanishIntro}
    \|W_{\rho_n}\|_\infty \xrightarrow[n\to\infty]{} 0.
\end{align}
Furthermore, we show that $W_{\rho_n}\in L^1(\C^m)$ for all $n\in\N$ with $L^1$-norms blowing up and scaling reciprocally to their $L^\infty$-norms, i.e.
 \begin{align}
 \label{eq:SpikeL1Intro}
    \|W_{\rho_n}\|_1 \sim \|W_{\rho_n}\|^{-1}_\infty.
\end{align}
Using \eqref{eq:SpikeVanishIntro} and \eqref{eq:SpikeL1Intro} together with the established bounds on $\NW(\rho),$ we show in Proposition~\ref{prop:SpikeStates} that also for the spike states the corresponding $L^1$-norm characterises the sample complexity of fidelity estimation, i.e. precisely
\begin{align}
\label{eq:SampleSpikeIntro}
    \NW(\rho_{n}) = \Theta\left(\Bigg(\frac{\left\|W_{\rho_{n}}\right\|_1}{\eps}\Bigg)^2\log\left(\frac{1}{\delta}\right)\right).
\end{align}
We believe that the existence of a sequence of pure states satisfying \eqref{eq:SpikeVanishIntro} and \eqref{eq:SpikeL1Intro} can be of independent interest.

For $\rho$ being a Gaussian state, we know that $W_\rho\ge0$ and hence by normalisation of the Wigner function that $\|W_\rho\|_1 = 1.$ From this, we see in Proposition~\ref{prop:Gaussian} that the sample complexity of fidelity estimation in the Wigner model satisfies for all pure Gaussian states $\rho$ satisfies
\begin{align}
\label{eq:GaussianIntro}
 \NW(\rho) = \Theta\left(\frac{1}{\eps^2}\log\left(\frac{1}{\delta}\right)\right)
\end{align}
and hence that for all states $\sigma,$ the fidelity $F(\rho,\sigma)$ can efficiently be estimated using displaced parity measurements on a few copies of $\sigma.$

 %Going beyond the instanced based setup, i.e. with pure target state $\rho$ being fixed, 
  In addition, looking at bounds that only depend on the $L^1$-norm of the state, we find in Theorem~\ref{thm:WorstCaseWigener} a characterisation of the sample complexity in the worst case over all pure states $\rho$ whose Wigner function has $L^1$-norm less or equal than some fixed threshold value. In particular, we see for all $t\ge1$
\begin{align}
\label{eq:WorstCaseWignerIntro}
\sup_{\substack{\rho\, \text{pure state,}\\ \|W_\rho\|_1\le t}}\NW(\rho) =\Theta\left(\left(\frac{t}{\eps}\right)^2\log\left(\frac{1}{\delta}\right)\right),
\end{align}
showcasing that also in the considered worst case, the sample complexity has the same scaling behaviour as the ones for the Fock, spike and Gaussian states portrayed in \eqref{eq:SampleFockIntro}, \eqref{eq:SampleSpikeIntro} and \eqref{eq:GaussianIntro} respectively. 

%In particular, we see from \eqref{eq:WorstCaseWignerIntro} that in the worst case over all pure states $\rho,$ the sample complexity of fidelity estimation in the Wigner model is infinite, i.e. $\sup_{\rho}\NW(\rho)=\infty.$

Negativity of the Wigner function of quantum states is an indication of non-classicality \cite{kenfack_negativity_2004,albarelli2018resource,Kok_NegofQuasi_2020} as it portrays the fact that the Wigner function deviates from being a proper probability distribution, implies that the quantum state cannot be written as a convex combination of coherent or, more generally, Gaussian states and furthermore is equivalent to contextuality
\cite{Ulysse_Cont=WignNeg_2022}.  
It has been shown that Wigner negativity is necessary for quantum advantage, as computations involving only positive Wigner functions are efficiently classically simulable \cite{Bartlett_EffecienClassSimContVar_2002,MariEisert_SimPosWigner_2012,veitch2013efficient}.
The $L^1$-norm of the Wigner function can be regarded as a quantification of Wigner negativity \cite{kenfack_negativity_2004} with it being equal to 1 if and only if the Wigner function is non-negative and otherwise being strictly greater than 1. However, in this context, the $L^1$-norm of the Wigner function only provides an upper bound on the cost of classical simulation~\cite{MariEisert_SimPosWigner_2012,veitch2013efficient,pashayan2015estimating,hahn2024bridging} as there are states with large $L^1$-norm that can still be simulated efficiently~\cite{garcia2020efficient,calcluth2023vacuum}. Interestingly, for the fidelity estimation problem we consider here, we can achieve lower bounds on the sample cost in terms of the Wigner $L^1$-norm. The results presented here in \eqref{eq:SampleCharBlackBoxIntro}, \eqref{eq:SampleFockIntro}, \eqref{eq:SampleSpikeIntro}, \eqref{eq:GaussianIntro} and \eqref{eq:WorstCaseWignerIntro} show an interesting connection between the complexity of learning properties of a quantum state and its non-classicality and we hope the techniques we introduce here, such as the lower bound techniques and the smoothed $L^1$-norm, find further applications in the context of non-classicality.
%In addition, 
%as far as we are aware, this is the first lower bound on the cost of an information processing task using the 

\bigskip

\noindent\textbf{Fidelity estimation in Pauli model:}

\noindent Applying again the results from the abstract paradigm of black box overlap estimation to the Pauli model of fidelity estimation yields upper and lower bounds on the corresponding sample complexity $\NP(\rho)$ for all pure states $\rho$ on a $n$-qubit system, c.f. Theorem~\ref{thm:PauliUpperBound} and Theorem~\ref{thm:LowerBoundPauli}. 

In Proposition~\ref{thm:HaarRandomSample} we use these to find for Haar random pure states $\rho$ with probability at least $1-1/d$ that the corresponding sample complexity is characterised as
\begin{align}
\label{eq:ThetaRandomPauliIntro}
    \NP(\rho) = \widetilde \Theta\left(\left(\frac{\|\chi_\rho\|_1}{\eps}\right)^2\log\left(\frac{1}{\delta}\right)\right)=\widetilde \Theta\left(\frac{d}{\eps^2}\log\left(\frac{1}{\delta}\right)\right),
\end{align}
with $L^1$-norm defined by $\|\chi_\rho\| = \frac{1}{d}\sum_{P\in\cP_n\setminus\{1\}}|\chi_{\rho}(P)|.$
Here, we used $\widetilde \Theta$ to denote the fact that the upper and lower bounds in \eqref{eq:ThetaRandomPauliIntro} match up to factors that scale logarithmically in the leading order term. In particular, the result implies that in the worst case over all pure states $\rho$ the sample complexity blows up as $d\to\infty$, i.e. $\sup_{\rho}\NP(\rho)\xrightarrow[d\to\infty]{}\infty.$ The $L^1$-norm of $\chi_{\rho}$ has appeared in the literature under the name of st-norm in~\cite{campbell2011catalysis} and it was shown in~\cite{hahn2022quantifying} that it provides a lower bound on a magic measure obtained via an encoding into continuous variable systems.

For stabiliser states $\rho$ on $n$ qubits, we have that $\|\chi_\rho\|_1 = \frac{d-1}{d}=\Theta(1)$ (Proposition~\ref{prop:Stab}). Hence, $\NP(\rho)$ is independent of the system size and the fidelity $F(\rho,\sigma)$ can be estimated efficiently by performing Pauli measurements on a few copies of $\sigma$. This was previously shown in \cite{Flammia_DirectFidelityEstimation_2011}.

% We can use this together with the bounds on the sample complexity of fidelity estimation to see in Proposition~\ref{prop:Stab} that
% \begin{align}
% \label{eq:StabliserThetaIntro}
%     \NP(\rho) = \Theta\left(\frac{1}{\eps^2}\log\left(\frac{1}{\delta}\right)\right).
% \end{align}
% Hence, $\NP(\rho)$ is independent of the system size and the fidelity $F(\rho,\sigma)$ can be estimated efficiently by performing Pauli measurements on a few copies of $\sigma$. The upper bound in \eqref{eq:StabliserThetaIntro} can already be found in \cite{Flammia_DirectFidelityEstimation_2011}.

\subsection{Related work}
In \cite{Flammia_DirectFidelityEstimation_2011,Silva_PracticalFidel_2011}, employing the overlap formula \eqref{eq:FidelPauliIntro}, upper bounds on the sample complexity of fidelity estimation in the Pauli model have been established. In particular in \cite{Flammia_DirectFidelityEstimation_2011} the upper bound in Proposition~\ref{prop:Stab} for the sample complexity for stabiliser states has already been found as well as the worst case bound
$\sup_{\rho \text{ pure}}\NP(\rho)= \mathcal{O}\left(\frac{1}{\eps^2\delta^2} + \frac{d}{\eps^2\delta}\log\left(\frac{1}{\delta}\right)\right)$ (using the Markov argument around Eq.~$(10)$ in \cite{Flammia_DirectFidelityEstimation_2011}) or $\sup_{\rho \text{ pure}}\NP(\rho)= \mathcal{O}\left(\frac{1}{\eps^2\delta} + \frac{d}{\eps^4}\log\left(\frac{1}{\delta}\right)\right) $ (using the  ‘‘truncating bad events'' technique)\footnote{In the main text of \cite{Flammia_DirectFidelityEstimation_2011} the authors claim to have proven $\sup_{\rho \text{ pure}}\NP(\rho)= \mathcal{O}\left(\frac{1}{\eps^2\delta}+ \frac{d}{\eps^2}\log\left(\frac{1}{\delta}\right)\right).$ However, considering their proof in the appendix (in particular the ‘‘Truncating bad events" section), the error parameter $\beta$ introduced there needs to scale as $\eps$ in order estimate the fidelity with the desired accuracy leading to a $1/\eps^4$ of their actual upper bound.}.  In Remark~\ref{rem:WorstCasePauli}, we improve these worst case bounds by establishing $\sup_{\rho \text{ pure}}\NP(\rho)= \mathcal{O}\left(\frac{d}{\eps^2}\log\left(\frac{1}{\delta}\right)\right).$ Furthermore, we show that this worst case sample complexity is tight (up to a logarithmic factor in the dimension) for Haar random pure state $\rho$ with high probability (c.f. Proposition~\ref{thm:HaarRandomSample}).

In \cite{Guo2023May}, an upper bound on the sample complexity in terms of the $L^1$-norm is given for the problem of estimating the expectation values of observables in the Pauli model. In this work, we use a similar sampling rule to derive (with a simpler analysis) our upper bounds. Moreover, we  provide nearly matching lower bounds in the case of observables being either Haar random pure states (c.f. Proposition~\ref{thm:HaarRandomSample}) or stabilizer states (c.f. Proposition \ref{prop:Stab}) as well as all   pure observables in the low accuracy regime (c.f. Remark \ref{rk:matching bound-small eps}). 

In \cite{Silva_PracticalFidel_2011,Sivak_HectorPaper_2022} using the overlap formula in terms of Wigner functions \eqref{eq:WignerFidelityIntro} for fidelity estimation of CV systems is discussed, though without establishing bounds on the corresponding sample complexity. 
The work \cite{Sivak_HectorPaper_2022} as well as  \cite{Pashayan2015Aug}, in which a method for estimating the probabilities of outcomes of a quantum circuit is presented, rely on a sampling rule similar to ours (c.f. \eqref{L1sample}). Moreover, in both works it is claimed that this sampling procedure is optimal as it has the smallest variance. We stress that this optimality proof has implicit assumptions: the estimator should be of a particular form (c.f. \cite[Eq. $(14)$]{Pashayan2015Aug} or \cite[Eq. $(5)$]{Sivak_HectorPaper_2022}), the sampling distribution remains the same during  different steps, and the global estimator is given by the empirical one. In fact, we know that for some CV states for which the Wigner function has infinite $L^1$-norm~\cite[Theorem 2]{Simon_WignerWeyl_1992}, estimators satisfying these assumptions have infinite variance but nonetheless our bound~\eqref{eq:SampleCharBlackBoxIntro} gives an algorithm with finite sample complexity.  
Here, we prove our lower bounds on the sample complexity (c.f. Theorems \ref{thm:LowerBoundblackbox}, \ref{thm:LowerBoundWigner} and~\ref{thm:LowerBoundPauli}) without making these assumptions. In particular, we allow the learner to choose adaptively the sampling procedures. This allows us to prove the optimality of the sampling rule \eqref{L1sample} (applied on the corresponding smoothing function $\widetilde f $)  for $\cS= \cS_{\max}$ (c.f. \eqref{eq:SampleCharBlackBoxIntro}), for the Fock, spike and Gaussian states as well as the worst case in the Wigner model (c.f. \eqref{eq:SampleFockIntro}, \eqref{eq:SampleSpikeIntro}, \eqref{eq:GaussianIntro} and \eqref{eq:WorstCaseWignerIntro}) and for Haar random and stabiliser states in the Pauli model (c.f.~\eqref{eq:ThetaRandomPauliIntro} and Proposition~\ref{prop:Stab}).

In \cite{Chabaud_BuildingTrust_2020} the authors consider fidelity estimation of CV systems using heterodyne measurements\footnote{Heterodyne measurements are given by performing the POVM $\left(\frac{1}{\pi^m}\kb{\alpha}\right)_{\alpha\in\C^m}$ defined through the coherent states.} instead of performing Wigner tomography through displaced parity measurements, and establish upper bounds on the corresponding sample complexity. 

The recent paper \cite{Huan_Certifying_2024} considers the task of quantum state certification, i.e. the decision problem of whether the actually prepared state $\sigma$ is close or far from the target state $\rho$, on $n$ qubit systems from Pauli string measurements. Incorporating the full $n$-bit string of measurement outcomes obtained from measuring each Pauli matrix for a given $P\in\cP_n$ individually enables them to prove a $\mathcal{O}(n^2) = \mathcal{O}(\log^2(d))$ upper bound on the corresponding sample complexity for almost all Haar random pure states. This is in sharp contrast to the lower bound we found in \eqref{eq:ThetaRandomPauliIntro}, which is exponential in the number of qubits due to only using the binary measurement outcome from the full Pauli string $P\in\cP_n.$ 

Quantum state tomography for continuous variable systems for different measurement setups has been extensively studied:
In \cite{Lvovsky_ContinuousHomo_2009} quantum state tomography using homodyne detection,\footnote{Homodyne detection is given by measuring the observable $X_\theta = \cos(\theta) X+\sin(\theta)P$ where, $X$ and $P$ are the standard position and momentum operators and $\theta$ is an angle, which can be freely chosen and can be seen as rotating phase space.} a practically implementable measurement scheme, which differs from displaced parity measurements considered here, has been studied.
Furthermore, continuous variable tomography under energy constraints on the unknown state has been a topic of interest in the literature:
Restricting to classical shadow tomography and using homodyne, heterodyne,  photon number resolving\footnote{Photon number resolving measurements are given by performing the POVM $\left(\frac{1}{\pi^m}D(\alpha)\kb{n}D(\alpha)^*\right)_{\alpha\in\C^m}$ for different photon numbers $n\in\N.$} or photon parity measurements, \cite{gandhari2023precisionboundscontinuousvariablestate,Becker_CVLearning_2024} establish performance guarantees by providing upper bounds on the corresponding sample complexity provided the unknown state satisfies certain energy constraints with respect to the photon number operator. In the recent paper \cite{Mele_CVLearning_2024}, lower bounds (for general mixed states) and upper bounds (for pure and Gaussian states) on the sample complexity of full quantum state tomography using general measurements are provided. There, the authors prove their results again under the assumption that the unknown state satisfies an energy constraint. Interestingly, their lower bound implies that the scaling of the sample complexity in terms of error parameter $\eps$ worsen when the number of modes $m$ increases. Note that in comparison, our analysis here does not rely on an energy constraint on either the target state $\rho$ or prepared state $\sigma$. In fact, our results are for arbitrary states $\sigma$ and involve only the (smoothed) $L^1$-norm of the Wigner function of $\rho,$ which can remain bounded for arbitrarily high energy of $\rho.$\footnote{This can for example be seen by considering the displaced state $\rho_\alpha=D(\alpha)\rho D(\alpha)^*$ which satisfies $\|W_{\rho_\alpha}\|_1 = \|W_{\rho}\|_1$ for all $\alpha\in\C^m$ but $\Tr(N\rho_\alpha)\xrightarrow[|\alpha|\to\infty]{}\infty.$ }

\subsection{Outline of the article}
The rest of the article is outlined as follows:
\begin{itemize}
\item[-] In Section~\ref{sec:Prelim} we discuss some preliminary facts which are needed for the derivations of the rest of the paper. 
    \item[-] In Section~\ref{sec:BlackBox} we formally define the task of black box overlap estimation and show upper and lower bounds on the corresponding sample complexity in Theorems~\ref{thm:blackboxOverEstimation} and~\ref{thm:LowerBoundblackbox} and its classification for functions with rapidly decaying tail in Corollary~\ref{cor:SmaxTheta}. 
    \item[-] In Section~\ref{sec:FidelWigner} we formally define the task of fidelity estimation in the Wigner model, discuss its connection to the abstract black box estimation and show upper and lower bounds on the corresponding sample complexity in Theorems~\ref{thm:WignerUpperBound} and~\ref{thm:LowerBoundWigner}. In Section~\ref{sec:MatchWigner} we discuss instances for which these upper and lower bounds can be shown to be matching up to constants including the Fock states in Section~\ref{sec:Fockstates}, spike states in Section~\ref{sec:SpreadedSpikeStates}, Gaussian states in Section~\ref{sec:Gaussian} and the worst case behaviour of the sample complexity in Theorem~\ref{thm:WorstCaseWigener}.
    \item[-] In Section~\ref{sec:PauliFidel} we formally define the task of fidelity estimation in the Pauli model, discuss its connection to the abstract black box estimation and show upper and lower bounds on the corresponding sample complexity in Theorems~\ref{thm:PauliUpperBound} and~\ref{thm:LowerBoundPauli}. We then continue to show in Section~\ref{sec:MatchPauli}  that these upper and lower bounds are essentially matching for Haar random pure states in Proposition~\ref{thm:HaarRandomSample} and for stabiliser states in Proposition~\ref{prop:Stab}.
    \item[-] In the Appendix we study properties of the smoothed $L^1$-norm. 
    In Lemma~\ref{lem:ContSmoothedL1} we show  continuity in the sense of $\lim_{\eps\downarrow 0}\|f\|^{(\eps)}_1=\|f\|_1$. We then provide convergence rates for this limit for functions with rapidly decaying tail in Appendix~\ref{sec:Tail} with explicit example being the Wigner functions of the Fock states.  
\end{itemize}

\section{Preliminaries}
\label{sec:Prelim}
In this section we discuss some preliminary definitions and facts which are needed for the derivations of the main results of this paper.
\subsection{Notation}
For $X$ being a set and $f,g:X\to \R$ being two functions we write $f\lesssim g$ if there exists $C\ge 0$ such that for all $x\in X$ we have $f(x)\le Cg(x)$ and furthermore $f=\mathcal{O}(g)$ if $|f|\lesssim |g|.$ Similarly we write  $f\gtrsim g$ if there exists $c>0$ such that $f(x)\ge cg(x)$ for all $x\in X$ and $f=\Omega(g)$ if $|f|\gtrsim |g|.$ We denote by $f\sim g$ that $f\lesssim g$ and $g\gtrsim f$ and by $f=\Theta(g)$ that $f=\mathcal{O}(g)$ and $f=\Omega(g).$
\subsection{$L^q$-spaces and smoothed $L^1$-norm}
\label{sec:LqNotation}
Let $\Omega$ be some measurable space with some measure $\mu.$ In the following we denote for $q\in[1,\infty]$ by $L^q(\Omega,\C)$ and $L^q(\Omega,\R)$ the Banach spaces of complex and real valued $L^q$-functions respectively. For brevity we often denote for the case of complex valued functions $L^q(\Omega)\equiv L^q(\Omega,\C).$  We denote the ball of real valued $L^q$-functions with radius $r\ge 0$ centered around some $f\in L^q(\Omega,\R)$ by
\begin{align}
\label{eq:RealBalls}
    B^q_r(f) = \left\{g\in L^q(\Omega,\R) \Big|\, \|f-g\|_q\le r\right\} \quad \text{ where } \quad \| f \|_{q} = \left(\int_{\Omega} |f(\lambda)|^{q} d\mu(\lambda)\right)^{1/q}.
\end{align}
The \emph{smoothed $L^1$-norm}\footnote{Note that the smoothed $L^1$-norm is not a norm itself as it does not satisfy triangle inequality or faithfulness due to the employed infimum.} of a function $f\in L^2(\Omega,\R)$ for $\eps\ge 0$ is defined to be \begin{align}
\label{eq:SmoothedL1Norm}
    \|f\|^{(\eps)}_{1} = \inf_{\substack{\widetilde f\in L^1(\Omega,\R)\\ \|f-\widetilde f\|_2 \le \eps}} \|\widetilde f\|_1.
\end{align} 
As seen in Lemma~\ref{lem:ContSmoothedL1} in the Appendix, the smoothed $L^1$-norm satisfies the natural continuity property $\lim_{\eps\downarrow 0}\|f\|^{(\eps)}_1=\|f\|_1.$

Note that for all $f\in L^2(\Omega,\R)$ and $\eps>0$ the respective smoothed $L^1$-norm is always finite, i.e. $\|f\|^{(\eps)}_1<\infty,$ which follows since $L^1(\Omega,\R)\cap L^2(\Omega,\R)$ is dense in $L^2(\Omega,\R)$.\footnote{Density can be seen by defining for $c>0$ the function $\tilde f_c = f\,\mathbf{1}_{\{|f| \ge c\}}$ and noting that by the dominated convergence theorem it approximates $f$ in $L^2$-norm as $\lim_{c\to 0}\int_\Omega |\tilde f_c(\lambda)-f(\lambda)|^2d\mu(\lambda)=\lim_{c\to 0}\int_{\{|f| < c\}} |f(\lambda)|^2 d\mu(\lambda) =0.$ Furthermore, note that for all $c>0$ we have $\tilde f_c\in L^1(\Omega,\R)$ as $\int_{\{|f| \ge c\}} |f(\lambda)| d\mu(\lambda)\le \|f\|^2_2/c. $}

%Note that in the two cases we are interested in this article, i.e. finite measure spaces, $\mu(\Omega)<\infty,$ or $\Omega =\R^m$ with $\mu$ being the Lebesgue measure, the smoothed $L^1$-norm with $\eps>0$ is always finite for all $f\in L^2(\Omega,\R).$ This is true because for finite measure spaces any $L^2$-function is immediately also in $L^1$ and for the other case because $L^1(\R^m,\R)\cap L^2(\R^m,\R)$ is dense in $L^2(\R^m,\R).$

\subsection{Information theoretic divergences}

%We denote the Kullback-Leibler (KL) divergence between two probability distributions $P=\left(p_x\right)_{x\in\cX}$ and $Q=\left(q_x\right)_{x\in\cX}$ on some finite alphabet $\cX$ by
%\begin{align*}
%\KL(P\|Q) = \sum_{x\in\cX} p_x\log\left(\frac{p_x}{q_x}\right)
%\end{align*}
%if $q_x=0$ implies $p_x=0$ and $\KL(P\|Q) =\infty$ otherwise. 

The Kullback-Leibler (KL) divergence between two probability distributions $P$ and $Q$ on $\cX$ is defined by
\begin{align*}
\KL(P\|Q) =\int_{\cX}  dP(x)\log\left(\frac{dP}{dQ}(x)\right) 
\end{align*}
%if $q_x=0$ implies $p_x=0$ and $\KL(P\|Q) =\infty$ otherwise. 
whenever $P$ is absolutely continuous with respect to $Q$ and $\KL(P \| Q) = +\infty$ otherwise.

For $p, q\in [0,1]$, we write $\KL(p \| q) = \KL(\mathrm{Ber}(p) \| \mathrm{Ber}(q))$, where $\mathrm{Ber}(p)$ denotes a Bernouilli distribution with parameter $p$. Note that if $\delta \leq 1/2$, and $p \geq 1-\delta$ and $q \leq \delta$ we have $\KL(p \| q) \geq \KL(1-\delta \| \delta) \geq \log(1/(3\delta))$.

Furthermore, the $\chi^2$-divergence is defined by
\begin{align}\label{def:chi2}
    \chi^2(P\|Q) = \int_{\cX} \left( \frac{d P}{d Q}(x) - 1 \right)^2 Q(dx)
    %\sum_{x\in\cX} \frac{(p_x-q_x)^2}{q_x} = \bigg(\sum_{x\in\cX} \frac{p_x^2}{q_x}  \bigg)-1
\end{align}
if $P$ is absolutely continuous with respect to $Q$ and $+\infty$ otherwise.

%if $q_x=0$ implies $p_x=0$ and $\chi^2(P\|Q) =\infty$ otherwise.

These divergences satisfy the well-known bound
\begin{align}
\label{eq:KLDomChi}
    \KL(P\|Q)\le \chi^2(P\|Q).
\end{align}
\subsection{Wigner functions for continuous variable systems}

\label{sec:BasicFactsWignerFunction}
For $m\in\N$ we consider a $m$-mode quantum harmonic oscillator on the Hilbert space $\cH = L^2(\R^m).$ For a point in phase space $\alpha\in\C^m\cong \R^{2m}$ the \emph{displacement operator} is defined as
\begin{align}
D(\alpha) = e^{\sum_{k=1}^m\left(\alpha_ka^*_k -\alpha_k^* a_k\right)} =\bigotimes_{k=1}^me^{\alpha_ka^*_k -\alpha_k^* a_k},
\end{align}
where we denoted by
 $a_k$ and $a^*_k$ the canonical bosonic annihilation and creation operators of the $k^{th}$ mode (see e.g.~\cite[Chapter 5.2]{BratteliRobinson_operator2_1997} or \cite[Chapter 8.3]{Teschl_mathematical_2014} for a definition and discussion). Note that $D(\alpha)$ is a unitary operator with $D^*(\alpha) =D(-\alpha)$ and satisfying the relation $D(\alpha)D(\beta) = e^{(\alpha\cdot\beta^*-\alpha^*\cdot\beta)/2}D(\alpha+\beta),$ where we denoted the real dot product between two complex vectors $\alpha,\beta\in\C^m$ by $\alpha\cdot\beta  = \sum_{k=1}^m\alpha_k\beta_k.$
 
 For a state $\rho$ we can define its \emph{characteristic function} as 
 \begin{align}
     \chi_\rho(\alpha) \coloneqq \Tr\left(\rho D(\alpha)\right).
 \end{align}
 The characteristic function is complex valued and satisfies $\chi^*_\rho(\alpha) = \chi_\rho(-\alpha)$ and furthermore $\chi_\rho \in L^q(\C^m)\cong L^q(\R^{2m})$ for all $q\in[2,\infty]$ \cite{Werner_QuantumHarmonicAnal_1984}. Therefore, we can define the so-called \emph{Wigner function} of $\rho$ as the (symplectic) Fourier transform of its characteristic function, i.e. formally\footnote{The integral expression \eqref{eq:WignerDefChar} is well-defined for $\chi_\rho$ being a Schwartz function or more generally in $L^1(\C^m)\cap L^2(\C^m).$ For general functions in $L^2(\C^m)$ the Fourier transform is then extended by density in the usual way, see e.g. \cite[Section 7.1]{Teschl_mathematical_2014}.} through the integral expression
\begin{align}
\label{eq:WignerDefChar}
    W_\rho(\alpha) \coloneqq \frac{1}{\pi^{2m}}\int_{\C^m} e^{\alpha\cdot\beta^*-\alpha^*\cdot\beta} \chi_\rho(\beta)d \beta = \frac{1}{\pi^{2m}}\int_{\C^m} e^{2i\Im(\alpha\cdot \beta^*)} \chi_\rho(\beta)d \beta,
\end{align}
with $d\beta \equiv d\Re(\beta)d\Im(\beta)$ denoting the Lebesgue measure on $\C^m.$ Note, that the Wigner function \eqref{eq:WignerDefChar}
is real valued and can have positive and negative values. For $\rho=\kb{\psi}$ being a pure state, the Wigner function is directly given by
\begin{align}
\label{eq:WignerPure}
    W_\rho(x,p) = \left(\frac{2}{\pi}\right)^m\int_{\R^m}\psi^*(x+y)\psi(x-y) e^{i2p\cdot y}dy,
\end{align}
where we changed coordinates to position and momentum variables $(x,p)\in\R^{2m}$ through the relation $\alpha = \frac{1}{\sqrt{2}}\left(x+ip\right).$ Note the integral in \eqref{eq:WignerPure} is well-defined as the integrand is in $L^1(\R^m)$ and furthermore that this definition of the Wigner function can be extended to general mixed states by linearity.

The Wigner function $W_\rho$ for any state $\rho$  is a \emph{quasiprobability distribution} as it can have positive and negative values and satisfies the normalisation\footnote{To be precise, this property holds rigorously for all states $\rho$ such that $W_\rho\in L^1(\C^m)$ \cite[Proposition 18]{de_gosson_wigner_2017}, which is a proper subset of the set of all states as discussed below.}
\begin{align}
\label{eq:WignerNormalisation}
    \int_{\C^m} W_\rho(\alpha) d\alpha = \frac{1}{2^m}\int_{\R^{2m}} W_\rho(x,p) dxdp = 1.
\end{align}
An important class of states which have non-negative Wigner function are \emph{Gaussian states} with defining property being that their Wigner (or equivalently characteristic) function is a Gaussian function. Notably, by Hudson's theorem \cite{HUDSON1974249,Soto_HudsonThMultimode_1983}, a pure state has non-negative Wigner function if and only if it is Gaussian. On the other hand, the set of general mixed states with non-negative Wigner function has been shown to be strictly bigger than the convex hull of all Gaussian states and is itself hard to characterise \cite{Mandilara_ExtendingHudson_2009,Radim_PosWignerbeyonMixGauss_2011}.

The Hilbert-Schmidt inner product between two states $\rho$ and $\sigma$ can by Plancherel's identity be expressed as the inner product of the corresponding characteristic- or Wigner functions as
\begin{align}
\label{eq:WignerOverlap}
\Tr(\rho\sigma) = \frac{1}{\pi^m}\int_{\C^m} \chi^*_\rho(\alpha)\chi_\sigma(\alpha)d \alpha = \pi^m\int_{\C^m} W_\rho(\alpha)W_\sigma(\alpha)d \alpha. 
\end{align}
Hence, in particular we know for $\rho$ being a state that
\begin{align}
\label{eq:L2Wigner}
    \left\|W_\rho\right\|_2 = \frac{1}{\pi^{m/2}}\|\rho\|_2 \le  \frac{1}{\pi^{m/2}}
\end{align}
where we denoted the Hilbert-Schmidt norm of $\rho$ by $\|\rho\|_2$ and used the fact that it is dominated by the trace norm.
Alternatively to \eqref{eq:WignerDefChar}, we can write the Wigner function of $\rho$ as the expectation value of the \emph{displaced parity operator} \cite{Royer_WignerAsExpOfParity_1977}, i.e.
\begin{align}
\label{eq:WignerParity}
    W_\rho(\alpha) = \left(\frac{2}{\pi}\right)^m\Tr\Big(\rho\, D(\alpha)\Par D^*(\alpha)\Big)
\end{align}
with $\Par$ being the \emph{parity operator} defined through the relation $(\Par\psi)(x) \coloneqq \psi(-x)$ for all $\psi\in L^2(\R^{m})$ or alternatively in Fock basis $$\Par = \sum_{n_1,\cdots,n_m\in\N_0}(-1)^{n_1+\cdots+n_m} \kb{n_1,\cdots,n_m}.$$ 
From that we see that the Wigner function is bounded with
\begin{align}
\label{eq:LinftyWigner}
\|W_\rho\|_\infty \le \left(\frac{2}{\pi}\right)^m.
 \end{align}
 From \eqref{eq:L2Wigner} and \eqref{eq:LinftyWigner} and Hölder's inequality we see that $W_\rho \in L^q(\C^m)$ for all $q\in[2,\infty].$ However, for all $q \in[1,2)$ there exists explicit (pure) states $\rho$ such that $W_\rho \notin L^q(\C^m)$ \cite[Theorem 2]{Simon_WignerWeyl_1992}. The set of functions $\psi\in L^2(\R^m)$ whose Wigner function of the corresponding (possibly unnormalised) pure state satisfies
 $W_{\kb{\psi}}\in L^1(\C^m)$ is hence a proper subset of $L^2(\R^m).$ It is called the \emph{Feichtinger algebra} and posses rich mathematical properties which are extensively studied in the literature \cite[Chapter 7]{de_gosson_wigner_2017}.

 \comment{The relation \eqref{eq:WignerParity} provides a way to measure the Wigner function experimentally \cite{Lutterbach_DirectMeasWigner_1997,Banaszek_DirectMeasurWigner_1999,Winkelmann_DirectMeasWigner_2022}: For that first displace the corresponding state $\rho$ by applying the unitary $D^*(\alpha) = D(-\alpha)$ and then then measure the photon number parity of the resulting state.}

\section{Black box overlap estimation}
\label{sec:BlackBox}

We consider some measurable space $\Omega$ with measure $\mu$ and employ the notation of Section~\ref{sec:LqNotation}. Let $r>0$ be some number, which is fixed throughout the whole section. We say one has \emph{black box access to some function} $g\in B^\infty_r(0)$ if for every $\lambda\in\Omega$ one can obtain samples from the unique\footnote{Note that the condition on $Y_\lambda$ uniquely determines its distribution to be $\mathbb{P}(Y_\lambda =r) = \frac{1+g(\lambda)/r}{2}$ and $\mathbb{P}(Y_\lambda =-r) = \frac{1-g(\lambda)/r}{2}$} random variable $Y_\lambda\equiv Y_\lambda(g)$ with outcome space $\{r,-r\}$ and \begin{align}
\label{eq:YAssumptions}
    \mathbb{E}[Y_\lambda] = g(\lambda).
\end{align} In the following we consider subsets $\mathcal{S} \subseteq  B^\infty_r(0)$ and assume to have black box access to the functions $g\in \cS.$

With that terminology the task of \emph{black box overlap estimation} is defined to be the following:
Let $f\in L^2(\Omega,\R) $ and $\cS\subseteq L^2(\Omega,\R)\cap B^\infty_r(0).$ Then,
provided full access to the function $f$ and black box access to a function $g\in\cS,$ the goal is to estimate 
\begin{align}
\label{eq:Overlapblack box}
    \int_{\Omega}f(\lambda)g(\lambda)d\mu(\lambda) = \int_{\Omega}f(\lambda)\, \mathbbm{E}[Y_\lambda] d\mu(\lambda)
\end{align}
using as few samples as possible of the random variable $Y_\lambda\equiv Y_\lambda(g)$ for different points $\lambda\in\Omega$. More precisely, for $N\in\N,$ a \emph{$N$-query algorithm} $\mathcal{\cA}$ for the task of black box overlap estimation takes points $\lambda_1,\cdots,\lambda_N\in \Omega$ and for each $\lambda_t$ a sample $y_{\lambda_t}\in\{r,-r\}$ from $Y_{\lambda_t}.$ Here, we allow for adaptive algorithms meaning that the points $\lambda_t,$ or their distributions when chosen at random, can depend on $f$ and the previous points $\lambda_1,\cdots,\lambda_{t-1}$ and corresponding samples, i.e.
\begin{align}\label{eq:AlgorihmBlackBoxOverlap}
    \nn\lambda_1\equiv&\,\lambda_1(f) &\longrightarrow \quad &y_{\lambda_1} \ \text{ sample from } Y_{\lambda_1}\\\nn\lambda_2\equiv &\, \lambda_2(f,\lambda_1,y_{\lambda_1}) &\longrightarrow \quad &y_{\lambda_2} \ \text{ sample from } Y_{\lambda_2}\\\nn&\vdots& &\vdots\\\lambda_N\equiv&\,\lambda_N(f,\lambda_1,y_{\lambda_1},\cdots,\lambda_{ N-1},y_{\lambda_{N-1}}) &\longrightarrow \quad &y_{\lambda_N} \ \text{ sample from } Y_{\lambda_N}.
\end{align}
The algorithm then outputs a value $\phi_\mathcal{A}(f,\lambda_1,y_{\lambda_1},\cdots,\lambda_N,y_{\lambda_N})\in\R.$ We demand that $\cA$ works for any function $g\in\cS$, i.e. for all $g\in\cS$ fixed its output should estimate the corresponding overlap \eqref{eq:Overlapblack box}.
We say the algorithm $\cA$ has accuracy $\eps>0$ with failure probability bounded by $\delta>0$ if for all $g\in\cS$
\begin{align*}
\mathbb{P}\left(\,\left|\phi_\mathcal{A}(f,\lambda_1,y_{\lambda_1},\cdots,\lambda_N,y_{\lambda_N})- \int_{\Omega}f(\lambda)g(\lambda)d\mu(\lambda)\right| >  \eps\right) \le \delta. 
\end{align*}
The optimal sample complexity of black box overlap estimation for the function $f$ and set $\cS$ is hence defined as
\begin{align*}
N_{\eps,\delta}(f,\cS) = \inf\left\{N\in\N \Big| \,\exists  N\text{-query algorithm with accuracy $\eps$ and failure probability bounded by $\delta$}\right\}.
\end{align*}

 Theorem~\ref{thm:blackboxOverEstimation} below, with proof given in Section~\ref{sec:ProofUpperBlack}, proposes an algorithm achieving this objective and provides an upper bound on the needed number of points $\lambda$ to take samples from $Y_\lambda.$ Furthermore, Theorem~\ref{thm:LowerBoundblackbox} provides corresponding lower bounds and is proved in Section~\ref{sec:ProofLowerBlack}.
 We see that in the case of $\cS\equiv\cS_{\max}=B^2_1(0)\cap B^\infty_r(0),$ the optimal sample complexity is characterised by the smoothed $L^1$-norm of $f$ (see~\eqref{eq:SmoothedL1Intro} for the definition) as we show for all $\eps,\delta>0$
 \begin{align}
 \label{eq:SmoothedUpperLowerSmax}
\left(\frac{r\|f\|^{(2\eps)}_1}{2\eps}\right)^2\log\left(\frac{1}{3\delta}\right) \ \le\, N_{\eps,\delta}(f,\cS_{\max}) \,\le \ 2\inf_{\eps'\in [0,\eps)}\left(\frac{r\|f\|^{(\eps')}_1}{\eps-\eps'}\right)^{2}\log\left(\frac{1}{\delta}\right).
 \end{align}

\begin{theorem}[Upper bound for black box overlap estimation]\label{thm:blackboxOverEstimation}
   Let $f\in L^2(\Omega,\R),$ $r>0$ and $\cS \subseteq B^2_1(0)\cap B^\infty_r(0).$ Then we have  
   \begin{align}
   \label{eq:blackboxUpper}
       N_{\eps,\delta}(f,\cS) \le 2\inf_{\eps'\in [0,\eps)}\left(\frac{r\|f\|^{(\eps')}_1}{\eps-\eps'}\right)^{2}\log\left(\frac{1}{\delta}\right)
   \end{align}
   The upper bound holds true even when restricting to non-adaptive algorithms, i.e. where the $\lambda_t$ in \eqref{eq:AlgorihmBlackBoxOverlap} can depend on the function $f$ but not on the previous points $\lambda_1,\cdots,\lambda_{t-1}$ and the corresponding samples $y_{\lambda_1},\cdots,y_{\lambda_{t-1}}.$
\end{theorem}
\begin{remark}
The algorithm proposed in the proof of Theorem~\ref{thm:blackboxOverEstimation} also works more generally if instead of having access to the two-outcome random variable $Y_\lambda,$ we could sample from a bounded random variable $\tilde Y_\lambda$ satisfying
\begin{align}
    |\tilde Y_\lambda|\le r\quad\quad\text{and}\quad\quad\mathbbm{E}[\tilde Y_\lambda]=g(\lambda)
\end{align}
for all $\lambda\in\Omega.$ In fact, for this slightly changed model of black box overlap estimation, we obtain the same upper bound as \eqref{eq:blackboxUpper} for the corresponding sample complexity.
\end{remark}

\begin{theorem}[Lower bound for black box overlap estimation]\label{thm:LowerBoundblackbox}
   Let $ f\in L^2(\Omega,\R),$  $r>0$ and $\cS\subseteq L^2(\Omega,\R)\cap B^{\infty}_r(0).$ Then for $\eps,\delta >0$ the sample complexity of black box overlap estimation is lower bounded as
    \begin{align}
   \label{eq:LowerBoundblackboxSup}
       N_{\eps,\delta}(f,\cS) \ge \sup_{\substack{g_1,g_2\in\cS,\\ \left|\int \! f(g_1-g_2)\right| \,>\,2\eps}}\,\left\|\frac{(g_1-g_2)^2}{r^2-g_2^2}\right\|^{-1}_{\infty}\log\left(\frac{1}{3\delta}\right).
   \end{align}
   Furthermore, for $\cS\equiv\cS_{\max}= B^2_1(0)\cap B^\infty_r(0)$ and $f\in\cS_{\max}$ with $\|f\|_2=1$ this gives
   \begin{align}
   \label{eq:LowerBoundblackbox}
       N_{\eps,\delta}(f,\cS) \ge \left(\frac{r\|f\|^{(2\eps)}_1}{2\eps}\right)^2\log\left(\frac{1}{3\delta}\right).
   \end{align}

\end{theorem}
\comment{\begin{remark}
A slightly weakened version of \eqref{eq:LowerBoundblackbox} can be proven to  hold, without assuming $f\in\cS$ with $\|f\|_2=1.$ In this case, merely assuming $\cS= B^2_1(0)\cap B^\infty_r(0)$, we see from \eqref{eq:LowerBoundblackboxSup} by using the same argument as in the proof of Theorem~\ref{thm:LowerBoundblackbox} that 
  \begin{align}
       N_{\eps,\delta}(f,\cS) \ge \sup_{\eps'>\eps}\left(\frac{r\|f\|^{(2\eps')}_1}{2\eps'}\right)^2\log\left(\frac{1}{3\delta}\right).
   \end{align}
   
\end{remark}}
For a special class of functions $f$, satisfying a certain decay of tail condition, we can show in Corollary~\ref{cor:SmaxTheta} below that the upper and lower bounds in \eqref{eq:SmoothedUpperLowerSmax} match up to constants. Here, it is important to have good control over the behaviour of the smoothed $L^1$-norm $\|f\|^{(\eps)}_1$ for different values of $\eps.$ As we see below, the sample complexity of black box overlap estimation is in this case then characterised by the usual $L^1$-norm of the function $f.$

For $\gamma,\kappa> 0$ and $\Omega_0\subseteq \Omega$ a finite measure set, we say a measurable function $f$ has a \emph{rapidly decaying tail of order $(\gamma,\kappa)$ outside $\Omega_0$} if for all $\delta>0$ we have 
\begin{align}
\label{eq:DecayingTailBlackBlox}
    \int_{\{|f|\le \delta\}\cap \Omega^c_0}\, |f(\lambda)| d\mu(\lambda) \le \kappa \,\delta^\gamma.
\end{align}
Note that a function $f\in L^2(\Omega,\R)$ satisfying this condition can be also shown to be in $L^1(\Omega,\R).$
Functions satisfying this condition are discussed in detail in Appendix~\ref{sec:Tail}.  In particular, in Lemma~\ref{lem:SmoothedLowerBound} it is shown that their smoothed $L^1$-norm satisfies the lower bound
\begin{align}
\label{eq:SmoothedLowerBlackBox}
    \|f\|^{(\eps)}_1 \ge \|f\|_1 -\kappa' \eps^{\beta} - \eps\sqrt{\mu(\Omega_0)}
\end{align}
with $\beta := \tfrac{\gamma}{\gamma+1}\in (0,1]$ and $\kappa':=2\kappa^{\frac{1}{1+\gamma}}.$
\begin{corollary}
\label{cor:SmaxTheta}
 Let $f\in\cS_{\max}= B^2_1(0)\cap B^\infty_r(0)$ with $\|f\|_2=1$ and rapidly decaying tail of order $(\gamma,\kappa)$ outside $\Omega_0$ for some $\gamma,\kappa>0$ and finite measure set $\Omega_0\subseteq\Omega$ such that $\kappa,\sqrt{\mu(\Omega_0)}\le c\|f\|_1$ for some universal constant $c\ge 0.$ Then we have 
\begin{align}
    N_{\eps,\delta}(f,\cS_{\max}) = \Theta\left(\left(\frac{r\|f\|_1}{\eps}\right)^2\log\left(\frac{1}{\delta}\right)\right)
\end{align}
for $\eps,\delta>0$ small enough, i.e. for all $\eps,\delta\in(0,c']$ where $c'$ depends on $c$ but not on $f.$
\end{corollary}
\begin{proof}
We get from Theorem~\ref{thm:blackboxOverEstimation} and the definition of the smoothed $L^1$-norm in \eqref{eq:SmoothedL1Norm} the upper bound
\begin{align*}
    N_{\eps,\delta}(f,\cS_{\max}) \le 2\left(\frac{r\|f\|_1}{\eps}\right)^2\log\left(\frac{1}{\delta}\right).
\end{align*}
On the other hand from Lemma~\ref{lem:SmoothedLowerBound} and the fact that $\kappa,\sqrt{\mu(\Omega_0)}\le \mathcal{O}(\|f\|_1),$ we have
\begin{align*}
    \|f\|^{(2\eps)}_1\ge \|f\|_1\left(1-C\eps^\beta\right)
\end{align*}
for some $C\ge 0$ independent of $f$ and $\eps.$
Hence, by the 
lower bound \eqref{eq:LowerBoundblackbox} in Theorem~\ref{thm:LowerBoundblackbox} we find for all $\eps,\delta>0$ small enough
\begin{align*}
    N_{\eps,\delta}(f,\cS_{\max}) &\ge \left(\frac{r\|f\|^{(2\eps)}_1}{2\eps}\right)^2\log\left(\frac{1}{3\delta}\right)  \gtrsim \left(\frac{r\|f\|_1}{2\eps}\right)^2\log\left(\frac{1}{3\delta}\right).
\end{align*}
\end{proof}
\subsection{Proof of the upper bound on the sample complexity $N_{\eps,\delta}(f,\cS)$}
\label{sec:ProofUpperBlack}
In this section we prove the upper bound on the sample complexity of black box overlap estimation given in Theorem~\ref{thm:blackboxOverEstimation}. The proof works by providing a specific algorithm for the estimation task: Here, the tester chooses random points $\lambda\in\Omega$ by sampling from the distribution
\begin{align}\label{L1sample}
    p(\lambda) = \frac{|\widetilde f(\lambda)|}{\|\widetilde f\|_1}, 
\end{align}
where $\widetilde f$ is essentially the minimiser in the optimisation of the smoothed $L^1$-norm in \eqref{eq:SmoothedL1Norm}. For every chosen $\lambda,$ a sample of the bounded random variable $Y_\lambda\,\|\widetilde f\|_1\sgn(f(\lambda))$ is obtained. Using these samples, an empirical estimator is calculated which, provided the number of samples is high enough, is close to the overlap 
\begin{align*}
    \int_\Omega \widetilde f(\lambda)g(\lambda) d\mu(\lambda)
\end{align*}
by Hoeffding's inequality \cite{hoeff}. This overlap approximates the desired one in \eqref{eq:Overlapblack box} since $\widetilde f$ and the target function $f$ are close in $L^2$-norm.

In \cite{Flammia_DirectFidelityEstimation_2011,Silva_PracticalFidel_2011} a similar algorithm has been considered for the concrete case of fidelity estimation in the Pauli model. Formulating their approach in terms in the abstract black box overlap estimation framework, the authors consider the distribution
\begin{align}
p_2(\lambda) = \frac{f^2(\lambda)}{\|f\|^2_2},
\end{align}
instead the distribution $p(\lambda)$ above, for sampling $\lambda\in\Omega.$ For every chosen $\lambda$ they then take samples from the random variable $\|f\|^2_2Y_\lambda/f(\lambda)$ to construct an empirical estimator for the overlap \eqref{eq:Overlapblack box}. This random variable is no longer bounded, leading to a worse upper bound on the corresponding sample complexity for many functions $f.$ Recently, \cite{Guo2023May} use 
the same  sampling distribution  as ours for  the problem of  estimation the expectation values of observables in the Pauli model. However, their analysis is different than ours.

\begin{proof}[Proof of Theorem~\ref{thm:blackboxOverEstimation}]
Let $\eps>0$, $\eps'\in[0,\eps)$ \comment{and such that $\|f\|^{(\eps')}_1$ is finite. The existence of such a $\eps'$ can without loss of generality be assumed as otherwise \eqref{eq:blackboxUpper} is trivially satisfied. } and  $\widetilde f\in L^2(\Omega,\R)\cap L^1(\Omega,\R)$ such that $\|f-\widetilde f\|_2\le \eps'.$ Note that such a function always exists since $L^2(\Omega,\R)\cap L^1(\Omega,\R)$ is dense in $L^2(\Omega,\R)$ as remarked in Section~\ref{sec:LqNotation}.
Consider the random variable $\Lambda$ with outcome space $\Omega$ and distributed according to the probability density
\begin{align*}
    p(\lambda) = \frac{|\widetilde f(\lambda)|}{\|\widetilde f\|_1}.  
\end{align*}
Let $g\in\cS\subseteq  B^2_1(0)\cap B^\infty_r(0)$ for which we assume to have black box access and $Y_\lambda\equiv Y_\lambda(g)$ the corresponding random variable with outcome space $\{r,-r\}$ satisfying \eqref{eq:YAssumptions}.
Consider the random variable
\begin{align*}
    X =  Y_\Lambda\  \|\widetilde f\|_1\,\sgn(\widetilde f(\Lambda)).
\end{align*}
Note that by definition we have
\begin{align*}
    \mathbbm{E}[X] = \mathbbm{E}\big[ Y_\Lambda\  \|\widetilde f\|_1\,\sgn(\widetilde f(\Lambda)) \big]=\int_\Omega \frac{|\widetilde f(\lambda)|}{\|\widetilde f\|_1} g(\lambda) \|\widetilde f\|_1\,\sgn(\widetilde f(\Lambda)) d\mu(\lambda)=\int_\Omega \widetilde f(\lambda)g(\lambda) d\mu(\lambda).
\end{align*} 

For $N\in\N$ to be determined later we take samples $\lambda_1,\cdots,\lambda_N$ of $\Lambda$ and then for each $\lambda_i$ a sample $y_{\lambda_i}$ of $Y_{\lambda_i}.$ Furthermore, we define the empirical average 
\begin{align*}
    \overline{X} = \frac{1}{N}\sum_{i=1}^N\ y_{\lambda_i}\ \|\widetilde f\|_1\,\sgn(\widetilde f(\lambda_i)).
\end{align*}
Noting that by definition and assumption \eqref{eq:YAssumptions} the random variable $X$ is bounded as $$\left|X\right|  \le r \|\widetilde f\|_1,$$ we can employ Hoeffding's inequality \cite{hoeff} which gives
\begin{align}
\label{eq:ProbEstimateBlack}
\mathbb{P}\left(\,\left|\,\overline{X} - \int_\Omega \widetilde f(\lambda)g(\lambda) d\mu(\lambda)\right|\ge \eps -\eps'\right) \le  \exp\left(-\frac{N(\eps-\eps')^2}{2r^2\|\widetilde f\|^2_1}\right).
\end{align}
Furthermore note that by Cauchy Schwarz inequality and the assumption that $\|g\|_2\le 1$ we have
\begin{align*}
    \left|\int_\Omega f(\lambda)g(\lambda) d\mu(\lambda) - \int_\Omega \widetilde f(\lambda)g(\lambda) d\mu(\lambda)\right| \le\|f -\widetilde f\|_2 \le \eps'.
\end{align*}
Hence, we see that for $\delta>0$ a total number of    
\begin{align*}
N = 2\left(\frac{r\|\widetilde f\|_1}{\eps-\eps'}\right)^2\log\left(\frac{1}{\delta}\right)  
\end{align*}
 samples suffices to have with probability greater than $1-\delta$:
\begin{align*}
    \left|\,\overline{X} - \int_\Omega \widetilde f(\lambda)g(\lambda) d\mu(\lambda)\right| &\le \eps -\eps'
   \;\;\text{and}\;\; \left|\int_\Omega f(\lambda)g(\lambda) d\mu(\lambda) - \int_\Omega \widetilde f(\lambda)g(\lambda) d\mu(\lambda)\right| \le \eps'
\end{align*}
 and thus, by the triangle inequality, 
give a good approximation of the desired overlap as 
\begin{align*}
  \left|\int_\Omega f(\lambda)g(\lambda) d\mu(\lambda) -  \overline{X}\right| &\le  \left|\int_\Omega f(\lambda)g(\lambda) d\mu(\lambda) - \int_\Omega \widetilde f(\lambda)g(\lambda) d\mu(\lambda)\right|+ \left|\,\overline{X} - \int_\Omega \widetilde f(\lambda)g(\lambda) d\mu(\lambda)\right|
 \\&\le \eps'+ \eps-\eps'= \eps.
\end{align*}
 Since $\eps'\in[0,\eps)$ and $\widetilde f$ under the constraint above were arbitrary, this shows \eqref{eq:blackboxUpper}.

\end{proof}
\subsection{Proof of the lower bound on the sample complexity $N_{\eps,\delta}(f,\cS)$}
\label{sec:ProofLowerBlack}
Here, we prove the lower bound on the sample complexity of black box overlap estimation given in Theorem~\ref{thm:LowerBoundblackbox}. The first part given in \eqref{eq:LowerBoundblackboxSup} follows by some information theoretic argument in which we show hardness of the estimation task by embedding the problem of distinguishing two hypotheses. For the second part in \eqref{eq:LowerBoundblackbox} showing the lower bound for the choice $\cS\equiv\cS_{\max} = B^2_1(0)\cap B^{\infty}_r(0),$ we use the established \eqref{eq:LowerBoundblackboxSup} together with a specific test function $g\in\cS_{\max},$ whose existence is provided in Lemma~\ref{lem:tildeFGLowerbound} below.
We first state and prove this lemma and then continue to give the proof of Theorem~\ref{thm:LowerBoundblackbox}.
\begin{lemma}
\label{lem:tildeFGLowerbound}
Let $f\in L^2(\Omega,\R)$ and $\eps\in(0,\|f\|_2).$ Then there exists  $g\in L^2(\Omega,\R) \cap L^\infty(\Omega,\R)$ such that the following three points hold:
\begin{enumerate}
    \item $\int_\Omega f(\lambda)g(\lambda) d\mu(\lambda) \ge \eps,$ \label{it:Overlap}
    \item $\|g\|_2 \le 1,$ \label{it:L2Norm}
    \item $\|g\|_\infty \le \frac{\eps}{\| f\|^{(\eps)}_1}.$\label{it:LinftyNorm}
\end{enumerate}
\end{lemma}
\begin{remark}
\label{rem:GNoState}
In the context of Sections~\ref{sec:FidelWigner} and \ref{sec:PauliFidel} below we consider the specific choices of $\cS$ being either the set of Wigner functions or characteristic functions for valid quantum states on $L^2(\R^m)$ or $(\C)^{\otimes n}$ respectively. There, the considered $f$ is a Wigner or characteristic function of some pure state $\rho.$ However, the function $g$ constructed in Lemma~\ref{lem:tildeFGLowerbound} does in general not correspond the Wigner or characteristic function of some state $\sigma.$ Hence, in the context of fidelity estimation for continuous or discrete variable systems of Sections~\ref{sec:FidelWigner} and \ref{sec:PauliFidel}, we cannot use the lemma above to establish the corresponding the lower bound \eqref{eq:LowerBoundblackbox} when $\sigma$ is assumed to be a valid state.
\end{remark}
\begin{proof}[Proof of Lemma~\ref{lem:tildeFGLowerbound}]
In the following we explicitly construct a function $\widetilde f \in L^1(\Omega,\R)\cap L^2(\Omega,\R)$ with \begin{equation}
\label{eq:tildefapprox}
\|f -\widetilde f\|_2 \le \eps
\end{equation} and furthermore $g\in L^2(\Omega,\R)\cap L^\infty(\Omega,\R)$ satisfying points \ref{it:Overlap} and \ref{it:L2Norm} and furthermore \begin{align*}
    \|g\|_\infty \le \frac{\eps}{\|\widetilde f\|_1} \le \frac{\eps}{\| f\|^{(\eps)}_1}\,,
\end{align*}
where the last inequality followed by the definition of the smoothed $L^1$-norm in \eqref{eq:SmoothedL1Norm}.

For that, note first of all that by the dominated convergence theorem we have\\  $\lim_{c\to 0}\int_{\{|f| < c\}} |f(\lambda)|^2 d\mu(\lambda) =0$ and hence we can define 
\begin{align*}
    c_\star = \sup\left\{c\in[0,\infty]\,\Big|\, \int_{\{|f| < c\}} |f(\lambda)|^2 d\mu(\lambda) \le \eps^2\right\} >0 .
\end{align*}
\sloppy 
Note that since $\eps<\|f\|_2$ we have $c_\star <\infty.$
Furthermore, since for all $\lambda\in\Omega$ we have $\lim_{c\uparrow c_\star}|f(\lambda)|^2 \mathbf{1}_{\{|f| < c\}}(\lambda) = |f(\lambda)|^2 \mathbf{1}_{\{|f| < c_\star\}}(\lambda),$ we have again by the dominated convergence theorem \begin{align}
\label{eq:Alsobstar}
    \int_{\{|f| <\,  c_\star\}} |f(\lambda)|^2 d\mu(\lambda) = \lim_{c \uparrow c_\star}\int_{\{|f| < \, c\}} |f(\lambda)|^2 d\mu(\lambda) \le \eps^2.
\end{align} 
Therefore, we can define 
\begin{align}
\label{eq:tildef}
    \widetilde f = f \,\mathbf{1}_{\{|f| \ge c_\star\}} \neq 0
\end{align}
which by \eqref{eq:Alsobstar} approximates $f$ in $L^2$-norm, i.e. satisfies  \eqref{eq:tildefapprox}. Note that by definition we have that $\widetilde f\in L^1(\Omega,\R)$ as 
\begin{align*}
    \|\widetilde f\|_1 = \int_{\{|f|\ge c_\star\}} |f(\lambda)|d\mu(\lambda) \le \frac{1}{c_\star}\int_{\{|f|\ge c_\star\}} |f(\lambda)|^2d\mu(\lambda) \le \frac{\|f\|^2_2}{c_\star}.
\end{align*}
We define the function $g$ differently depending on whether $c_\star < \frac{\eps^2}{\|\widetilde f\|_1}$ and $c_\star \ge \frac{\eps^2}{\|\widetilde f\|_1}.$ In the first case $c_\star < \frac{\eps^2}{\|\widetilde f\|_1},$ we define \begin{align*}
    h = f\,\mathbf{1}_{\left\{|f|\, < \,\frac{\eps^2}{\|\widetilde f\|_1}\right\} }
\end{align*} 
and the normalised function $g = h/\|h\|_2.$ By definition of $c_\star$ and using $\frac{\eps^2}{\|\widetilde f\|_1} > c_\star$ we have
\begin{align}
\label{eq:LowerBoundTail}
    \|h\|_2 = \sqrt{\int_{\left\{|f|\, < \,\frac{\eps^2}{\|\widetilde f\|_1}\right\}} |f(\lambda)|^2 \, d\mu(\lambda) } >  \eps
\end{align}
and hence \begin{align*}
    \|g\|_\infty = \frac{1}{\|h\|_2}\left\||f|\,\mathbf{1}_{\left\{|f|\, < \,\frac{\eps^2}{\|\widetilde f\|_1}\right\} }\right\|_\infty \le \frac{\eps}{\|\widetilde f\|_1}. 
\end{align*}
Furthermore, using \eqref{eq:LowerBoundTail} again we see
\begin{align*}
    \int_\Omega f(\lambda) g(\lambda) d\mu(\lambda) = \sqrt{\int_{\left\{|f|\, < \,\frac{\eps^2}{\|\widetilde f\|_1}\right\}} |f(\lambda)|^2 d\mu(\lambda) } \ge \eps.
\end{align*}
Therefore, the function $g$ satisfies the points \ref{it:Overlap} - \ref{it:LinftyNorm}.

On the other hand, in the case $ c_\star \ge \frac{\eps^2}{\|\widetilde f\|_1}$ define 
\begin{align*}
    g =  \frac{\eps}{\|\widetilde f\|_1}\,\sgn(f)\,\mathbf{1}_{\{|f| \ge c_\star\}}.
\end{align*}
Clearly, we have $\|g\|_\infty \le \frac{\eps}{\|\widetilde f\|_1}$ and furthermore by definition of $\widetilde f$ in \eqref{eq:tildef}
\begin{align*}
    \int_\Omega f(\lambda)g(\lambda) d\mu(\lambda) = \frac{\eps}{\|\widetilde f\|_1}\int_{\{|f|  \ge c_\star\}} |f(\lambda)| d\mu(\lambda) = \eps.
\end{align*}
Lastly, we have using $c_\star \ge \frac{\eps^2}{\|\widetilde f\|_1}$ that also
\begin{align*}
    \|g\|^2_2 = \frac{\eps^2}{\|\widetilde f\|^2_1}\,\mu\left(\left\{|f|  \ge c_\star\right\}\right) \le   \frac{\eps^2}{\|\widetilde f\|^2_1 \,c_\star} \int_{\{|f|  \ge c_\star\}}  |f(\lambda)| d\mu(\lambda) = \frac{\eps^2}{\|\widetilde f\|_1 \, c_\star} \le 1
\end{align*}
and hence, also in this case, the defined function $g$ satisfies the points \ref{it:Overlap} - \ref{it:LinftyNorm}.
\end{proof}

\begin{proof}[Proof of Theorem~\ref{thm:LowerBoundblackbox}]
Let $g_1,g_2\in\cS$ such that
 \begin{align}
 \label{eq:BigOverlap}
     \left|\int_\Omega f(\lambda)\left(g_1(\lambda)-g_2(\lambda)\right)d\mu(\lambda)\right| > 2\eps.
 \end{align}
 If such functions do not exist, \eqref{eq:LowerBoundblackboxSup} is trivially satisfied.  
 Consider for $i=1,2$ and $\lambda\in\Omega$ the random variable $Y^{(i)}_\lambda$ with outcome space $\{r,-r\},$ with $r>0$ being the number which was fixed at the beginning of Section~\ref{sec:BlackBox}, and distributed according to
\begin{align}
\label{eq:DefProbY}
    \mathbb{P}\left( Y^{(i)}_\lambda =r\right) = \frac{1+g_i(\lambda)/r}{2},\quad\quad\mathbb{P}\left( Y^{(i)}_\lambda =-r\right) = \frac{1-g_i(\lambda)/r}{2}.
\end{align}
Note, both pairs $(g_1,(Y^{(1)}_\lambda)_{\lambda\in\Omega})$ and $(g_2,(Y^{(2)}_\lambda)_{\lambda\in\Omega})$ satisfy the assumption \eqref{eq:YAssumptions}. 
We consider two situations:
\begin{enumerate}
	\item[$\cH_1:$] The tester has black box access to $g_1$ through the random variable $Y^{(1)}_\lambda,$
	\item[$\cH_2:$] The tester has black box access to $g_2$ through the random variable $Y^{(2)}_\lambda.$
\end{enumerate}

Hence, if \begin{align*}\Big(g,\left(Y_\lambda\right)_{\lambda\in\Omega}\Big)\in \left\{\left(g_1,\left(Y^{(1)}_\lambda\right)_{\lambda\in\Omega}\right)\,,\, \left(g_2,\left(Y^{(2)}_\lambda\right)_{\lambda\in\Omega}\right)\right\},\end{align*} and using \eqref{eq:BigOverlap} a black box overlap estimation algorithm with success probability $1-\delta$ can  find which pair is present only by estimating the overlap $\int_\Omega f(\lambda)g(\lambda)d\mu(\lambda)$ up to additive error $\eps$.  Let $N$ be a sufficient number of steps for this test and let $\lambda_1, \dots, \lambda_N\in\Omega$ be the points the tester chooses to sample the random variable $Y_\lambda$. Let $y_{\lambda_1}, \dots, y_{\lambda_N}\in\{r,-r\}$ be the corresponding obtained samples. As explained around \eqref{eq:AlgorihmBlackBoxOverlap}, we allow for the algorithm to be adaptive, i.e. at each step $t$ the variable $\lambda_t$ can depend on the previous points $\lambda_1,\dots,\lambda_{t-1},$ outcomes $y_{\lambda_1}, \dots, y_{\lambda_{t-1}}$ as well as on the function $f,$ i.e.
\begin{align*}
    \lambda_t\equiv \lambda_t(f,\lambda_1,y_{\lambda_1},\dots, \lambda_{t-1},y_{\lambda_{t-1}}).
\end{align*}
We denote $Z_t = (\lambda_t,Y_{\lambda_t})$ and let $\mathbb{P}_{\cH_1}^{Z_{1}, \dots, Z_{N}}$ be the distribution of $Z_1, \dots, Z_N$ under $\cH_1$ and $\mathbb{P}_{\cH_2}^{Z_{1}, \dots, Z_{N}}$ the law of $Z_1, \dots, Z_N$ under $\cH_2$. Let $\cE$ be the event that the tester returns the first hypothesis.  We have by the data processing inequality of the KL divergence:
\begin{align}
	\KL\left(\mathbb{P}_{\cH_1}^{Z_{1}, \dots, Z_{N}} \Big\| \mathbb{P}_{\cH_2}^{Z_{1}, \dots, Z_{N}}\right)&\ge 	\KL\left(\mathbb{P}_{\cH_1} (\cE)\big\| \mathbb{P}_{\cH_2}(\cE)\right) \notag
	\\&\ge \KL\left(1-\delta\big\| \delta\right) \notag
	\\&\ge  \log\left(\frac{1}{3\delta}\right). \label{LB-KL}
\end{align} 
On the other hand, we have %by the chain rule of the KL divergence, the relations \eqref{eq:KLDomChi}  and \eqref{eq:DefProbY} and 
using the notation $\bm{Z}_{<t}= (Z_{1},\dots, Z_{t-1})$ for $t\in[N]$ that
\begin{align}
	\KL\left(\mathbb{P}_{\cH_1}^{Z_{1}, \dots, Z_{N}} \Big\| \mathbb{P}_{\cH_2}^{Z_{1}, \dots, Z_{N}}\right)&\overset{(a)}{=}  \sum_{t=1}^N 	\mathbb{E}_{\bm{Z}_{<t}\sim \mathbb{P}_{\cH_1}}\left[\KL\left(\mathbb{P}_{\cH_1}^{Z_{t}|\bm{Z}_{<t}} \Big\| \mathbb{P}_{\cH_2}^{Z_{t}|\bm{Z}_{<t}}\right)\right]\notag
 \\&\overset{(b)}{=}  \sum_{t=1}^N 	\mathbb{E}_{(\lambda_t,\bm{Z}_{<t}) \sim \mathbb{P}_{\cH_1} }\left[\KL\left(\mathbb{P}_{\cH_1}^{Y_{\lambda_t}|(\lambda_t,\bm{Z}_{<t})} \Big\| \mathbb{P}_{\cH_2}^{Y_{\lambda_t}|(\lambda_t,\bm{Z}_{<t})}\right)\right]\notag
	\\&\overset{(c)}{\le}   \sum_{t=1}^N 	\mathbb{E}_{(\lambda_t,\bm{Z}_{<t}) \sim \mathbb{P}_{\cH_1} }\left[\chi^2\left(\mathbb{P}_{\cH_1}^{Y_{\lambda_t}|(\lambda_t,\bm{Z}_{<t})} \Big\| \mathbb{P}_{\cH_2}^{Y_{\lambda_t}|(\lambda_t,\bm{Z}_{<t})}\right)\right]\notag
	\\& \overset{(d)}{=}   \sum_{t=1}^N \mathbb{E}_{(\lambda_t,\bm{Z}_{<t})\sim \mathbb{P}_{\cH_1} }\left[\frac{1}{2}\left(\frac{\left(1+g_1(\lambda_t)/r\right)^2}{1+g_2(\lambda_t)/r}+ \frac{\left(1-g_1(\lambda_t)/r\right)^2}{1-g_2(\lambda_t)/r}\right)-1\right]\notag
			\\&= \sum_{t=1}^N \mathbb{E}_{(\lambda_t,\bm{Z}_{<t}) \sim \mathbb{P}_{\cH_1} }\left[\frac{\left(g_1(\lambda_t)-g_2(\lambda_t)\right)^2}{r^2-g^2_2(\lambda_t)}\right] \notag
   \\& \le N\, \left\|\frac{\left(g_1-g_2\right)^2}{r^2-g^2_2}\right\|_\infty, \label{UB-KL}
\end{align}
where in $(a)$, we have used the chain rule of the KL divergence; in $(b)$, we have used that the distribution of $\lambda_t$ conditioned on $\mathbf{Z}_{<t}$ does not depend on $\cH_1$ or $\cH_2$ combined again with the chain rule; in $(c)$, we have used the well known inequality $\KL \le \chi^2$ (c.f.  \eqref{eq:KLDomChi}); in $(d)$ we have used the definition of $Y_\lambda$ (c.f. \eqref{eq:DefProbY}) and the definition of the $\chi^2$-divergence (c.f. \eqref{def:chi2}).  

Therefore, combining both inequalities \eqref{LB-KL} and \eqref{UB-KL}, we obtain
\begin{align*}
   \log\left(\frac{1}{3\delta}\right)\le  \KL\left(\mathbb{P}_{\cH_1}^{Z_{1}, \dots, Z_{N}} \Big\| \mathbb{P}_{\cH_2}^{Z_{1}, \dots, Z_{N}}\right)\le  N\, \left\|\frac{\left(g_1-g_2\right)^2}{r^2-g^2_2}\right\|_\infty
\end{align*}
hence
\begin{align*}
N\ge \left\|\frac{\left(g_1-g_2\right)^2}{r^2-g^2_2}\right\|^{-1}_\infty\log\left(\frac{1}{3\delta}\right). 
\end{align*}
Since $g_1,g_2\in\cS$ under the constraint \eqref{eq:BigOverlap} were arbitrary this shows \eqref{eq:LowerBoundblackboxSup}.

To prove \eqref{eq:LowerBoundblackbox} in the case of $\cS=\cS_{\max}=B^2_1(0)\cap B^\infty_r(0)$ and $f\in\cS_{\max}$ with $\|f\|_2=1,$ we first use that since $g_2 = 0 \in \cS_{\max}$ we can conclude from \eqref{eq:LowerBoundblackboxSup} that
\begin{align*}
    N_{\eps,\delta}(f,\cS_{\max}) \ge \sup_{\substack{ g\in\cS_{\max},\\ \left|\int \! fg\right| \,> \,2\eps}}\,\frac{r^2}{\|g\|^2_\infty}\log\left(\frac{1}{3\delta}\right).
\end{align*}
Furthermore, we can, without loss of generality, restrict to $2\eps <\|f\|_2=1$ as otherwise $\|f\|^{(2\eps)}_1= 0$ and hence \eqref{eq:LowerBoundblackbox} is trivially satisfied.

Under this constraint, we want to show 
\begin{align}
\label{eq:OhMan}
    \sup_{\substack{\, g\in\cS_{\max},\\ \left|\int \! fg\right| \,> \,2\eps}}\,\frac{1}{\|g\|^2_\infty}=\sup_{\substack{g\in\cS_{\max},\\ \left|\int \! fg\right| \,\ge \,2\eps}}\,\frac{1}{\|g\|^2_\infty}.
\end{align}
Note that the left hand side of \eqref{eq:OhMan} is clearly less or equal than the right hand side, so we only need to show the opposite inequality. 
Let $(\eps_n)_{n\in\N}\subseteq(0,\infty)$ be a sequence such that $\eps_n>\eps$ and $\lim_{n\to\infty} \eps_n =\eps.$ Define furthermore 
\begin{align}
    \kappa_n = \frac{2(\eps_n-\eps)}{1-2\eps}
\end{align}
which for $n\in\N$ large enough satisfies $\kappa_n\in(0,1]$ and to which we shall restrict in the following. Let $g\in \cS_{\max}$ such that $\left|\int_{\Omega}f(\lambda)g(\lambda)d\mu(\lambda) \right| \ge2\eps$ and define
\begin{align}
\label{eq:g_nOhMan}
    g_n = (1-\kappa_n) g + \kappa_n \sgn\left(\int_{\Omega}f(\lambda)g(\lambda)d\mu(\lambda)\right)f
\end{align}
Note first of all that since $f,g\in \cS_{\max} =B^2_1(0)\cap B^\infty_r(0)$ we also have by triangle inequality that $g_n\in\cS_{\max}.$ Moreover, we have
\begin{align*}
    \left|\int_\Omega f(\lambda)g_n(\lambda)d\mu(\lambda)\right| = (1-\kappa_n) \left|\int_\Omega f(\lambda)g(\lambda) d\mu(\lambda)\right| + \kappa_n \ge 2\eps_n>2\eps.
\end{align*}
Since, we also have
\begin{align*}
    \|g_n\|_\infty \le (1-\kappa_n)\|g\|_\infty + \kappa_n\|f\|_\infty \xrightarrow[n\to\infty]{}\|g\|_\infty,
\end{align*}
this shows \eqref{eq:OhMan}. 

To conclude \eqref{eq:LowerBoundblackbox} note that we can restrict to $\delta\le 1/3$ as otherwise \eqref{eq:LowerBoundblackbox} is trivially satisfied. For the case $\frac{2\eps}{\|f\|^{(2\eps)}_1}>r$ the statement follows by simply lower bounding the supremum in \eqref{eq:OhMan} by choosing the feasible function $g=f$. On the other hand, for the case $\frac{2\eps}{\|f\|^{(2\eps)}_1}\le r$, we can employ Lemma~\ref{lem:tildeFGLowerbound} and use the function $g$ (with 
parameter $2\eps$) from there as a feasible function which gives
 \begin{align*}
       N_{\eps,\delta}(f,\cS_{\max}) \ge \sup_{\substack{g\in\cS_{\max},\\ \left|\int \! fg\right| \,\ge \,2\eps}}\,\frac{r^2}{\|g\|^2_\infty}\log\left(\frac{1}{3\delta}\right) \ge \left(\frac{r\|f\|^{(2\eps)}_1}{2\eps}\right)^2\log\left(\frac{1}{3\delta}\right)
   \end{align*}
   and finishes the proof.
\end{proof}

\section{Fidelity estimation of continuous variable systems}\label{sec:FidelWigner}
In this section we consider for $m\in\N$ an $m$-mode harmonic oscilator on the Hilbert space $\cH = L^2(\R^m).$ We study the task of estimating the fidelity between a desired pure state $\rho$ on $L^2(\R^m),$ of which we have a theoretical description, and a general state $\sigma,$ which is prepared experimentally, by performing measurements on $\sigma.$ 
Since $\rho$ is a pure state, we can write the fidelity by \eqref{eq:WignerOverlap} as
\begin{align}
\label{eq:FidelWigner}
    F(\rho,\sigma) = \pi^m\int_{\C^{m}} W_\rho(\alpha)W_\sigma(\alpha)d \alpha\end{align}
Assume that we are able to perform measurements of the displaced parity operator \\ $\left(2/\pi\right)^m D(\alpha)\Par D(\alpha)$ with outcomes in $\{(2/\pi)^m,-(2/\pi)^m\}$ on the state $\sigma$ for chosen values of $\alpha\in \C^m$. Here, the values of $\alpha$ can be chosen adaptively and hence depend on the previous values and corresponding measurement outcomes. In the following we refer to this model as the \emph{Wigner model of fidelity estimation.}

For this model, $\eps>0$ and $\delta>0$ being the allowed additive error and probability of failure respectively and pure state $\rho$ fixed, which we call the \emph{instanced based setup,} we denote the \emph{sample complexity of fidelity estimation in the Wigner model} by $N^{\scalebox{0.65}{$W$}}_{\eps,\delta}(\rho)$: That is $N^{\scalebox{0.65}{$W$}}_{\eps,\delta}(\rho)$ is the optimal number of measurements on different copies of the unknown state $\sigma$ needed to estimate $F(\rho,\sigma)$ in the worst case over all states $\sigma.$ 

 Due to the relation \eqref{eq:WignerParity}, the assumption that we are able to perform measurements of the displaced parity operator on $\sigma$ exactly means that we have black box access to the Wigner function $W_\sigma.$ Hence, estimating the fidelity \eqref{eq:FidelWigner} in this setup  becomes a special instance of black box overlap estimation as discussed in Section~\ref{sec:BlackBox}: In particular we can make the identifications
\begin{align}
\label{eq:IdentificationsWigner}
    \nn\Omega &\mapsto \C^m\cong \R^{2m} &&\text{and } \quad d\mu(\lambda) \mapsto  \pi^m d\alpha \\\nn
    r &\mapsto \left(2/\pi\right)^m \\\nn
    f(\lambda) &\mapsto W_\rho(\alpha) &&\text{to which we have full access} \\ \nn
    g(\lambda) &\mapsto W_\sigma(\alpha)&& \text{to which we have black box access} \\ 
      Y_\lambda &\mapsto Y_\alpha &&\text{measurement outcome of observable $\left(2/\pi\right)^m D(\alpha)\Par D(\alpha)$ on state $\sigma.$}
\end{align}
Since $\mu\equiv \mu_W$ is given by the Lebesgue measure rescaled by $\pi^m,$ we introduce the corresponding $L^q$-norms for $q\in[1,\infty)$ for notational clarity and to distinguish them from the usual $L^q$-norms by
\begin{align}
\label{eq:WignerLqnormConvention}
    \|W_\rho\|_{L^q(\C^m,\mu_W)} = \left(\int_{\C^m} |W_\rho(\alpha)|^q d\mu_W(\alpha)\right)^{1/q} = \pi^{m/q} \|W_\rho\|_q. 
\end{align}
Note in particular, with that convention, we have from \eqref{eq:L2Wigner} for all states $\rho$
\begin{align*}
\|W_\rho\|_{L^2(\C^m,\mu_W)} \le 1
\end{align*}
with equality if the state is pure.

For the subset $\cS$ employed in Section~\ref{sec:BlackBox}, we are interested here in the particular choice $\cS_W$ consisting of all Wigner functions $W_\sigma$ for valid states $\sigma.$ Using \eqref{eq:L2Wigner} and \eqref{eq:LinftyWigner}, we see that $\cS_W\subseteq B^2_1(0)\cap B^\infty_{(2/\pi)^m}(0),$ again with balls with respect to the rescaled Lebesgue measure $d\mu(\alpha)=\pi^md\alpha.$ 

With the notation above, the sample complexity of fidelity estimation can be written as
\begin{align}
\label{eq:SampleComplexityFidelity}
    N^{\scalebox{0.65}{$W$}}_{\eps,\delta}(\rho) \equiv  N_{\eps,\delta}(W_\rho,\cS_W)
\end{align}
with $N_{\eps,\delta}(W_\rho,\cS_W)$ being defined in Section~\ref{sec:BlackBox}. 

Formulating the achievability result of black box estimation, Theorem~\ref{thm:blackboxOverEstimation}, for this special case leads to following result:
\begin{theorem}[Upper bound on fidelity estimation for CV systems] \label{thm:WignerUpperBound}
Let $\rho$ be a pure state on $L^2(\R^{m})$ to which we have a theoretical description. Then, for any state $\sigma,$ we can estimate the fidelity $F(\rho,\sigma)$ with precision $\eps>0$ and failure probability at most $\delta>0$ by measuring the displaced parity operator on a number of
\begin{equation}
\label{eq:WignerFidelyyUpperBound}
    N^{\scalebox{0.65}{$W$}}_{\eps,\delta}(\rho) \le 2^{2m+1}\inf_{\eps'\in [0,\eps)}\left(\frac{\|W_\rho\|^{(\eps')}_1}{\eps-\eps'}\right)^{2}\log\left(\frac{1}{\delta}\right)  
\end{equation}
copies of the state $\sigma.$
Note that by the argument around \eqref{eq:SmoothedL1Norm}, this in particular shows $N^{\scalebox{0.65}{$W$}}_{\eps,\delta}(\rho)<\infty.$
 Furthermore, the upper bound holds true even when restricting to non-adaptive algorithms.
\end{theorem}
\begin{proof}
   Here, with the conventions from \eqref{eq:IdentificationsWigner}, the smoothed $L^1$-norm in Theorem~\ref{thm:blackboxOverEstimation} is given by $\|W_\rho\|^{(\eps')}_{L^1(\C^m,\mu_W)} = \pi^m \|W_\rho\|^{(\eps')}_1.$ 
  Hence, we can immediately read off from \eqref{eq:blackboxUpper} that
   \begin{equation*}
    N^{\scalebox{0.65}{$W$}}_{\eps,\delta}(\rho) \le 2\left(\frac{2}{\pi}\right)^{2m}\inf_{\eps'\in [0,\eps)}\left(\frac{\|W_\rho\|^{(\eps')}_{L^1(\C^m,\mu)}}{\eps-\eps'}\right)^{2}\log\left(\frac{1}{\delta}\right)  = 2^{2m+1}\inf_{\eps'\in [0,\eps)}\left(\frac{\|W_\rho\|^{(\eps')}_1}{\eps-\eps'}\right)^{2}\log\left(\frac{1}{\delta}\right).
\end{equation*} 
\end{proof}

Furthermore, from the lower bound on the sample complexity of black box estimation in Theorem~\ref{thm:LowerBoundblackbox}, we find the corresponding lower bound on $N^{\scalebox{0.65}{$W$}}_{\eps,\delta}(\rho):$
\begin{theorem}[Lower bound on fidelity estimation for CV systems]
\label{thm:LowerBoundWigner}
    Let $\rho$ be a pure state on $L^2(\R^{m})$ to which we have a theoretical description. Then for $0<\eps<1/2$ and $0<\delta\le1/3$ we have  that the sample complexity of fidelity estimation in the Wigner model is lower bounded as
\begin{align}
\label{eq:InstanceBasedWigner} 
    \nn N^{\scalebox{0.65}{$W$}}_{\eps,\delta}(\rho) &\ge \sup_{\substack{\sigma_1,\sigma_2 \text{ state,}\\ |F(\rho,\sigma_1)-F(\rho,\sigma_2)| \,> \,2\eps}}\left\|\frac{\left(W_{\sigma_1}-W_{\sigma_2}\right)^2}{r^2 - W^2_{\sigma_2}}\right\|^{-1}_\infty\log\left(\frac{1}{3\delta}\right)\\&\ge\left(\frac{2}{\pi}\right)^{2m}\sup_{\substack{\sigma\text{ state,}\\ F(\rho,\sigma) \,\ge \,2\eps}}\,\frac{1}{\|W_\sigma\|^2_\infty}\log\left(\frac{1}{3\delta}\right)
\end{align}
\end{theorem}
\begin{remark}
Note that the restriction $0<\eps<1/2$ in Theorem~\ref{thm:LowerBoundWigner} is sensible as for allowed additive error $\eps\ge1/2$ one can always output $1/2$ to estimate $F(\rho,\sigma)$ with demanded precision without performing any measurements on $\sigma.$ In other words, for $\eps\ge 1/2$ we have the trivial result
\begin{align*}
    \NW(\rho) = 0.
\end{align*}
\end{remark}
\begin{proof}[Proof of Theorem~\ref{thm:LowerBoundWigner}]

Using \eqref{eq:FidelWigner} and the identifications \eqref{eq:IdentificationsWigner}, we immediately get from \eqref{eq:LowerBoundblackboxSup} in Theorem~\ref{thm:LowerBoundblackbox} that
\begin{align}
\label{eq:FirstLowerWignerproof}
    N^{\scalebox{0.65}{$W$}}_{\eps,\delta}(\rho) \ge \sup_{\substack{\sigma_1,\sigma_2 \text{ state,}\\ |F(\rho,\sigma_1)-F(\rho,\sigma_2)| \,> \,2\eps}}\left\|\frac{\left(W_{\sigma_1}-W_{\sigma_2}\right)^2}{r^2 - W^2_{\sigma_2}}\right\|^{-1}_\infty\log\left(\frac{1}{3\delta}\right).
\end{align}
We now focus on proving the second inequality in \eqref{eq:InstanceBasedWigner}.
Consider for that $\sigma_2$ in \eqref{eq:FirstLowerWignerproof} to be the sequence $\left(\rho_n\right)_{n\in\N}$ from Lemma~\ref{lem:SpreadedWigner} below. Since $\rho$ is pure, \eqref{eq:WeakConvergenceTo0} yields $\lim_{n\to\infty}F(\rho,\rho_n)=0.$ Combining this with the fact that $\lim_{n\to\infty}\|W_{\rho_n}\|_{\infty} = 0,$ c.f. \eqref{eq:LinftyL1SpredWigner}, this gives for every state $\sigma$ that
\begin{align*}
    \left\|\frac{\left(W_{\sigma}-W_{\rho_n}\right)^2}{r^2 - W^2_{\rho_n}}\right\|_\infty \le \frac{\left\|W_{\sigma}-W_{\rho_n}\right\|^2_\infty}{r^2 -\|W_{\rho_n}\|^2_\infty} \xrightarrow[
   n\to\infty ]{} \left(\frac{\|W_{\sigma}\|_\infty}{r}\right)^2
\end{align*}
and hence
\begin{align}
    N^{\scalebox{0.65}{$W$}}_{\eps,\delta}(\rho) \ge \left(\frac{2}{\pi}\right)^{2m} \sup_{\substack{\sigma_\text{ state},\\ F(\rho,\sigma)  \,> \,2\eps}}\,\frac{1}{\|W_{\sigma}\|^2_\infty}\log\left(\frac{1}{3\delta}\right) = \left(\frac{2}{\pi}\right)^{2m} \sup_{\substack{\sigma_\text{ state},\\ F(\rho,\sigma)  \,\ge \,2\eps}}\,\frac{1}{\|W_{\sigma}\|^2_\infty}\log\left(\frac{1}{3\delta}\right),
\end{align}
Here, the last equality follows due to $\eps < 1/2$ and the same argument as for \eqref{eq:OhMan}. Note that the corresponding function $g_n$ defined in \eqref{eq:g_nOhMan} in the context here is the convex combination of two Wigner functions and hence a valid Wigner function itself.
\end{proof}
\subsection{Matching upper and lower bounds for $\NW(\rho)$ for examples of states}
\label{sec:MatchWigner}
In this section we discuss two examples of families of states 
for which the upper and lower bounds on $\NW(\rho)$ provided in Theorem~\ref{thm:WignerUpperBound} and Theorem~\ref{thm:LowerBoundWigner} match. In particular we give a full characterisation of the sample complexity of fidelity estimation in the Wigner model for 
\begin{enumerate}
    \item the Fock states $\left(\kb{n}\right)_{n\in\N}$ corresponding to the eigenstates of the (1-mode) harmonic oscilator,
    \item the \emph{spike states,} which we define in Section~\ref{sec:SpreadedSpikeStates},
    \item Gaussian states,
\end{enumerate} 
in terms of the $L^1$-norms of the corresponding Wigner functions.

In the case of the Fock states the key insight is that the corresponding Wigner functions can explicitly be expressed in terms of Laguerre polynomials, c.f.~\eqref{eq:WignerFockSec} below.
Combining this with the results on the asymptotic behaviour of the Laguerre polynomials in \cite{Askey_LaguerreExpansion_1965,szegho_orthogonal_1975,Markett_ScalingLpLaguerre_1982,landau_bessel_2000}, we show in Proposition~\ref{prop:FockStates} that the upper bound in Theorem~\ref{thm:WignerUpperBound} and the first lower bound in Theorem~\ref{thm:LowerBoundWigner} match, giving that
\begin{align*}
%\label{eq:FockThetaSec}
   \NW(\kb{n}) = \Theta\left(\Bigg(\frac{\left\|W_{\kb{n}}\right\|_1}{\eps}\Bigg)^2\log\left(\frac{1}{\delta}\right)\right)=\Theta\left(\frac{n}{\eps^2}\log\left(\frac{1}{\delta}\right)\right). 
\end{align*}

The spike states on the other hand are a sequence of pure states $\left(\rho_n\right)_{n\in\N}$ on $L^2(\R^m)$ constructed in such a way, that their Wigner functions uniformly vanish as $n$ increases, i.e.
\begin{align}
\label{eq:VanishWignerSec}
    \|W_{\rho_n}\|_\infty\xrightarrow[n\to\infty]{}0.
\end{align}
Furthermore, we show in Lemma~\ref{lem:SpreadedWigner} that the $L^1$-norms of their Wigner functions blow up and exactly scale reciprocally to the $L^\infty$-norms, i.e.
\begin{align}
\label{eq:ReciScalWignerSec}
    \|W_{\rho_n}\|_1 \sim \|W_{\rho_n}\|^{-1}_\infty.
\end{align}
Hence, we can use the upper bound in Theorem~\ref{thm:WignerUpperBound} and the second lower bound in Theorem~\ref{thm:LowerBoundWigner} showing that
\begin{align}
\label{eq:FockThetaSec1}
   \NW(\rho_{n}) = \Theta\left(\Bigg(\frac{\left\|W_{\rho_n}\right\|_1}{\eps}\Bigg)^2\log\left(\frac{1}{\delta}\right)\right). 
\end{align}
We believe the existence of such a sequence of pure states satisfying \eqref{eq:VanishWignerSec} and \eqref{eq:ReciScalWignerSec} to be potentially of independent interest.

Lastly, for Gaussian states $\rho,$ the $L^1$-norm of the corresponding Wigner function is by the normalisation \eqref{eq:WignerNormalisation} equal to 1, i.e. $\|W_\rho\|_1=1.$ Combining this with the lower bound provided in Theorem~\ref{thm:LowerBoundWigner} we see in Proposition~\ref{prop:Gaussian}
\begin{align}
 \NW(\rho) = \Theta\left(\frac{1}{\eps^2}\log\left(\frac{1}{\delta}\right)\right)
\end{align}
for all pure Gaussian states $\rho.$

As a consequence of these results we find the following worst case characterisation of the sample complexity of fidelity estimation in the Wigner model: 
\begin{theorem}[Sample complexity of fidelity estimation in Wigner model in worst case]
\label{thm:WorstCaseWigener}
Let $0<\eps<1/2$ and $0<\delta<1/4$. Then we have for all $t\ge 1$
\begin{align}
\label{eq:WorstCaseWigner}
    \sup_{\substack{\rho \text{ pure state,}\\\|W_\rho\|_1\le t}} N^{\scalebox{0.65}{$W$}}_{\eps,\delta}(\rho) = \Theta\left(\left(\frac{t}{\eps}\right)^2\log\left(\frac{1}{\delta}\right)\right),
\end{align}
where the constants in the $\Theta(\placeholder)$ can be $m$-dependent. 
In particular, this shows that in worst case over all pure states $\rho$ the sample complexity of fidelity estimation in the Wigner model is infinite, i.e.
\begin{align}
    \sup_{\rho \text{ pure state}} N^{\scalebox{0.65}{$W$}}_{\eps,\delta}(\rho) = \infty.
\end{align}
\end{theorem}
In the following sections we formally state and prove the above mentioned results for the Fock, spike states and Gaussian states and then provide the proof of Theorem~\ref{thm:WorstCaseWigener}.

\subsubsection{Sample complexity for Fock states}
\label{sec:Fockstates}
\begin{proposition}[Sample complexity for Fock states]
\label{prop:FockStates} Let $0<\eps,\delta<1/4$ and $\left(\kb{n}\right)_{n\in\N}$ be the 1-mode Fock states. Then we have the following characterisation of the corresponding sample complexity of fidelity estimation in the Wigner model:
\begin{align}
\label{eq:FockTheta}
   \NW\left(\kb{n}\right) = \Theta\left(\Bigg(\frac{\left\|W_{\kb{n}}\right\|_1}{\eps}\Bigg)^2\log\left(\frac{1}{\delta}\right)\right)=\Theta\left(\frac{n}{\eps^2}\log\left(\frac{1}{\delta}\right)\right). 
\end{align}
\end{proposition}
\begin{remark}
As seen in Proposition~\ref{prop:WignerSmoothedLower} in the Appendix, the smoothed $L^1$-norm of the Wigner functions of the Fock states satisfy
\begin{align*}
    \left\| W_{\kb{n}}\right\|^{(\eps)}_1 \ge \left\|W_{\kb{n}}\right\|_1(1-C\eps^\beta)
\end{align*}
for all $\eps>0,\,\beta\in(0,1/2)$ and some $C>0$ possibly dependent on $\beta$ but independent of $n$ and $\eps.$ As discussed in the beginning of Appendix~\ref{sec:Tail}, this gives
\begin{align*}
 \inf_{\eps' \in[0,\eps)}\,\frac{\|W_{\kb{n}}\|_1^{(\eps')}}{\eps-\eps'}=\Theta\left(\frac{\|W_{\kb{n}}\|_1}{\eps}\right).
\end{align*}
Hence, for the Fock states, the upper bound on the sample complexity in terms of the smoothed $L^1$-function in Theorem~\ref{thm:WignerUpperBound} does essentially not become weaker when expressing it directly through the non-smoothed $L^1$-norm. 
This further justifies that we find in Proposition~\ref{prop:FockStates}  a characterisation of the sample complexity of fidelity estimation for the Fock states in terms of the $L^1$-norm of the corresponding Wigner functions. 
\end{remark}
For the proof of 
Proposition~\ref{prop:FockStates},
we use the fact that the corresponding Wigner functions can be expressed in terms of the Laguerre polynomials $L_n$ as
\cite[Equation 36]{leonhardt_measuring_1995}\cite[Equation 7.3]{Scheel_QuantumOptics}
\begin{align}
\label{eq:WignerFockSec}
    W_{\kb{n}}(\alpha) = \left(\frac{2}{\pi}\right)(-1)^n e^{-2|\alpha|^2}L_n(4|\alpha|^2).
\end{align}
Using the results of \cite{Markett_ScalingLpLaguerre_1982} on the $L^p$-norms of the Laguerre polynomials, we can deduce the following scaling of the $L^1$-norms of the Wigner function of the Fock states:
\begin{align}
\label{eq:WignerL1Scaling}
    \left\|W_{\kb{n}}\right\|_1 \sim \sqrt{n}.
\end{align}
This can be combined with Theorem~\ref{thm:WignerUpperBound} yielding the desired upper bound on $\NW\left(\kb{n}\right).$  
On the other hand, using the multiple asymptotic expressions of the Laguerre polynomials from \cite{szegho_orthogonal_1975,Askey_LaguerreExpansion_1965,landau_bessel_2000} for different regions in phase space, we show explicitly below that
\begin{align*}
    \left\|\frac{\left(W_{\kb{n+2}}-W_{\kb{n}}\right)^2}{(2/\pi)^2 - W^2_{\kb{n}}}\right\|_\infty\lesssim \frac{1}{n}.
\end{align*}
Applying this to the first lower bound in Theorem~\ref{thm:LowerBoundWigner} for specific choices of $\sigma_1$ and $\sigma_2$ provides the corresponding lower bound of $\NW\left(\kb{n}\right)$ 
in Proposition~\ref{prop:FockStates}.

\begin{proof}[Proof of Proposition~\ref{prop:FockStates}]
We first start by proving the corresponding upper bound in \eqref{eq:FockTheta}. We use the upper bound of Theorem~\ref{thm:WignerUpperBound} which gives
\begin{align*}
    \NW(\kb{n}) \lesssim \left(\frac{\left\|W_{\kb{n}}\right\|_1}{\eps}\right)^2\log\left(\frac{1}{\delta}\right).
\end{align*}
The Wigner function of the Fock states $\kb{n}$ can be written explicitly as \cite[Equation 36]{leonhardt_measuring_1995}\cite[Equation 7.3]{Scheel_QuantumOptics}
\begin{align}
\label{eq:WignerFock}
    W_{\kb{n}}(\alpha/2) = r(-1)^n e^{-|\alpha|^2/2}L_n(|\alpha|^2)
\end{align}
with $L_n$ denoting the \emph{Laguerre polynomial} of $n^{th}$ order and $r=2/\pi.$
From that we see that the scaling of the corresponding $L^1$-norms is given as
\begin{align}
\label{eq:FockWignerL1}
    \left\|W_{\kb{n}}\right\|_1 \sim \int^\infty_0 e^{-s^2/2}|L_n(s^2)| sds = \int^\infty_0 e^{-s/2}|L_n(s)| ds \sim \sqrt{n}, 
\end{align}
where we used \cite[Lemma 1]{Markett_ScalingLpLaguerre_1982} for the last step. 
Hence, we have shown
\begin{align}
\label{eq:FockUpperProof}
    \NW(\kb{n}) \lesssim \frac{n}{\eps^2}\log\left(\frac{1}{\delta}\right).
\end{align}

Let us finish the proof by providing the corresponding lower bound in \eqref{eq:FockTheta}: Since $\eps<1/4$ we can pick states $\sigma^{(n)}_1=(1-3\eps)\kb{n}+3\eps\kb{n+2}$ and $\sigma^{(n)}_2 = \kb{n}$, which satisfy \begin{align*}
    \left|F(\kb{n},\sigma^{(n)}_1)-F(\kb{n},\sigma^{(n)}_2)\right| =3\eps>2\eps
\end{align*}
Hence, we see from the first lower bound in Theorem~\ref{thm:LowerBoundWigner} that
\begin{align}
\label{eq:FockLowerBoundProof}
   \nn \NW(\kb{n}) &\gtrsim \left\|\frac{\left(W_{\sigma^{(n)}_1}-W_{\sigma^{(n)}_2}\right)^2}{r^2 - W^2_{\sigma^{(n)}_2}}\right\|^{-1}_\infty\log\left(\frac{1}{3\delta}\right)\\&= \frac{1}{9\eps^2}\left\|\frac{\left(W_{\kb{n+2}}-W_{\kb{n}}\right)^2}{r^2 - W^2_{\kb{n}}}\right\|^{-1}_\infty\log\left(\frac{1}{3\delta}\right).
\end{align}

In the following, we use  estimate the quotient appearing in infinity norm in the second line of \eqref{eq:FockLowerBoundProof} to find the desired lower bound on the sample complexity. We do this by using the explicit form of the Wigner functions in \eqref{eq:WignerFock} and provide bounds in three different regions of phase space, i.e. for phase space points $\alpha$ satisfying \begin{enumerate}
    \item $|\alpha|^2\le 1/\tilde n,$  \item $|\alpha|^2\in[\tilde n^{-1},1]$\quad or \item $|\alpha|^2>1.$
\end{enumerate}
Here, we introduced the short hand notation $\tilde n=n+1/2.$
For the first two regions we use the \emph{Hilb's type formula} \cite[Theorem 8.22.4]{szegho_orthogonal_1975} which gives for $|\alpha|^2\in[0,1]$
\begin{align}
\label{eq:HilbsLaguerre}
    e^{-|\alpha|^2/2}L_n(|\alpha|^2) = J_0(2\sqrt{\tilde n}|\alpha|) + \begin{cases}\mathcal{O}\left(|\alpha|^{4}\log(|\alpha|^{-2}n^{-1})\right),&\text{for $|\alpha|^2\in[0,\tilde n^{-1}]$}\vspace{0.2cm}\\\mathcal{O}\left(n^{-3/4}\right),&\text{for $|\alpha|^2\in [\tilde n^{-1},1]$}
    \end{cases}
\end{align}
 Here, the \emph{Bessel function} of order $0$ was denoted by
\begin{align}
\label{eq:DefBessel}
    J_0(y) = \sum_{k=0}^{\infty} \frac{(-1)^k\left(y/2\right)^{2k}}{(k!)^2} = 1- \left(\frac{y}{2}\right)^2 + R(y),
\end{align}
with the last equation holding by Taylor's theorem including up to third order for small $y,$ say $y\in[0,3],$ with $R(y)$ being the corresponding remainder. The remainder satisfies
\begin{align}
\label{eq:RemainderTaylorBessel}
    |R(y)| \le \frac{y^4}{4!} 
\end{align}
which follows by the fact\footnote{This can be seen by denoting the Bessel function of order $n\in\N$ by $J_n$ and using that $J'_0=-J_1$ and the recurrence relation $J'_n= \frac{1}{2}(J_{n-1}-J_{n+1})$ \cite[Equation 9.1.27]{abramowitz+stegun}  combined with the fact that $|J_n(y)|\le 1$ for all $y\in\R$ \cite[Equation 9.1.60]{abramowitz+stegun}.} that $|J^{(4)}_0(y)|\le 1$ together with the standard estimate for the remainder in Taylor's theorem \cite[Theorem 7.7]{apostol_calculus_1991}. Let us focus on the first phase space region close to the origin, i.e. $|\alpha|^2\le 1/\tilde n,$ in which case we
get from the series expression of the Bessel function
\begin{align*}
  \left|J_0(2\sqrt{\tilde n+2}|\alpha|)-J_0(2\sqrt{\tilde n}|\alpha|)\right| &= \left|\sum_{k=0}^{\infty}\frac{(-1)^k\left((\tilde n+2)^k-\tilde n^k\right)|\alpha|^{2k}}{(k!)^2}\right| \lesssim \sum_{k=1}^{\infty} \frac{\tilde n^{k-1} |\alpha|^{2k}}{k!} \\&\le |\alpha|^2 e^{\tilde n |\alpha|^2} \lesssim |\alpha|^2,
\end{align*}
where we used for the first inequality that $|(\tilde n+2)^k-\tilde n^k| =|\sum_{l=0}^{k-1}\tilde n^l2^{k-l}\binom{k}{l}| \le \tilde n^{k-1} 3^{k} $ together with the fact that $3^k/k! \le C$ for some $C>0.$ From this and \eqref{eq:HilbsLaguerre} we see
\begin{align*}
    &r^{-2}\left(W_{\kb{n+2}}(\alpha/2)-W_{\kb{n}}(\alpha/2)\right)^2 =\left(e^{-|\alpha|^2/2}L_{n+2}(|\alpha|^2)-e^{-|\alpha|^2/2}L_n(|\alpha|^2)\right)^2 \\&\lesssim \left(|\alpha|^2 + \mathcal{O}(|\alpha|^{4}\log(|\alpha|^{-2}n^{-1}))\right)^2 \lesssim |\alpha|^4,
\end{align*}
where we used in the last inequality that $|\alpha|^2\log(|\alpha|^{-2}) \le C$ and $|\alpha|^2\log(n) \le C$ for some $C>0.$
Furthermore, again for $|\alpha|^2\le 1/\tilde n,$ we have
\begin{align*}
1 - r^{-2}W^2_{\kb{n}}(\alpha/2) &=1-e^{-|\alpha|^2}L^2_n(|\alpha|^2)\\&=1-(J_0(2\sqrt{\tilde n}|\alpha|)+\mathcal{O}(|\alpha|^{4}\log(|\alpha|^{-2}n^{-1}))^2\\&=1-(1-n|\alpha|^2+ R(\sqrt{\tilde n}|\alpha|)+\mathcal{O}(|\alpha|^{4}\log(|\alpha|^{-2}n^{-1})))^2\\&\ge 1- (1-cn|\alpha|^2)^2 \ge cn|\alpha|^2
\end{align*}
where the first inequality holds for $n$ large enough and some $0<c<1$ and we used \eqref{eq:RemainderTaylorBessel}. Dividing both we see
\begin{align}
\label{eq:PhaseSpaceBound1}
  \sup_{0\le|\alpha|^2\le1/\tilde n}\frac{\left(W_{\kb{n+2}}(\alpha/2)-W_{\kb{n}}(\alpha/2)\right)^2}{r^2 - W^2_{\kb{n}}(\alpha/2)} \lesssim \sup_{0\le|\alpha|^2\le\frac{1}{\tilde n}}\frac{|\alpha|^2}{n}  \lesssim \frac{1}{n^2}.
\end{align}

For the other two regions in phase space we use \cite{Askey_LaguerreExpansion_1965}, in particular the table on page 699, which gives 
\begin{align}
\label{eq:AskeyBoundLaguerre}
\sup_{\alpha\in\C}\left(e^{-|\alpha|^2/2}\left(L_{n+2}(|\alpha|^2)-L_n(|\alpha|^2)\right)\right) = \mathcal{O}\left(n^{-1/2}\right).
\end{align}

Let us focus first on second region in phase space, i.e.  $|\alpha|^2\in[\tilde n^{-1},1]$: 
For those phase space points we use the fact that 
$|J_0(y)|\le \frac{0.8}{y^{1/3}}$ \cite{landau_bessel_2000} for all $y\in\R$ which gives by \eqref{eq:WignerFock} and \eqref{eq:HilbsLaguerre} that 
\begin{align}
\label{eq:PhaseSpace2Denom}
  1 - r^{-2}W^2_{\kb{n}}(\alpha/2)  =1-\left(J_0(2\sqrt{\tilde n}|\alpha|)+\mathcal{O}(n^{-3/4})\right)^2 \ge c
\end{align}
for $n$ large enough and some $c>0.$ Dividing now \eqref{eq:AskeyBoundLaguerre} by \eqref{eq:PhaseSpace2Denom} this shows 
\begin{align}
\label{eq:PhaseSpaceBound2}
     \nn&\sup_{|\alpha|^2\in[\tilde n^{-1},1]}\frac{\left(W_{\kb{n+2}}(\alpha/2)-W_{\kb{n}}(\alpha/2)\right)^2}{r^2 - W_{\kb{n}}(\alpha/2)} \\&\le \sup_{|\alpha|^2\in[\tilde n^{-1},1]} \frac{\left(e^{-|\alpha|^2/2}\left(L_{n+2}(|\alpha|^2)-L_n(|\alpha|^2)\right)\right)^2}{c} \lesssim \frac{1}{n}.
\end{align}
For the last phase space region, we use again \cite{Askey_LaguerreExpansion_1965} which gives for $|\alpha|>1,$ 
\begin{align*}
e^{-|\alpha|^2/2}L_n(|\alpha|^2) = \mathcal{O}\left(n^{-1/4}\right)
\end{align*}
and hence using again \eqref{eq:WignerFock}
\begin{align*}
    1 - r^{-2}W^2_{\kb{n}}(\alpha/2) \ge c
\end{align*}
for $n$ large enough and some $c>0.$
Hence, combining this again with \eqref{eq:AskeyBoundLaguerre} we see \begin{align}
\label{eq:PhaseSpaceBound3}
     &\sup_{|\alpha|>1}\frac{\left(W_{\kb{n+2}}(\alpha/2)-W_{\kb{n}}(\alpha/2)\right)^2}{r^2 - W^2_{\kb{n}}(\alpha/2)} \lesssim \frac{1}{n}
\end{align}
Therefore, combining \eqref{eq:PhaseSpaceBound1},
\eqref{eq:PhaseSpaceBound2} and \eqref{eq:PhaseSpaceBound3} we see
\begin{align*}
    \left\|\frac{\left(W_{\kb{n+2}}-W_{\kb{n}}\right)^2}{r^2 - W^2_{\kb{n}}}\right\|_\infty\lesssim \frac{1}{n}
\end{align*}
and therefore, from \eqref{eq:FockLowerBoundProof} 
\begin{align*}
    \NW(\kb{n}) \gtrsim \frac{n}{\eps^2}\log\left(\frac{1}{3\delta}\right).
\end{align*}
which together with the established \eqref{eq:FockUpperProof} finishes the proof.
\end{proof}
\subsubsection{Sample complexity for spike states}
\label{sec:SpreadedSpikeStates}
In this section we define the \emph{spike states} and analyse the scaling behaviour of the $L^1$-\, and $L^\infty$-norms of their Wigner functions.
In particular for every number of modes $m\in\N$ the spike states are a sequence of pure states $\left(\rho_n\right)_{n\in\N} =\left(\kb{\psi_n}\right)_{n\in\N}$ on $L^2(\R^m)$ whose Wigner functions satisfy
\begin{align}
\label{eq:VanishingWigner}    \|W_{\rho_n}\|_\infty\xrightarrow[n\to\infty]{}0
\end{align} 
Here, the idea of the construction is to pick wave functions $\psi_n$ with increasing uncertainty in both position and momentum which then leads to a further and further spreaded out Wigner function. We pick such wave functions by super positioning `spikes' in position space, which each have a high momentum uncertainty. In the particular construction of Lemma~\ref{lem:SpreadedWigner} these `spikes' are simply given by Gaussian functions. Super-positioning now sufficiently many at increasing distance also leads to high position uncertainty and to a point wise uniformly vanishing Wigner function. 

More precisely, the spike states are defined as follows: Let $\phi\in L^2(\R^m)$ be the normalised (within $L^2(\R^m)$, i.e. $\|\phi\|_2=1$) Gaussian function centered at the origin, i.e. for $x\in\R^m$
\begin{align}
\label{eq:L2Gaussian}
    \phi(x) = \left(\frac{2}{\pi}\right)^{m/4}\,e^{-x^2}.
\end{align}
Consider for $n\in\N$ the  wave function defined for $x\in\R^m$ by
\begin{align}
\label{eq:DefSpikeStates}
    \psi_n(x) = \frac{1}{\sqrt{c_n}}\sum_{k=1}^n \phi(x-\mu_k),
\end{align}
with \begin{align}
\label{eq:muKDef}
\mu_k= n\,3^k e_1,
\end{align} and $e_1=(1,0,\cdots,0)\in\R^m$ and where we drop the $n$ dependence for brevity. Furthermore, $c_n$ being the normalisation constant such that $\|\psi_n\|_2 =1$ which in particular satisfies
\begin{align}
\label{eq:cnLowerBound}
    c_n = \sum_{k,l=1}^n\int_{\R^m} \phi(x-\mu_k)\phi(x-\mu_l)dx \ge \sum_{k=1}^n\int_{\R^m} |\phi(x-\mu_k)|^2dx = n.
\end{align}
The $n^{th}$ \emph{spike state} is defined to be the corresponding pure state $\rho_n =\kb{\psi_n}.$ 

The particular spreading of the spikes parameterised by the $\mu_k$ in \eqref{eq:muKDef} is chosen in such a way that  for every fixed point in phase space essentially at most only two spikes contribute to the Wigner function in the formula \eqref{eq:WignerPure} (see the proof of Lemma~\ref{lem:SpreadedWigner} for more details).

Notably, we provide control of the speed of convergence for the limit \eqref{eq:VanishingWigner} and also on the blow up of the corresponding $L^1$-norms. In particular, we see that the the sequence of pure states satisfies
\begin{align*}
    \|W_{\rho_n}\|^{-1}_\infty \sim \|W_{\rho_n}\|_1
\end{align*}
(c.f.~Remark~\ref{rem:SpreadedWigner}). This is then used as  key insight for proving the worst case characterisation \eqref{eq:WorstCaseWigner} of the sample complexity in Theorem~\ref{thm:LowerBoundWigner}.

In the following lemma we provide the precise scaling behaviour of the $L^1$- and $L^\infty$-norms of the spike states.
\begin{lemma}[Asymptotic behaviour of spike states]
\label{lem:SpreadedWigner}
The pure states $\left(\rho_n\right)_{n\in\N}=\left(\kb{\psi_n}\right)_{n\in\N}$ on $L^2(\R^m)$ defined by \eqref{eq:DefSpikeStates} have Wigner functions satisfying $W_{\rho_n}\in L^1(\C^m)$ and 
\begin{align}
\label{eq:LinftyL1SpredWigner}
\|W_{\rho_n}\|_\infty\le  4\left(\frac{2}{\pi}\right)^m\frac{1}{n}\quad\text{ and }\quad\|W_{\rho_n}\|_1\le \, n.
\end{align}
Moreover, the corresponding sequence $\left(\psi_n\right)_{n\in\N}$ converges weakly to $0$ in $L^2(\R^m),$ i.e. for all $\varphi \in L^2(\R^m)$ we have
\begin{align}
\label{eq:WeakConvergenceTo0}
    \lim_{n\to\infty}\langle\varphi,\psi_n\rangle  = 0.
\end{align}
\end{lemma}
\begin{remark}
\label{rem:SpreadedWigner}
Note that the $n$-dependence in of the bounds \eqref{eq:LinftyL1SpredWigner} is tight. This can be seen by using that since $\rho_n$ is pure we have $\|W_{\rho_n}\|^2_2=\pi^{-m}$ by \eqref{eq:L2Wigner} and hence
\begin{align*}
    \|W_{\rho_n}\|^{-1}_\infty \le  \pi^{m} \|W_{\rho_n}\|_1. 
\end{align*}
Using this and the established bounds in \eqref{eq:LinftyL1SpredWigner} shows
\begin{align}
     \|W_{\rho_n}\|_\infty\sim \frac{1}{n}\quad\text{ and }\quad\|W_{\rho_n}\|_1\sim  n
\end{align}
\end{remark}
Before giving the proof of Lemma~\ref{lem:SpreadedWigner}, we show how we can use it to obtain a characterisation of the sample complexity $\NW(\rho_n)$ from the upper bound and lower bounds in Theorems~\ref{thm:WignerUpperBound} and~\ref{thm:LowerBoundWigner}:
\begin{proposition}[Sample complexity for spike states]
\label{prop:SpikeStates} Let $0<\eps<1/2,$ $0<\delta<1/4$ and $\left(\rho_n\right)_{n\in\N}=\left(\kb{\psi_n}\right)_{n\in\N}$ be the spike states defined in \eqref{eq:DefSpikeStates}. Then we have the following characterisation of the corresponding sample complexity of fidelity estimation in the Wigner model:
\begin{align*}
%\label{eq:FockTheta}
   \NW\left(\rho_n\right) = \Theta\left(\Bigg(\frac{\left\|W_{\rho_n}\right\|_1}{\eps}\Bigg)^2\log\left(\frac{1}{\delta}\right)\right),
\end{align*}
where the constants in the $\Theta(\placeholder)$ can be $m$-dependent.  
\end{proposition}
\begin{proof}
By the upper bound provided in Theorem~\ref{thm:WignerUpperBound} we immediately see
\begin{align*}
    \NW(\rho_n) \le 2^{2m+1}\left(\frac{\|W_{\rho_n}\|_1}{\eps}\right)^{2}\log\left(\frac{1}{\delta}\right)  \le 2^{2m+1}\left(\frac{n}{\eps}\right)^{2}\log\left(\frac{1}{\delta}\right),
\end{align*}
where for the second inequality we have used Lemma~\ref{lem:SpreadedWigner}.

To prove the corresponding lower bound we take $n' = \ceil{n/\eps}$ and state $\sigma_n = 2\eps \rho_n + (1-2\eps)\rho_{n'},$ which satisfies $F(\rho_n,\sigma_n) \ge 2\eps.$
Furthermore, we have by Lemma~\ref{lem:SpreadedWigner} that
\begin{align*}
    \|W_{\sigma_n}\|_\infty \le 2\eps\|W_{\rho_n}\|_\infty + (1-2\eps) \|W_{\rho_{n'}}\|_{\infty} \le 16\left(\frac{2}{\pi}\right)^m\frac{\eps}{n}.     \end{align*}
    Therefore, using the second lower bound in Theorem~\ref{thm:LowerBoundWigner} we see that
    \begin{align*}
    \NW(\rho_n)\ge \left(\frac{2}{\pi}\right)^{2m}\,\frac{1}{\|W_{\sigma_n}\|^2_\infty}\log\left(\frac{1}{3\delta}\right) \ge \frac{1}{16}\left(\frac{n}{\eps}\right)^2\log\left(\frac{1}{3\delta}\right)
    \end{align*}
    which finishes the proof.
\end{proof}

\begin{proof}[Proof of Lemma~\ref{lem:SpreadedWigner}]
We proof the result with phase space being parametrised with position and momentum variables $(x,p)\in\R^{2m}$ which corresponds to the change of coordinates given by $\alpha = \frac{1}{\sqrt{2}}\left(x+ip\right).$

Using \eqref{eq:WignerPure} and the explicit form of the wave functions in \eqref{eq:DefSpikeStates}, we can calculate the Wigner function of $\rho_n$ at phase space point $(x,p)\in\R^{2m}$ by
\begin{align}
\label{eq:WignerGaussSum}
    \nn&W_{\rho_n}(x,p) = \frac{1}{c_n}\left(\frac{2}{\pi}\right)^{3m/2}  \sum_{k,l=1}^n\int_{\R^m} e^{-(x-\mu_k+y)^2-(x-\mu_l-y)^2} \,e^{i2p\cdot y}\,dy\\&\nn=  \frac{1}{c_n}\left(\frac{2}{\pi}\right)^{3m/2}\sum_{k,l=1}^n e^{-(x-\mu_k)^2-(x-\mu_l)^2+\frac{(\mu_k-\mu_l)^2}{2}}\int_{\R^m} e^{-2(y-(\mu_k-\mu_l)/2)^2}e^{i2p\cdot y}dy \\&
     =\frac{1}{c_n}\left(\frac{2}{\pi}\right)^m\sum_{k,l=1}^ne^{-2\big(x-(\mu_k+\mu_l)/2\big)^2} \ e^{ip\cdot (\mu_k-\mu_l)-p^2/2},
     \end{align} 
where for the last line the standard relation for Fourier transforming a Gaussian function (see e.g. \cite[Lemma 7.3]{Teschl_mathematical_2014}).

Note that because of the specific choice of $\mu_k$ in \eqref{eq:muKDef}, we have for $x\in\R^m$ fixed that there exists at most one $(k_x,l_x)\in \N^2$ with $k_x\ge l_x$ such that for the euclidean norm we have \begin{align}
\label{eq:kxlx}
    \left\|x-\frac{\mu_{k_x}+\mu_{l_x}}{2}\right\| \le \frac{n}{2},
\end{align} 
which can be seen by the following argument: Assume for contradiction that there exists $(k,l)\neq(i,j)$ with $k\ge l$ and $i\ge j$ and such that $\|x-(\mu_k+\mu_l)/2\| \le n/2$ and $\|x-(\mu_i+\mu_j)/2\| \le n/2.$ This hence gives 
\begin{align}
\label{eq:MuCloseness2}
    \big\|\mu_k+\mu_l - \mu_i-\mu_j\big\| = n\left|3^k+3^l - 3^i- 3^j\right|\le 2n
\end{align}
Without loss of generality we can assume that $i>k$ as for $k=i$ we immediately also have $l=j$ from \eqref{eq:MuCloseness2}.
Using $k\ge l$ this gives 
\begin{align*}
    3^i +3^j \ge 3^i \ge 3^{k+1} \ge 3^k +3^l +3,
\end{align*}
which contradicts \eqref{eq:MuCloseness2} and hence proves the desired claim.

Using hence for $(x,p)\in\R^{2m}$ the uniqueness of $(k_x,l_x)\in \N^2$ with $k_x\ge l_x$ and satisfying \eqref{eq:kxlx} and plugging this into \eqref{eq:WignerGaussSum} gives together with \eqref{eq:cnLowerBound}
\begin{align*}
    |W_{\rho_n}(x,p)| &\le \frac{1}{n}\left(\frac{2}{ \pi }\right)^m\sum_{k,l=1}^n e^{-2\big(x-(\mu_k+\mu_l)/2\big)^2} \\&\le \frac{2}{n}\left(\frac{2}{\pi }\right)^m \left(e^{-2\big(x-(\mu_{k_x}+\mu_{l_x})/2\big)^2}+ \sum_{\substack{k,l=1\\k\ge l, (k,l)\neq(k_x,l_x)} }^n e^{-2\big(x-(\mu_k+\mu_l)/2\big)^2}\right) \\& \le \frac{2}{n}\left(\frac{2}{\pi}\right)^m\left(1+n^2 e^{- n^2/2}\right) \le \frac{4}{n}\left(\frac{2}{\pi}\right)^m,
     \end{align*} 
   where we used $t^2e^{-t^2/2}\le1$ for all $t\in\R$ and which hence shows the $L^\infty$-bound in \eqref{eq:LinftyL1SpredWigner}.

 On the other hand, using again the explicit form of the Wigner function of $\rho_n$ in \eqref{eq:WignerGaussSum}, we see
\begin{align*}
    \|W_{\rho_n}\|_1 & = \int_{\C^m}|W_\rho(\alpha)|d\alpha = 2^{-m}\int_{\R^{2m}}|W_\rho(x,p)|dxdp
    \\&\le \frac{1}{n\pi^m  }
    \sum_{k,l=1}^n\int_{\R^{2m}}e^{-2\big(x-(\mu_k+\mu_l)/2\big)^2} \ e^{-p^2/2}dxdp \\&= \frac{n}{\pi^m}
    \int_{\R^{2m}}e^{-x^2} e^{-p^2} dxdp =n.
\end{align*}

Lastly we want to show the weak convergence of the sequence $\left(\psi_n\right)_{n\in\N}$ to $0,$ i.e. \eqref{eq:WeakConvergenceTo0}. For that let $\varphi\in L^2(\R)$ and $\eps>0$ and pick $R\ge 0$ such that
\begin{align*}
    \int_{\|x\|> R}|\varphi(x)|^2 dx < \eps,
\end{align*}
which is possible due to the dominated convergence theorem. Furthermore, note that we have for the Gaussian function $\phi$ defined in \eqref{eq:L2Gaussian},$k\in[n]$ and $n3^k\ge R$ 
\begin{align*}
\int_{\|x\|\le R} |\phi(x-\mu_k)|^2dx \le \sqrt{\frac{2}{\pi}}\int_{-R}^{R} e^{-2(x_1 -n3^k)^2} dx_1 \le \sqrt{\frac{2}{\pi}}\int_{-R+n3^k}^{\infty}e^{-2x_1^2}dx_1 \le e^{-2(n3^k-R)^2} 
\end{align*}
where the first inequality follows by definition of $\mu_k$ in \eqref{eq:muKDef} and second inequality follows by standard Gaussian tail bounds.
Hence, using the Cauchy Schwarz inequality and the fact that $\psi_n$ is normalised in $L^2$, we see for $n$ large enough
\begin{align*}
    \left|\langle \varphi,\psi_n\rangle\right| &\le \sqrt{\int_{\|x\|> R}|\varphi(x)|^2dx} \ + \ \frac{\|\varphi\|_2}{\sqrt{c_n}}\sum_{k=1}^n\sqrt{\int_{\|x\|\le R}|\phi(x-\mu_k)|^2dx } \\&\le \sqrt{\eps} \ +\  \frac{\|\varphi\|_2}{\sqrt{c_n}}\sum_{k=1}^ne^{-(n3^k-R)^2} \xrightarrow[n\to\infty]{} \sqrt{\eps}.
\end{align*}
Since, $\eps>0$ was arbitrary, this shows \eqref{eq:WeakConvergenceTo0}.
\end{proof}
\subsubsection{Sample complexity for Gaussian states}
\label{sec:Gaussian}
\begin{proposition}
\label{prop:Gaussian}
Let $\rho$ be a pure Gaussian state on $L^2(\R^m)$ and $0<\eps,\delta<1/4.$ Then
\begin{align}
\NW(\rho) = \Theta\left(\frac{1}{\eps^2}\log\left(\frac{1}{\delta}\right)\right),
\end{align}
where the constants in the $\Theta(\placeholder)$ can be $m$-dependent.
\end{proposition}
\begin{proof}
For $\rho$ being a Gaussian state, we have that $W_\rho\ge0$ and hence by normalisation \eqref{eq:WignerNormalisation} 
\begin{align}
    \|W_\rho\|_1 = 1.
\end{align}
Therefore, we see from Theorem~\ref{thm:WignerUpperBound}
the upper bound
\begin{align*}
\NW(\rho) = \mathcal{O}\left(\frac{1}{\eps^2}\log\left(\frac{1}{\delta}\right)\right).
\end{align*}
For the corresponding lower bound we use the second inequality in \eqref{eq:InstanceBasedWigner} of Theorem~\ref{thm:LowerBoundWigner} for the sequence of feasible states given by $\sigma_n=2\eps \rho+(1-2\eps)\rho_n$ where $\left(\rho_n\right)_{n\in\N}$ is the sequence of spike states defined in Section~\ref{sec:SpreadedSpikeStates} which satisfies $\lim_{n\to\infty}\|W_{\rho_n}\|_\infty=0$ and hence  $\lim_{n\to\infty}\|W_{\sigma_n}\|_\infty=2\eps\|W_\rho\|_\infty \le 2\eps.$ 
\end{proof}
\subsubsection{Proof of Theorem~\ref{thm:WorstCaseWigener}}
In this section we give the proof of Theorem~\ref{thm:WorstCaseWigener}, which essentially follows the same argument as the proof of Proposition~\ref{prop:SpikeStates}.

\begin{proof}[Proof of Theorem~\ref{thm:WorstCaseWigener}]
Let $t\ge 1$ and use  Theorem~\ref{thm:WignerUpperBound} for the corresponding upper bound as
\begin{align*}
     \sup_{\substack{\rho \text{ pure state,}\\\|W_\rho\|_1\le t}} N_{\eps,\delta}(W_{\rho},\cS_W) &\le 2^{2m+1}\sup_{\substack{\rho \text{ pure state,}\\\|W_\rho\|_1\le t}}\left(\frac{\|W_\rho\|_1}{\eps}\right)^{2}\log\left(\frac{1}{\delta}\right)\\& \le 2^{2m+1}\left(\frac{t}{\eps}\right)^{2}\log\left(\frac{1}{\delta}\right) .
\end{align*}
For the corresponding lower bound in \eqref{eq:WorstCaseWigner} we take the second lower bound in Theorem~\ref{thm:LowerBoundWigner} for the specific choice $\rho \equiv \rho_n$ where $n = \floor{t}$ and $\left(\rho_{n}\right)_{n\in \N}$ being the spike states defined and discusses in Section~\ref{sec:SpreadedSpikeStates}. Note that $\rho_{n}$ is feasible in \eqref{eq:WorstCaseWigner} as by \eqref{eq:LinftyL1SpredWigner} we have $\|W_{\rho_{n}}\|_1\le  n\le t.$ Furthermore, let $\sigma_n = 2\eps \rho_n +(1-2\eps) \rho_{n'}$ where $n' = \ceil{\frac{t}{\eps}}.$ Note that for that choice we have $F(\rho_n,\sigma_n)\ge 2\eps$ and furthermore by \eqref{eq:LinftyL1SpredWigner} and the fact that $1/n\le 2/t$ since $t\ge1$ that
\begin{align*}
    \|W_{\sigma_n}\|_\infty \le 2\eps\|W_{\rho_n}\|_\infty + (1-2\eps) \|W_{\rho_{n'}}\|_{\infty} \le C\left(\frac{2}{\pi}\right)^m\frac{ \eps}{t} 
    \end{align*}
for some $C\ge 0$ independent of $t,\eps$ and $m.$ Hence, $\sigma_n$ is a feasible state in \eqref{eq:InstanceBasedWigner} and we obtain 
\begin{align*}
  \sup_{\substack{\rho \text{ pure state,}\\\|W_\rho\|_1\le t}} N_{\eps,\delta}(W_{\rho},\cS_W)  &\ge N_{\eps,\delta}(W_{\rho_n},\cS_W) \ge\left(\frac{2}{\pi}\right)^{2m} \frac{1}{\|W_{\sigma_n}\|^2_\infty}\log\left(\frac{1}{3\delta}\right) \\ &\ge C'\left(\frac{t}{\eps}\right)^2\log\left(\frac{1}{3\delta}\right)
\end{align*}
for some $C'\ge 0.$

\end{proof}

\section{Fidelity estimation for discrete variable systems}
\label{sec:PauliFidel}
We consider for $n\in\N$ an $n$-qubit quantum system on the Hilbert space $\cH= \left(\C^{2}\right)^{\otimes n}$ with dimension $d = 2^n.$ The corresponding set of Pauli strings is denoted by $\cP_n = \{\1,X,Y,Z\}^{\otimes n}.$
Note that $\cP_n$ forms an orthogonal basis on the space of operators on $ \left(\C^{2}\right)^{\otimes n}$ and we have for all $P,Q\in\cP_n$
\begin{align}
\label{eq:PauliONB}
    \frac{1}{d}\Tr(PQ) = \delta_{PQ}.
\end{align}
For $\rho$ a state on $\left(\C^{2}\right)^{\otimes n}$ we can define its \emph{characteristic function} as  
\begin{align*}
    \chi_\rho : \cP_n&\to [-1,1]\\
    P&\mapsto \Tr(P\rho).
\end{align*}

In the following we consider the task of estimating the fidelity between a desired $n$-qubit pure state $\rho,$  of which we have a theoretical description, and a general state $\sigma,$ which we prepared experimentally, by performing measurements on $\sigma.$ Since $\rho$ is a pure state, we can write the fidelity by using \eqref{eq:PauliONB} as
\begin{align}
\label{eq:FidelPauli}
    F(\rho,\sigma) = \frac{1}{d}\sum_{P\in \cP_n}\chi_\rho(P)\chi_\sigma(P) = \frac{1}{d} + \sum_{P\in \cP_n\setminus\{\1\}}\chi_\rho(P)\chi_\sigma(P)
    \end{align}

Assume that we are able to measure any observable $P\in\cP_n$ with outcome in $\{1,-1\}$ on the state $\sigma$. Here, the observables $P$ can be chosen adaptively and hence depend on the previous ones and their corresponding measurement outcomes. 
Since $\chi_\sigma(\1)=1$ we can restrict to Pauli measurements with $P\in\cP_n\setminus\{\1\}$ as for $P=\1$ a tester does not gain any information. 
In the following we refer to this model as the \emph{Pauli-parity model of fidelity estimation.}

For this model, $\eps>0$ and $\delta>0$ being the allowed additive error and probability of failure respectively and pure state $\rho$ fixed, which we call the \emph{instanced based setup,} we denote the \emph{sample complexity of fidelity estimation in the Pauli-parity model} by $N^{\scalebox{0.65}{$P$}}_{\eps,\delta}(\rho)$: That is $N^{\scalebox{0.65}{$P$}}_{\eps,\delta}(\rho)$ is the optimal number of measurements on different copies of the unknown state $\sigma$ needed to estimate $F(\rho,\sigma)$ in the worst case over all states $\sigma.$ 

 Due to the definition of the characteristic function $\chi_\sigma$, the assumption that we are able to perform measurements of the observables $P\in\cP_n$ on $\sigma$ exactly means that we have black box access to the $\chi_\sigma.$ Hence, estimating the fidelity \eqref{eq:FidelPauli} in this setup  becomes special instance of black box overlap estimation as discussed in Section~\ref{sec:BlackBox}: In particular we can make the identifications
\begin{align}
\label{eq:IdentificationsPauli}
    \nn\Omega &\mapsto \cP_n\setminus\{\1\} &&\text{and } \quad d\mu(\lambda) \mapsto \frac{1}{d}d\mu_{\text{count}}(P) \text{ with $\mu_{\text{count}}$ being the counting measure on $\cP_n\setminus\{\1\}$}  \\\nn
    r &\mapsto 1 \\\nn
    f(\lambda) &\mapsto \chi_\rho(P) &&\text{to which we have full access} \\ \nn
    g(\lambda) &\mapsto \chi_\sigma(P)&& \text{to which we have black box access} \\ 
      Y_\lambda &\mapsto Y_P &&\text{measurement outcome of observable $P\in\cP_n\setminus\{\1\}$ on state $\sigma.$}
\end{align}
In particular this leads to the following convention of $L^q$-norms for $q\in[1,\infty)$
\begin{align}
\label{eq:PauliLqNorm}
    \|\chi_\sigma\|_q := \left(\frac{1}{d}\sum_{P\in\cP_n\setminus\{\1\}}|\chi_\rho(P)|^q\right)^{1/q}.
\end{align}
and 
\begin{align}
    \|\chi_\sigma\|_\infty := \max_{P\in\cP_n\setminus\{\1\}}|\chi_\rho(P)|.
\end{align}
With these definition we have $\|\chi_\sigma\|_2 = \sqrt{\Tr(\sigma^2)-1/d}\le 1$ and hence by the Cauchy-Schwarz inequality
\begin{align}
\label{eq:L1PauliBound}
\|\chi_\rho\|_1 = \frac{1}{d}\sum_{P\in\cP_n\setminus\{\1\}}|\chi_\rho(P)| \le \sqrt{d}\|\chi_\rho\|_2 \le \sqrt{d}.
\end{align}

For the subset $\cS$ employed in Section~\ref{sec:BlackBox}, we are interested here in the particular choice $\cS_P$ consisting of all characteristic functions $\chi_\sigma$ for valid states $\sigma.$ We see that $\cS_P\subseteq B^2_1(0)\cap B^\infty_{1}(0).$

With the notation above, the sample complexity of fidelity estimation can be written as
\begin{align}
\label{eq:SampleComplexityFidelityPauli}
    N^{\scalebox{0.65}{$P$}}_{\eps,\delta}(\rho) \equiv  N_{\eps,\delta}(\chi_\rho,\cS_P)
\end{align}
with $N_{\eps,\delta}(\chi_\rho,\cS_P)$ being defined in Section~\ref{sec:BlackBox}.

 Formulating the achievability result Theorem~\ref{thm:blackboxOverEstimation} for this special case leads to the following theorem:

\begin{theorem}[Upper bound on fidelity estimation with Pauli measurements] \label{thm:PauliUpperBound}
Let $\rho$ be a pure state on $\left(\C^2\right)^{\otimes n}$ to which we have a theoretical description. Then, for any state $\sigma,$ we can estimate the fidelity $F(\rho,\sigma)$ with precision $\eps>0$ and failure probability at most $\delta>0$ by performing Pauli measurements on a number of
\begin{equation}
\label{eq:PauliFidelyyUpperBound}
    \NP(\rho) \le 2\inf_{\eps'\in [0,\eps)}\left(\frac{\|\chi_\rho\|^{(\eps')}_1}{\eps-\eps'}\right)^{2}\log\left(\frac{1}{\delta}\right)  
\end{equation}
copies of the state $\sigma.$  Furthermore, the upper bound holds true even when restricting to non-adaptive algorithms.
\end{theorem}
\begin{remark}[Worst case bound]
\label{rem:WorstCasePauli}
Using \eqref{eq:L1PauliBound}, we see that the sample complexity in the Pauli model satisfies for all pure states $\rho$
\begin{align}
\label{eq:WorstCasePauli}
\NP(\rho) = \mathcal{O}\left(\frac{d}{\eps^2}\log\left(\frac{1}{\delta}\right)\right) .
\end{align}
As we see below, this worst case bound is essentially tight as Theorems~\ref{thm:LowerBoundPauli} and \ref{thm:HaarRandomSample}  show a lower bound on the sample complexity for random Haar states which matches with high probability \eqref{eq:WorstCasePauli} up to  factors logarithmic in the dimension $d$. Furthermore, we believe that this tightness result can be strengthened by not relying on a randomised argument but, possibly inspired by the construction of the spike states in Section~\ref{sec:SpreadedSpikeStates}, finding explicit examples of states for which the lower bound in Theorem~\ref{thm:LowerBoundPauli} below can be shown to match \eqref{eq:WorstCasePauli} up to constants which are independent of the system size. We, however, leave this as an open problem.
\end{remark}

Furthermore, from the lower bound on the sample complexity of black box estimation in Theorem~\ref{thm:LowerBoundblackbox}, we find the corresponding lower bound on $N^{\scalebox{0.65}{$P$}}_{\eps,\delta}(\rho):$
\begin{theorem}[Lower bound on fidelity estimation with Pauli measurements]
\label{thm:LowerBoundPauli}
    Let $\rho$ be a pure state on $\left(\C^2\right)^{\otimes n}$ to which we have a theoretical description.  Then for $0<\eps<1/2$ and $0<\delta\le1/3$ we have the that the sample complexity of fidelity estimation in the Wigner model is lower bounded as
\begin{align}
\label{eq:InstanceBasedPauli} 
    \nn \NP(\rho) &\ge \sup_{\substack{\sigma_1,\sigma_2 \text{ state,}\\ |F(\rho,\sigma_1)-F(\rho,\sigma_2)| \,> \,2\eps}}\left\|\frac{\left(\chi_{\sigma_1}-\chi_{\sigma_2}\right)^2}{1 - \chi^2_{\sigma_2}}\right\|^{-1}_\infty\log\left(\frac{1}{3\delta}\right)\\&\ge
    \sup_{\substack{\sigma\text{ state,}\\ |F(\rho,\sigma)-1/d| \,>\,2\eps}}\,\frac{1}{\|\chi_\sigma\|^2_\infty}\log\left(\frac{1}{3\delta}\right).
\end{align}
\end{theorem}
\begin{proof}
The first inequality immediately follows from Theorem~\ref{thm:LowerBoundblackbox}. For the second we use the first with the choice $\sigma_2 = \1/d$ which satisfies $\chi_{\sigma_2}=0$ on $\cP_n\setminus\{\1\}$ and $F(\rho,\sigma_2)=1/d.$ \comment{From that we see  
\begin{align} 
    \nn \NP(\rho) &\ge \sup_{\substack{\sigma_1,\sigma_2 \text{ state,}\\ |F(\rho,\sigma_1)-F(\rho,\sigma_2)| \,> \,2\eps}}\left\|\frac{\left(\chi_{\sigma_1}-\chi_{\sigma_2}\right)^2}{1 - \chi^2_{\sigma_2}}\right\|^{-1}_\infty\log\left(\frac{1}{3\delta}\right)\\&\ge
    \sup_{\substack{\sigma\text{ state,}\\ |F(\rho,\sigma)-1/d| \,> \,2\eps}}\,\frac{1}{\|\chi_\sigma\|^2_\infty}\log\left(\frac{1}{3\delta}\right)
    \\&\ge
    \sup_{\substack{\sigma\text{ state,}\\ |F(\rho,\sigma)-1/d| \,\ge \,2\eps}}\,\frac{1}{\|\chi_\sigma\|^2_\infty}\log\left(\frac{1}{3\delta}\right),
\end{align}
where the last inequality follows due to $\eps < 1/2$ and the same argument as for \eqref{eq:OhMan}.}
\end{proof}
\begin{remark}[Matching upper and lower bound for small $\eps$] \label{rk:matching bound-small eps}
Inspired from the construction in Lemma~\ref{lem:tildeFGLowerbound} consider for some pure state $\rho$ and $\eps>0$ the operator
 $$\omega=\frac{1}{d}\left(\1+\frac{3\eps}{\|\chi_\rho\|_1}\sum_{P\in \mathbf{P}_n\setminus\{\1\} }  \sgn(\chi_\rho(P))P\right)$$ which has characteristic function $
    \chi_\omega(P) = \Tr(P\omega) = \frac{3\eps}{\|\chi_\rho\|_1}\sgn(\chi_\rho(P))$
for all $P\in\cP_n\setminus\{\1\}.$ This operator $\omega$ has normalised trace but is in general not positive semi-definite and hence not a state (c.f.~Remark~\ref{rem:GNoState}) unless $\eps>0$ is small enough. In particular, a naive operator norm bound shows that $\omega$ is in fact positive semi-definite if $0<\eps\le \tfrac{1}{3d(d+1)}.$ 

Whenever $\omega$ is a state we can combine \eqref{eq:InstanceBasedPauli} with the argument in Lemma~\ref{lem:tildeFGLowerbound} to show the instanced based lower bound on the sample complexity of fidelity estimation in the Pauli model
\begin{equation}
\label{eq:SmallEpsLower}
    \NP(\rho) \ge \left(\frac{\|\chi_\rho\|_1}{3\eps}\right)^2\log\left(\frac{1}{3\delta}\right).
\end{equation}
In fact note that 
\begin{align*}
    F(\rho,\omega) = \frac{1}{d}\sum_{P\in\cP_n}\chi_\rho(P)\chi_\omega(P) = \frac{1}{d}\left(1 +\sum_{P\in\cP_n\setminus\{\1\}}\frac{3\eps|\chi_\rho(P)|}{\|\chi_\rho\|_1}\right) = \frac{1}{d} + 3\eps
\end{align*}
from which we see that $\omega$ is a feasible state in the second line of~\eqref{eq:InstanceBasedPauli}. Furthermore, since by definition $\|\chi_\omega\|_\infty \le \frac{3\eps}{\|\chi_\rho\|_1}$ we can conclude \eqref{eq:SmallEpsLower}.
\end{remark}

\subsection{Matching upper and lower bounds for $\NP(\rho)$ for example of states}
\label{sec:MatchPauli}
In this section we consider two examples of families of states $\rho$ for which the upper and lower bounds on $\NP(\rho)$ provided in Theorems~\ref{thm:PauliUpperBound} and~\ref{thm:LowerBoundPauli} can be shown to be essentially matching. In particular we give a full characterisation of the sample complexity of fidelity estimation in the Pauli model for 
\begin{enumerate}
    \item Haar random states,
    \item stabiliser states,
\end{enumerate} 
in terms of the $L^1$-norms of the corresponding characteristic function.

For a Haar random pure state $\rho = \kb{\Psi}$ we find in Proposition~\ref{thm:HaarRandomSample} with probability at least $1-1/d$ that
\begin{align}
\label{eq:ThetaRandomPauliPrev}
    \NP(\rho) = \widetilde \Theta\left(\left(\frac{\|\chi_\rho\|_1}{\eps}\right)^2\log\left(\frac{1}{\delta}\right)\right)=\widetilde \Theta\left(\frac{d}{\eps^2}\log\left(\frac{1}{\delta}\right)\right).
\end{align}
Here, we used $\widetilde \Theta$ to denote the fact that the upper and lower bounds in \eqref{eq:ThetaRandomPauliPrev} match up to factors that scale logarithmically in the leading order term.

On the other hand, for stabiliser states $\rho$ on $n$-qubits, we have that the $L^1$-norm satisfies $\|\chi_\rho\|_1\le 1$ and is hence independent of the system size. Using this we can characterise in Proposition~\ref{prop:Stab} the corresponding sample complexity of fidelity estimation in the Pauli model as
\begin{align}
\label{eq:StabliserTheta}
    \NP(\rho) = \Theta\left(\frac{1}{\eps^2}\log\left(\frac{1}{\delta}\right)\right).
\end{align}
The upper bound in \eqref{eq:StabliserTheta} can already be found in \cite{Flammia_DirectFidelityEstimation_2011}.

\subsubsection{Sample complexity for Haar random states}
The following proposition shows that the upper and lower bounds of Theorems~\ref{thm:PauliUpperBound} and~\ref{thm:LowerBoundPauli} on the sample complexity $\NP(\rho)$ match (up to factors logarithmic in the leading order) for Haar random pure states with high probability.
\begin{proposition}
\label{thm:HaarRandomSample}
Let $\rho=\kb{\Psi}$ be a Haar random pure state on $\left(\C^2\right)^{\otimes n}$, $0<\eps,\delta\le1/5$. Then with probability at least $1-1/d$, with $d=2^n,$ we have
\begin{align}
\label{eq:ThetaRandomPauli}
    \NP(\rho) = \widetilde \Theta\left(\left(\frac{\|\chi_\rho\|_1}{\eps}\right)^2\log\left(\frac{1}{\delta}\right)\right)=\widetilde \Theta\left(\frac{d}{\eps^2}\log\left(\frac{1}{\delta}\right)\right).
\end{align}
Here, we used $\widetilde \Theta$ to denote the fact that the upper and lower bounds in \eqref{eq:ThetaRandomPauli} match up to factors that scale logarithmically in the leading order term.
\end{proposition}
\begin{proof}
For the upper bound we use Theorem~\ref{thm:PauliUpperBound} which gives for all pure states $\rho$
\begin{align*}
    \NP(\rho) \le 2 \left(\frac{\|\chi_\rho\|_1}{\eps}\right)^2\log\left(\frac{1}{\delta}\right).
\end{align*}

To prove the lower bound, we show in the following that with high probability a Haar random pure state $\rho=\kb{\Psi}$ satisfies 
\begin{align}
\label{eq:SmallPauliState}
	\|\chi_\rho\|_\infty \lesssim \sqrt{\frac{\log d}{d}}. 
\end{align}
Let $U \in \mathbb{U}(d)$ be a Haar unitary such that $\ket{\Psi}=U^\dagger \ket{0}$. For a fixed $P\in \cP_n\setminus\{\1\}$ we consider the random variable 
\begin{align*}
	f(U)= |\bra{\Psi}P \ket{\Psi}|=|\bra{0}UPU^\dagger  \ket{0}|.
\end{align*}
The function $f : \mathbb{U}(d) \rightarrow \mathbb{R}$ is $2$-Lipschitz with respect to the Frobenius norm denoted $\|\cdot\|_{\rm{F}}$. Let $U,V \in \mathbb{U}(d)$ be two unitary matrices, we have 
\begin{align*}
	|f(U)-f(V)| &= \left||\bra{0}UPU^\dagger  \ket{0}|-|\bra{0}VPV^\dagger  \ket{0}|\right|
	\\& \overset{(a)}{\le} |\bra{0}(U-V)PU^\dagger  \ket{0}|+ |\bra{0}VP(U-V)^\dagger  \ket{0}|
	%\\&\overset{(b)}{\le}  \sqrt{|\bra{0}(U-V)PP(U-V)^\dagger  \ket{0}|}\sqrt{|\bra{0}UU^\dagger  \ket{0}|}
 %\\&\quad + \sqrt{|\bra{0}VV^\dagger  \ket{0}|}\sqrt{|\bra{0}(U-V)PP(U-V)^\dagger  \ket{0}|}
 \\&\overset{(b)}{\le} \|P(U-V)^\dagger  \ket{0}\|_2\|U^\dagger  \ket{0}\|_2  + \|V^\dagger  \ket{0}\|_2\|P(U-V)^\dagger  \ket{0}\|_2  
	\\&\overset{(c)}{\le} 2\|P(U-V)^\dagger\|_{\rm{F}} =2\|U-V\|_{\rm{F}}
\end{align*}
where $(a)$ follows from the triangle inequality,  $(b)$ follows from the Cauchy Schwarz inequality and $(c)$ uses the fact $\|A\ket{x}\|_2\le \|A\|_{\rm{F}}$ for a pure $\ket{x}$ as well as  the Frobenius norm is unitarily invariant.  
So by the concentration inequality of Lipschitz functions of Haar unitaries~\cite{meckes2013spectral}, for all $s>0$:
\begin{align*}
\mathbb{P}\left(|f(U)-\mathbb{E}(f)|\ge s \right)\le \exp\left(-\frac{ds^2}{48}\right). 
\end{align*}
Moreover we have by the Cauchy Schwarz inequality
\begin{align*}
	\mathbb{E}\left( f \right) \le \sqrt{\mathbb{E}\left( f^2 \right) }= \sqrt{\mathbb{E}\left( \bra{\Psi}P \ket{\Psi}^2 \right) }= \sqrt{\frac{\tr(PP^\dagger)+|\tr(P)|^2}{d(d+1)}}= \sqrt{\frac{1}{d+1}}.
\end{align*}
So by the union bound 
\begin{align*}
	\mathbb{P}\left(\exists P\in \cP_n\setminus\{ I\}:  |\bra{\Psi}P \ket{\Psi}|> 2\sqrt{\frac{48\log d^3}{d}} \right)&\le \sum_{P\in \cP_n\setminus\{I\}}\mathbb{P}\left(|\bra{\Psi}P \ket{\Psi}|> 2\sqrt{\frac{48\log d^3}{d}} \right)
	\\&\le d^2 \mathbb{P}\left(f(U)-\mathbb{E}(f)>2\sqrt{\frac{48\log d^3}{d}}-\sqrt{\frac{1}{d+1}} \right)
	\\&\le d^2 \mathbb{P}\left(f(U)-\mathbb{E}(f)>\sqrt{\frac{48\log d^3}{d}} \right)
	\\&\le d^2\exp\left(-\frac{d\cdot \frac{48\log d^3}{d}}{48}\right)=\frac{1}{d}.
\end{align*}
Hence, when sampling from the Haar measure, with probability at least $1-1/d$ we find a pure state $\ket{\Psi}$ such that for all $P\in \cP_n\setminus \{ \1\}$
\begin{align*}
	|\bra{\Psi}P \ket{\Psi}|\le  2\sqrt{\frac{48\log d^3}{d}}
\end{align*}
and therefore the corresponding pure state $\rho =\kb{\Psi}$ satisfies \eqref{eq:SmallPauliState}. Consider the state $\sigma= 5\eps \rho +(1-5\eps)\1/d$ which satisfies $F(\rho,\sigma)-1/d =5\eps(1-1/d)>2\eps$ and has characteristic function $\chi_\sigma = 5\eps \chi_\rho$ on $\cP_n\setminus\{\1\}.$
Hence, this state is feasible in the optimisation in the second lower bound provided in Theorem~\ref{thm:LowerBoundPauli} which gives using \eqref{eq:L1PauliBound}
\begin{align}
\NP(\rho) \gtrsim \frac{1}{\eps^2\|\chi_\rho\|^2_\infty}\log\left(\frac{1}{3\delta}\right)\gtrsim \frac{d}{\eps^2\log d}\log\left(\frac{1}{3\delta}\right)\ge\left(\log d\right)^{-1} \left(\frac{\|\chi_\rho\|_1}{\eps}\right)^2\log\left(\frac{1}{3\delta}\right).
\end{align}
\comment{Here, we have used \eqref{eq:SmallPauliState} again for the last inequality which gives by definition of the $L^1$-norm in \eqref{eq:PauliLqNorm} 
\begin{align*}
\|\chi_\rho\|_1 = \frac{1}{d}\sum_{P\in\cP_n\setminus\{\1\}}|\chi_\rho(P)| \lesssim \sqrt{d \log(d)}.
\end{align*}}
\end{proof}

\subsubsection{Sample complexity for stabiliser states}
We say a pure state $\rho=\kb{\psi}$ on $\left(\C^2\right)^{\otimes n}$ is a stabiliser state if there are $2^n$ Pauli strings $P\in\cP_n$ satisfying $P\ket{\psi}=\ket{\psi}.$ 
\begin{proposition}
\label{prop:Stab}
Let $\rho=\kb{\psi}$ be a stabiliser state on $\left(\C^2\right)^{\otimes n}.$ Then we have for all $0<\eps,\delta<1/5$ that
\begin{align}
    \NP(\rho) = \Theta\left(\frac{1}{\eps^2}\log\left(\frac{1}{\delta}\right)\right).
\end{align}
\end{proposition}
\begin{proof}
It is well-known that all the $2^n$ Pauli strings $P\in\cP_n$ which satisfy $P\ket{\psi}=\psi$ commute and furthermore, for all remaining $Q\in\cP_n$ we have $\chi_\rho(Q)=\langle\psi,Q\psi\rangle =0.$ Hence, we have that the $L^1$-norm of the characteristic function of $\rho$ satisfies
\begin{align}
    \|\chi_\rho\|_1 = \frac{1}{d}\sum_{P\in\cP_n\setminus\{1\}}|\chi_\rho(P)| = \frac{d-1}{d} \le 1.
\end{align}
From Theorems~\ref{thm:PauliUpperBound} we hence see 
\begin{align}
    \NP(\rho) = \mathcal{O}\left(\frac{1}{\eps^2}\log\left(\frac{1}{\delta}\right)\right).
\end{align}
For the remaining lower bound we use the second line in \eqref{eq:InstanceBasedPauli} in Theorem~\ref{thm:LowerBoundPauli} for the state $\sigma= 5\eps \rho +(1-5\eps)\1/d$ which is feasible since $F(\rho,\sigma)-1/d =5\eps(1-1/d)>2\eps$ and satisfies $\|\chi_\sigma\|_\infty = \max_{P\in\cP_n\setminus\{\1\}}|\chi_\sigma(P)| \le 5\eps.$ 
\end{proof}

\appendix 

\section{Some facts about the smoothed $L^1$-norm}
For $\Omega$ a measure space with measure $\mu,$ we have defined in the Section~\ref{sec:LqNotation} the smoothed $L^1$-norm of a function $f\in L^2(\Omega)$ as
\begin{align}
\label{eq:SmoothedL1NormApp}
    \|f\|^{(\eps)}_{1} = \inf_{\substack{\widetilde f\in L^1(\Omega)\\ \|f-\widetilde f\|_2 \le \eps}} \|\widetilde f\|_1.
\end{align} 
Different from the scope of the rest of the paper, in which we restricted to real valued functions, we allow in this section for real or complex valued functions. Hence, for simplicity, we write as $L^q(\Omega)$ as the corresponding $L^q$-spaces without denoting the field over which the considered functions are taking values.

The following lemma establishes that in the limit $\eps\to 0$ the smoothed $L^1$-norm converges to the usual (non-smoothed) $L^1$-norm.
\begin{lemma}
\label{lem:ContSmoothedL1}
For all $f\in L^2(\Omega)$ we have
\begin{align}
\label{eq:ContSmoothedL1Norm}
    \lim_{\eps\downarrow 0}\|f\|^{(\eps)}_1 = \|f\|_1,
\end{align}
where the right hand side is understood as infinity if $f\notin L^1(\Omega).$ 
\end{lemma}
\begin{proof}
We use the relation for a measurable function $f$
\begin{align}
\label{eq:supIntegral}
    \int_\Omega |f(\lambda)| d\mu(\lambda) = \sup_{\substack{\Omega_0\subseteq \Omega\\\mu(\Omega_0)<\infty}} \int_{\Omega_0}|f(\lambda)|d\mu(\lambda),
\end{align}
which follows directly from the definition of the Lebesgue integral.

First we consider the case $f\in L^1(\Omega)$: Let $\delta>0.$ Using \eqref{eq:supIntegral} we can pick a set of finite measure $\Omega_0\subseteq\Omega$ such that 
\begin{align*}
    \int_{\Omega^c_0}|f(\lambda)| d\mu(\lambda) \le \frac{\delta}{2}.
\end{align*}
Let $\eps \in(0,\frac{\delta}{2\sqrt{\mu(\Omega_0})})$ and $\tilde f\in L^1(\Omega)\cap L^2(\Omega)$ be such that $\|f-\tilde f\|_2\le \eps.$ Then
\begin{align*}
    \|f\|_1 &\le \int_{\Omega_0}|f(\lambda) -\tilde f(\lambda)|d\mu(\lambda) + \int_{\Omega_0}|\tilde f(\lambda)|d\mu(\lambda) + \int_{\Omega^c_0}|f(\lambda)|d\mu(\lambda)\\ &\le  \int_{\Omega_0}|f(\lambda) -\tilde f(\lambda)|d\mu(\lambda) + \|\tilde f\|_1 + \frac{\delta}{2}.
\end{align*}
For the first term we use the Cauchy-Schwarz ineqality which gives
\begin{align*}
    \int_{\Omega_0}|f(\lambda) -\tilde f(\lambda)|d\mu(\lambda) \le \|f-\tilde f\|_2 \sqrt{\mu(\Omega_0)}\le \frac{\delta}{2}.
\end{align*}
Plugging this into the above and using that $\tilde f$ was arbitrary under the constraints above gives
\begin{align*}
    \|f\|^{\eps}_1\ge \|f\|_1 -\delta.
\end{align*}
Noting that the opposite inequality $\|f\|^{(\eps)}_1\le \|f\|_1$ is trivially true and using that $\delta>0$ was arbitrary, shows $\lim_{\eps\downarrow 0}\|f\|^{(\eps)}_1=\|f\|_1$.

To finish the proof we consider the case $f\notin L^1(\Omega).$ We have to show $\liminf_{\eps\downarrow 0}\|f\|^{(\eps)}_1 = \infty.$ Assume for contradiction that there exists $C>0$ such that for all $\eps>0$ we have the uniform bound 
\begin{align*}
    \|f\|^{(\eps)}_1\le C.
\end{align*}
Let $\Omega_0\subseteq\Omega$ be a set of finite measure. Note by similar arguments as above, it holds true that $\|f\mathbf{1}_{\Omega_0}\|^{(\eps)}_1\le \|f\|^{(\eps)}_1\le C$ for all $\eps>0.$ Hence, using the established \eqref{eq:ContSmoothedL1Norm} for the function $f\mathbf{1}_{\Omega_0}\in L^1(\Omega),$ we see
\begin{align*}
    \int_{\Omega_0}|f|d\mu(\lambda) = \lim_{\eps\downarrow 0}\|f\mathbf{1}_{\Omega_0}\|^{(\eps)}_1 \le C.
\end{align*}
But using that $\Omega_0$ was an arbitrary set of finite measure together with \eqref{eq:supIntegral} shows that $f\in L^1(\Omega),$ which gives a contradiction and finishes the proof.
\end{proof}

\subsection{Lower bound on smoothed $L^1$-norm for functions with rapidly decaying tail}
\label{sec:Tail}
The smoothed $L^1$-norm of a function $f$, defined in \eqref{eq:SmoothedL1Norm}, satisfies the trivial upper bound $\|f\|^{(\eps)}_1\le \|f\|_1$ for all $\eps>0.$ The gap here between smoothed and non-smoothed $L^1$-norm can in general be arbitrarily large 
even for $L^1$-functions for which the right hand side is finite. In particular, although by Lemma~\ref{lem:ContSmoothedL1} we know that $\lim_{\eps\downarrow 0}\|f\|^{(\eps)}_1=\|f\|_1,$ in general we have no control on the speed of convergence when taking this limit.

In this section, we, however, restrict to functions with \emph{rapidly decaying tail} (as made precise in \eqref{eq:DecayingTail} below) and for which we can provide good control in the relevant lower bound (c.f.~Lemma~\ref{lem:SmoothedLowerBound}). Here, for instances of interest, we obtain the lower bound
\begin{align}
\label{eq:LowerBoundSmoothed}
    \|f\|^{(\eps)}_1 \ge \|f\|_1(1 - C\eps^{\beta})
\end{align}
for some $C\ge 0$ and $\beta\in (0,1]$ independent of $\eps$ and the specific function $f.$ In particular \eqref{eq:LowerBoundSmoothed} holds true for $f$ being the Wigner functions $W_{\kb{n}}$ of the 1-mode Fock states as shown in Proposition~\ref{prop:WignerSmoothedLower} below. 

For functions satisfying \eqref{eq:LowerBoundSmoothed}, we find that\footnote{This can be seen by $\inf_{\eps' \in[0,\eps)}\,\frac{\|f\|_1^{(\eps')}}{\eps-\eps'}\ge   \frac{\|f\|_1}{\eps} - C \eps^{-(1-\beta)} = \Omega\left(\frac{\|f\|_1}{\eps}\right)$ together with $\|f\|^{(\eps)}_1\le \|f\|_1.$}
\begin{align}
\label{eq:Smoothed=NonSmoothed}
    \inf_{\eps' \in[0,\eps)}\,\frac{\|f\|_1^{(\eps')}}{\eps-\eps'}=\Theta\left(\frac{\|f\|_1}{\eps}\right)
\end{align}
for $\eps>0$ small enough. As seen in the previous sections, the sample complexity of black box overlap estimation and fidelity estimation in the Wigner- and Pauli-model can be upper bounded in terms of the smoothed $L^1$-norm with quantity of interest being the left hand side of \eqref{eq:Smoothed=NonSmoothed}
(c.f.~Theorems~\ref{thm:blackboxOverEstimation},  \ref{thm:WignerUpperBound} and~\ref{thm:PauliUpperBound}).    Therefore, by \eqref{eq:Smoothed=NonSmoothed} we see that for functions satisfying \eqref{eq:LowerBoundSmoothed}, we do not weaken this upper bound when choosing $\eps'=0$ and providing it in terms of the (non-smoothed) $L^1$-norm. This, hence, further justifies that the characterisation of the sample complexity of the Fock states found in Proposition~\ref{prop:FockStates} is in terms of the usual $L^1$-norms of the Wigner functions $W_{\kb{n}}.$

\bigskip 

\smallskip

In the remainder of this section we formally define the notion of functions with rapidly decaying tail and prove the mentioned lower bound on their smoothed $L^1$-norm in  Lemma~\ref{lem:SmoothedLowerBound}. After that we discuss as an example the relavant class of \emph{eventually exponentially decaying} functions on $\R^m,$ which includes the Wigner functions of Fock states. 

Let $\Omega$ be some measurable space with some measure $\mu$. For $\gamma,\kappa> 0$ and $\Omega_0\subseteq \Omega$ a finite measure set, we say a measurable function $f$ has a \emph{rapidly decaying tail of order $(\gamma,\kappa)$ outside $\Omega_0$} if for all $\delta>0$ we have \begin{align}
\label{eq:DecayingTail}
    \int_{\{|f|\le \delta\}\cap \Omega^c_0}\, |f(\lambda)| d\mu(\lambda) \le \kappa \,\delta^\gamma.
\end{align}
For such functions we have good control on the respective smoothed $L^1$-norm as proven in the following lemma. In particular the proof follows similar lines as the one of Lemma~\ref{lem:ContSmoothedL1}, however, providing additionally a convergence rate as $\eps\to 0$ by utilising the rapidly decaying tail assumption.
\begin{lemma}[Lower bound on smoothed $L^1$-norm]
\label{lem:SmoothedLowerBound}
Let $\eps,\gamma,\kappa>0$ and $\Omega_0\subseteq \Omega$ be a set of finite measure. Furthermore, let $f\in L^2(\Omega)$ with $\|f\|_2\le 1$ have a rapidly decaying tail of order $(\gamma,\kappa)$ outside $\Omega_0.$ Then
\begin{align}
\label{eq:LowerBoundSmoothL1}
    \|f\|^{(\eps)}_1 \ge \|f\|_1 -\kappa' \eps^{\beta} - \eps\sqrt{\mu(\Omega_0)}
\end{align}
with $\beta := \tfrac{\gamma}{\gamma+1}\in (0,1]$ and $\kappa':=2\kappa^{\frac{1}{1+\gamma}}.$
\end{lemma}
\begin{remark}
    Note that since $f$ in Lemma~\ref{lem:SmoothedLowerBound} is supposed to be in $L^2(\Omega)$ and have a rapidly decaying tail outside a finite measure set, it is immediately also in $L^1(\Omega).$ Hence, the right hand side of \eqref{eq:LowerBoundSmoothL1} is finite.
\end{remark}
\begin{proof}
Let $\tilde f\in L^1(\Omega)\cap L^2(\Omega)$ with $\|f-\tilde f\|_2\le \eps$ and $\delta>0$ to be determined later. We have 
\begin{align}
\label{eq:TriangleTailProof}
   \nn \|f\|_1 & =\int_{\{|f| > \delta\}} \left|f(\lambda)\right| d\mu(\lambda) +\int_{\{|f| \le \delta\}\cap\Omega_0} |f(\lambda)| d\mu(\lambda) +\int_{\{|f| \le \delta\}\cap\Omega^c_0} |f(\lambda)| d\mu(\lambda) \\& \le \int_{\{|f| > \delta\}} |f(\lambda)-\tilde f(\lambda)| d\mu(\lambda) + \int_{\Omega_0} |f(\lambda)-\tilde f(\lambda)| d\mu(\lambda)+ \|\tilde f\|_1 + \kappa \,\delta^\gamma.
\end{align}
For the first term we use
 with the Cauchy-Schwarz inequality together with the fact that $\|f\|_2\le 1$ which gives
\begin{align*}
\int_{\{|f| > \delta\}} |f(\lambda)-\tilde f(\lambda)| d\mu(\lambda) \le \|f-\tilde f\|_2\sqrt{\mu\left(\{|f| > \delta\}\right)} \le \frac{\eps}{\delta}. 
\end{align*}
Similarly, we see for the second term in \eqref{eq:TriangleTailProof} that
\begin{align*}
    \int_{\Omega_0} |f(\lambda)-\tilde f(\lambda)| d\mu(\lambda) \le \eps\sqrt{\mu(\Omega_0)}.
\end{align*}
Plugging this into \eqref{eq:TriangleTailProof} and using the fact that $\tilde f$ is an arbitrary feasible function in the optimisation corresponding to the smoothed $L^1$-norm in \eqref{eq:SmoothedL1Norm}, we see
\begin{align*}
    \|f\|^{(\eps)}_1  \ge\|f\|_1 - \eps\sqrt{\mu(\Omega_0)}-\frac{\eps}{\delta}-\kappa\,\delta^{\gamma}.
\end{align*}
Optimising over $\delta>0$ and choosing $\delta = \left(\frac{\eps}{\kappa}\right)^{1/(\gamma+1)}$ finishes the proof.
\end{proof}

For $m\in\N$ and the special case $\Omega = \R^m$ equipped with the Lebesgue measure, a class of particular interest which satisfy the condition \eqref{eq:DecayingTail} are \emph{eventually exponentially decaying functions}, i.e. $f\in L^1(\R^m)$ which satisfy for some $t,C>0$ and $\Omega_0\subset \R^m$ of finite measure
\begin{align}
\label{eq:ExpDecFunct}
    |f(\alpha)| \le C e^{-t|\alpha|}
\end{align}
    for all $\alpha\in \R^m\setminus \Omega_0.$ In fact such functions have rapidly decaying tail of order $(\gamma,\kappa_{C,\gamma,t,m})$ outside $\Omega_0$ for every $\gamma\in(0,1)$ fixed and $\kappa_{C,\gamma,t,m}$ being linear in $C$ and  $\kappa_{C,\gamma,t,m}/C$ depending on $\gamma,t,m$ but not $f.$ This  essentially follows due to the standard inequality\footnote{More precisely, to show that an exponentially decaying function $f$ satisfies \eqref{eq:DecayingTail} for all $a\in(0,1),$ we note that $ \int_{\{|f|\le \delta< e^{-t|\alpha|}\}} |f(\alpha)|d\alpha \le  \int_{0}^{\log((1/\delta)^{1/t})} \,\delta \, d\alpha = \frac{\delta}{t}\log(1/\delta).$}
\begin{align}
    \int_{\{e^{-t|\alpha|}\le \delta\}} e^{-t|\alpha|}d\alpha \lesssim \int_{\log((1/\delta)^{1/t})}^\infty  e^{-tr}  r^{m-1} dr \lesssim \delta\big(\log\left(1/\delta\right)+1\big),  
\end{align}
where constants independent of $\delta$ are hidden in the $\lesssim$-notation. 

Note that the Wigner functions of the (1-mode) Fock states $W_{\kb{n}},$ which are discussed in Section~\ref{sec:Fockstates} and given in terms of Laguerre polynomials \eqref{eq:WignerFockSec}, are eventually exponentially decaying in the sense of \eqref{eq:ExpDecFunct}. This can be seen directly from the asymptotic expressions of the Laguerre polynomials provided in\cite{Askey_LaguerreExpansion_1965} (consider there the table on page 699 for convenience). From that we can conclude the mentioned lower bound on their respective smoothed $L^1$-norm in the following proposition:

\begin{proposition}
\label{prop:WignerSmoothedLower}
The Wigner functions of the the 1-mode Fock states $\left(\kb{n}\right)_{n\in\N}$ satisfy   
    \begin{align}
       \left\| W_{\kb{n}}\right\|^{(\eps)}_1 \ge \left\|W_{\kb{n}}\right\|_1(1-C\eps^\beta) \sim \sqrt{n}(1-C\eps^\beta)
    \end{align}
    for all $\eps>0,\,\beta\in(0,1/2)$ and some $C>0$ possibly dependent on $\beta$ but independent of $n$ and $\eps.$
\end{proposition}
\begin{proof}
We use that the Wigner function of the Fock states $\kb{n}$ can be written explicitly as \cite[Equation 36]{leonhardt_measuring_1995}\cite[Equation 7.3]{Scheel_QuantumOptics} in terms of the Laguerre polynomials as
\begin{align*}
 W_{\kb{n}}(\alpha) = (-1)^n\left(\frac{2}{\pi}\right) e^{-2|\alpha|^2}L_n(4|\alpha|^2).
\end{align*}
Using the asymptotic expressions for the Laguerre polynomials in \cite{Askey_LaguerreExpansion_1965}, we see that 
\begin{align*}
    W_{\kb{n}}(\alpha) \le c\, e^{-t|\alpha|^2}
\end{align*}
for all $|\alpha|\ge s \sqrt{n}$  for some $s,c,t>0$ independent of $n.$ Therefore, by the discussion around \eqref{eq:ExpDecFunct} we see that $W_{\kb{n}}$ has a rapidly decaying tail of order $(\gamma,\kappa_{c,\gamma,t})$ outside of $\Omega_0 = B_{s\sqrt{n}}(0)\subset \C$ for all $\gamma\in(0,1).$ Here, $\kappa_{c,\gamma,s}>0$  depends linearly on $c$ and is independent of $n.$  Noting that the volume of $\Omega_0$ is of order $n,$ we see from Lemma~\ref{lem:SmoothedLowerBound}
that
\begin{align*}
   \left\|W_{\kb{n}}\right\|_1-\left\| W_{\kb{n}}\right\|^{(\eps)}_1 \lesssim  \eps^\beta + \sqrt{n}\eps \lesssim \sqrt{n}\,\eps^{\beta},
\end{align*}
where $\beta = \tfrac{\gamma}{\gamma+1}\in(0,1/2)$ and we are hiding all constants which are independent of $n$ and $\eps$ in the $\lesssim$-notation.
Using the fact that $\|W_{\kb{n}}\|_1\sim \sqrt{n}$ as noted in \eqref{eq:FockWignerL1} as a consequence of \cite[Lemma 1]{Markett_ScalingLpLaguerre_1982}, finishes the proof.
\end{proof}

\bigskip

\noindent\textbf{Acknowledgments.} The authors would like to thank Ulysse Chabaud, Giulia Ferrini, Niklas Galke, Naixu Guo, 
Benjamin Huard, Hector Hutin, Ludovico Lami, Lauritz van Luijk and Corentin Moumard for helpful discussions. OF and RS acknowledge funding from the European Research Council (ERC Grant AlgoQIP, Agreement No. 851716), from the QuantERA II Programme of the European Union’s Horizon 2020 research and innovation programme under Grant Agreement No 101017733 (VERIqTAS) as well as the government grant managed by the Agence Nationale de la Recherche under the Plan France 2030 with the reference ANR-22-PETQ-0007. AO acknowledges funding by the European Research Council (ERC Grant Agreement No. 948139).

\bibliography{Ref}

\newcommand{\etalchar}[1]{$^{#1}$}
\begin{thebibliography}{G{\'A}CFF20}

\bibitem[AGPF18]{albarelli2018resource}
Francesco Albarelli, Marco~G Genoni, Matteo~GA Paris, and Alessandro Ferraro.
\newblock Resource theory of quantum non-gaussianity and wigner negativity.
\newblock {\em Physical Review A}, 98(5):052350, 2018.

\bibitem[Apo91]{apostol_calculus_1991}
T.M. Apostol.
\newblock {\em Calculus, Volume 1}.
\newblock Wiley, 1991.

\bibitem[AS64]{abramowitz+stegun}
Milton Abramowitz and Irene~A. Stegun.
\newblock {\em Handbook of Mathematical Functions with Formulas, Graphs, and Mathematical Tables}.
\newblock Dover, New York, ninth dover printing, tenth gpo printing edition, 1964.

\bibitem[AW65]{Askey_LaguerreExpansion_1965}
Richard Askey and Stephen Wainger.
\newblock Mean convergence of expansions in laguerre and hermite series.
\newblock {\em American Journal of Mathematics}, 87(3):695--708, 1965.

\bibitem[BAM{\etalchar{+}}02]{Bertet_ExpLutterDirectMeasWigner_2002}
P.~Bertet, A.~Auffeves, P.~Maioli, S.~Osnaghi, T.~Meunier, M.~Brune, J.~M. Raimond, and S.~Haroche.
\newblock Direct measurement of the wigner function of a one-photon fock state in a cavity.
\newblock {\em Phys. Rev. Lett.}, 89:200402, Oct 2002.

\bibitem[BCE22]{Ulysse_Cont=WignNeg_2022}
Robert~I. Booth, Ulysse Chabaud, and Pierre-Emmanuel Emeriau.
\newblock Contextuality and wigner negativity are equivalent for continuous-variable quantum measurements.
\newblock {\em Phys. Rev. Lett.}, 129:230401, Nov 2022.

\bibitem[BDLR24]{Becker_CVLearning_2024}
Simon Becker, Nilanjana Datta, Ludovico Lami, and Cambyse Rouze.
\newblock Classical shadow tomography for continuous variables quantum systems.
\newblock {\em IEEE Transactions on Information Theory}, 70(5):3427--3452, 2024.

\bibitem[BR97]{BratteliRobinson_operator2_1997}
Ola Bratteli and Derek~W. Robinson.
\newblock {\em Operator Algebras and Quantum Statistical Mechanics II}.
\newblock Springer Berlin Heidelberg, 1997.

\bibitem[BSBN02]{Bartlett_EffecienClassSimContVar_2002}
Stephen~D. Bartlett, Barry~C. Sanders, Samuel~L. Braunstein, and Kae Nemoto.
\newblock Efficient classical simulation of continuous variable quantum information processes.
\newblock {\em Phys. Rev. Lett.}, 88:097904, Feb 2002.

\bibitem[Cam11]{campbell2011catalysis}
Earl~T Campbell.
\newblock Catalysis and activation of magic states in fault-tolerant architectures.
\newblock {\em Physical Review A—Atomic, Molecular, and Optical Physics}, 83(3):032317, 2011.

\bibitem[CDG{\etalchar{+}}20]{Chabaud_BuildingTrust_2020}
Ulysse Chabaud, Tom Douce, Fr\'{e}d\'{e}ric Grosshans, Elham Kashefi, and Damian Markham.
\newblock {Building Trust for Continuous Variable Quantum States}.
\newblock In Steven~T. Flammia, editor, {\em 15th Conference on the Theory of Quantum Computation, Communication and Cryptography (TQC 2020)}, volume 158 of {\em Leibniz International Proceedings in Informatics (LIPIcs)}, pages 3:1--3:15, Dagstuhl, Germany, 2020. Schloss Dagstuhl -- Leibniz-Zentrum f{\"u}r Informatik.

\bibitem[CFF23]{calcluth2023vacuum}
Cameron Calcluth, Alessandro Ferraro, and Giulia Ferrini.
\newblock Vacuum provides quantum advantage to otherwise simulatable architectures.
\newblock {\em Physical Review A}, 107(6):062414, 2023.

\bibitem[CLO22]{chen2022toward}
Sitan Chen, Jerry Li, and Ryan O’Donnell.
\newblock Toward instance-optimal state certification with incoherent measurements.
\newblock In {\em Conference on Learning Theory}, pages 2541--2596. PMLR, 2022.

\bibitem[DG17]{de_gosson_wigner_2017}
M.A. De~Gosson.
\newblock {\em The Wigner Transform}.
\newblock Advanced Textbooks In Mathematics. World Scientific Publishing Company, 2017.

\bibitem[dSLCP11]{Silva_PracticalFidel_2011}
Marcus~P. da~Silva, Olivier Landon-Cardinal, and David Poulin.
\newblock Practical characterization of quantum devices without tomography.
\newblock {\em Phys. Rev. Lett.}, 107:210404, Nov 2011.

\bibitem[FL11]{Flammia_DirectFidelityEstimation_2011}
Steven~T. Flammia and Yi-Kai Liu.
\newblock Direct fidelity estimation from few pauli measurements.
\newblock {\em Phys. Rev. Lett.}, 106:230501, Jun 2011.

\bibitem[FMcv11]{Radim_PosWignerbeyonMixGauss_2011}
Radim Filip and Ladislav Mi\ifmmode~\check{s}\else \v{s}\fi{}ta.
\newblock Detecting quantum states with a positive wigner function beyond mixtures of gaussian states.
\newblock {\em Phys. Rev. Lett.}, 106:200401, May 2011.

\bibitem[G{\'A}CFF20]{garcia2020efficient}
Laura Garc{\'\i}a-{\'A}lvarez, Cameron Calcluth, Alessandro Ferraro, and Giulia Ferrini.
\newblock Efficient simulatability of continuous-variable circuits with large wigner negativity.
\newblock {\em Physical Review Research}, 2(4):043322, 2020.

\bibitem[GAG{\etalchar{+}}23]{gandhari2023precisionboundscontinuousvariablestate}
Srilekha Gandhari, Victor~V. Albert, Thomas Gerrits, Jacob~M. Taylor, and Michael~J. Gullans.
\newblock Precision bounds on continuous-variable state tomography using classical shadows, 2023.

\bibitem[GPR23]{Guo2023May}
Naixu Guo, Feng Pan, and Patrick Rebentrost.
\newblock {Estimating properties of a quantum state by importance-sampled operator shadows}, 2023.

\bibitem[HFH{\etalchar{+}}22]{hahn2022quantifying}
Oliver Hahn, Alessandro Ferraro, Lina Hultquist, Giulia Ferrini, and Laura Garc{\'\i}a-{\'A}lvarez.
\newblock Quantifying qubit magic resource with gottesman-kitaev-preskill encoding.
\newblock {\em Physical Review Letters}, 128(21):210502, 2022.

\bibitem[HFT24]{hahn2024bridging}
Oliver Hahn, Giulia Ferrini, and Ryuji Takagi.
\newblock Bridging magic and non-gaussian resources via gottesman-kitaev-preskill encoding.
\newblock {\em arXiv preprint arXiv:2406.06418}, 2024.

\bibitem[Hoe63]{hoeff}
Wassily Hoeffding.
\newblock Probability inequalities for sums of bounded random variables.
\newblock {\em J. Amer. Statist. Assoc.}, 58:13--30, 1963.

\bibitem[HPS24]{Huan_Certifying_2024}
Hsin-Yuan Huang, John Preskill, and Mehdi Soleimanifar.
\newblock Certifying almost all quantum states with few single-qubit measurements, 2024.

\bibitem[Hud74]{HUDSON1974249}
R.L. Hudson.
\newblock When is the wigner quasi-probability density non-negative?
\newblock {\em Reports on Mathematical Physics}, 6(2):249--252, 1974.

\bibitem[K{\.Z}04]{kenfack_negativity_2004}
Anatole Kenfack and Karol {\.Z}yczkowski.
\newblock Negativity of the wigner function as an indicator of non-classicality.
\newblock {\em Journal of Optics B: Quantum and Semiclassical Optics}, 6(10):396, 2004.

\bibitem[Lan00]{landau_bessel_2000}
L.~J. Landau.
\newblock Bessel functions: Monotonicity and bounds.
\newblock {\em Journal of the London Mathematical Society}, 61(1):197--215, 2000.
\newblock \_eprint: https://onlinelibrary.wiley.com/doi/pdf/10.1112/S0024610799008352.

\bibitem[LD97]{Lutterbach_DirectMeasWigner_1997}
L.~G. Lutterbach and L.~Davidovich.
\newblock Method for direct measurement of the wigner function in cavity qed and ion traps.
\newblock {\em Phys. Rev. Lett.}, 78:2547--2550, Mar 1997.

\bibitem[LP95]{leonhardt_measuring_1995}
U.~Leonhardt and H.~Paul.
\newblock Measuring the quantum state of light.
\newblock {\em Progress in Quantum Electronics}, 19(2):89--130, 1995.

\bibitem[LR09]{Lvovsky_ContinuousHomo_2009}
A.~I. Lvovsky and M.~G. Raymer.
\newblock Continuous-variable optical quantum-state tomography.
\newblock {\em Rev. Mod. Phys.}, 81:299--332, Mar 2009.

\bibitem[Mar82]{Markett_ScalingLpLaguerre_1982}
Clemens Markett.
\newblock Mean ces{\`a}ro summability of laguerre expansions and norm estimates with shifted parameter.
\newblock {\em Analysis Mathematica}, 8:19--37, 1982.

\bibitem[MDR{\etalchar{+}}24]{Huard_Wignertomoagraphy_2024}
Antoine Marquet, Simon Dupouy, Ulysse Réglade, Antoine Essig, Joachim Cohen, Emanuele Abertinale, Audrey Bienfait, Théau Peronnin, Sébastien Jezouin, Raphaël Lescanne, and Benjamin Huard.
\newblock Harnessing two-photon dissipation for enhanced quantum measurement and control, 2024.

\bibitem[ME12]{MariEisert_SimPosWigner_2012}
A.~Mari and J.~Eisert.
\newblock Positive wigner functions render classical simulation of quantum computation efficient.
\newblock {\em Phys. Rev. Lett.}, 109:230503, Dec 2012.

\bibitem[MEC{\etalchar{+}}24]{Antoni_Wigner_2024}
A.~Marquet, A.~Essig, J.~Cohen, N.~Cottet, A.~Murani, E.~Albertinale, S.~Dupouy, A.~Bienfait, T.~Peronnin, S.~Jezouin, R.~Lescanne, and B.~Huard.
\newblock Autoparametric resonance extending the bit-flip time of a cat qubit up to 0.3 s.
\newblock {\em Phys. Rev. X}, 14:021019, Apr 2024.

\bibitem[MKC09]{Mandilara_ExtendingHudson_2009}
A.~Mandilara, E.~Karpov, and N.~J. Cerf.
\newblock Extending hudson's theorem to mixed quantum states.
\newblock {\em Phys. Rev. A}, 79:062302, Jun 2009.

\bibitem[MM13]{meckes2013spectral}
Elizabeth Meckes and Mark Meckes.
\newblock Spectral measures of powers of random matrices.
\newblock {\em Electronic communications in probability}, 18, 2013.

\bibitem[MMB{\etalchar{+}}24]{Mele_CVLearning_2024}
Francesco~Anna Mele, Antonio~Anna Mele, Lennart Bittel, Jens Eisert, Vittorio Giovannetti, Ludovico Lami, Lorenzo Leone, and Salvatore F.~E. Oliviero.
\newblock Learning quantum states of continuous variable systems, 2024.

\bibitem[NRO{\etalchar{+}}00]{Lutter_ExperimentWigner_2000}
G.~Nogues, A.~Rauschenbeutel, S.~Osnaghi, P.~Bertet, M.~Brune, J.~M. Raimond, S.~Haroche, L.~G. Lutterbach, and L.~Davidovich.
\newblock Measurement of a negative value for the wigner function of radiation.
\newblock {\em Phys. Rev. A}, 62:054101, Oct 2000.

\bibitem[PWB15a]{pashayan2015estimating}
Hakop Pashayan, Joel~J Wallman, and Stephen~D Bartlett.
\newblock Estimating outcome probabilities of quantum circuits using quasiprobabilities.
\newblock {\em Physical review letters}, 115(7):070501, 2015.

\bibitem[PWB15b]{Pashayan2015Aug}
Hakop Pashayan, Joel~J. Wallman, and Stephen~D. Bartlett.
\newblock {Estimating Outcome Probabilities of Quantum Circuits Using Quasiprobabilities}.
\newblock 115(7):070501, 2015.

\bibitem[RBG{\etalchar{+}}24]{reglade_quantum_2024}
U.~R{\'e}glade, A.~Bocquet, R.~Gautier, J.~Cohen, A.~Marquet, E.~Albertinale, N.~Pankratova, M.~Hall{\'e}n, F.~Rautschke, L.-A. Sellem, P.~Rouchon, A.~Sarlette, M.~Mirrahimi, P.~Campagne-Ibarcq, R.~Lescanne, S.~Jezouin, and Z.~Leghtas.
\newblock Quantum control of a cat qubit with bit-flip times exceeding ten seconds.
\newblock {\em Nature}, 629(8013):778--783, 2024.
\newblock Publisher: Nature Publishing Group.

\bibitem[Roy77]{Royer_WignerAsExpOfParity_1977}
Antoine Royer.
\newblock Wigner function as the expectation value of a parity operator.
\newblock {\em Phys. Rev. A}, 15:449--450, Feb 1977.

\bibitem[SC83]{Soto_HudsonThMultimode_1983}
Francisco Soto and Pierre Claverie.
\newblock When is the wigner function of multidimensional systems nonnegative?
\newblock {\em Journal of Mathematical Physics}, 24(1):97--100, 1983.
\newblock \_eprint: https://pubs.aip.org/aip/jmp/article-pdf/24/1/97/19128812/97\_1\_online.pdf.

\bibitem[Sch]{Scheel_QuantumOptics}
Stefan Scheel.
\newblock Quantum optics.
\newblock {\em Available at \url{https://www.lkv.uni-rostock.de/storages/uni-rostock/Alle_MNF/Physik_Qms/Lehre_Scheel/quantenoptik/Quantenoptik-Vorlesung4.pdf}}.

\bibitem[SEL{\etalchar{+}}22]{Sivak_HectorPaper_2022}
V.~V. Sivak, A.~Eickbusch, H.~Liu, B.~Royer, I.~Tsioutsios, and M.~H. Devoret.
\newblock Model-free quantum control with reinforcement learning.
\newblock {\em Phys. Rev. X}, 12:011059, Mar 2022.

\bibitem[Sim92]{Simon_WignerWeyl_1992}
Barry Simon.
\newblock The weyl transform and lp functions on phase space.
\newblock {\em Proceedings of the American Mathematical Society}, 116(4):1045--1047, 1992.

\bibitem[Sze75]{szegho_orthogonal_1975}
G.~Szeg\"o.
\newblock {\em Orthogonal Polynomials}.
\newblock American Math. Soc: Colloquium publ. American Mathematical Society, 1975.

\bibitem[TCJ20]{Kok_NegofQuasi_2020}
Kok~Chuan Tan, Seongjeon Choi, and Hyunseok Jeong.
\newblock Negativity of quasiprobability distributions as a measure of nonclassicality.
\newblock {\em Phys. Rev. Lett.}, 124:110404, Mar 2020.

\bibitem[Tes14]{Teschl_mathematical_2014}
Gerald Teschl.
\newblock {\em Mathematical methods in quantum mechanics: with applications to Schr{\"o}dinger operators}.
\newblock Number volume 157 in Graduate studies in mathematics. American Mathematical Society, second edition edition, 2014.

\bibitem[VKL{\etalchar{+}}13a]{Vlastakis_ScienceWignerMeas_2013}
Brian Vlastakis, Gerhard Kirchmair, Zaki Leghtas, Simon~E. Nigg, Luigi Frunzio, S.~M. Girvin, Mazyar Mirrahimi, M.~H. Devoret, and R.~J. Schoelkopf.
\newblock Deterministically encoding quantum information using 100-photon schr{\"o}dinger cat states.
\newblock {\em Science}, 342(6158):607--610, 2013.

\bibitem[VKL{\etalchar{+}}13b]{Vlast_Science_2013}
Brian Vlastakis, Gerhard Kirchmair, Zaki Leghtas, Simon~E. Nigg, Luigi Frunzio, S.~M. Girvin, Mazyar Mirrahimi, M.~H. Devoret, and R.~J. Schoelkopf.
\newblock Deterministically encoding quantum information using 100-photon schr{\"o}dinger cat states.
\newblock {\em Science}, 342(6158):607--610, 2013.

\bibitem[VV17]{valiant2017automatic}
Gregory Valiant and Paul Valiant.
\newblock An automatic inequality prover and instance optimal identity testing.
\newblock {\em SIAM Journal on Computing}, 46(1):429--455, 2017.

\bibitem[VWFE13]{veitch2013efficient}
Victor Veitch, Nathan Wiebe, Christopher Ferrie, and Joseph Emerson.
\newblock Efficient simulation scheme for a class of quantum optics experiments with non-negative wigner representation.
\newblock {\em New Journal of Physics}, 15(1):013037, 2013.

\bibitem[Wer84]{Werner_QuantumHarmonicAnal_1984}
R.~Werner.
\newblock Quantum harmonic analysis on phase space.
\newblock {\em Journal of Mathematical Physics}, 25(5):1404--1411, 1984.
\newblock \_eprint: https://pubs.aip.org/aip/jmp/article-pdf/25/5/1404/11262255/1404\_1\_online.pdf.

\bibitem[WGR{\etalchar{+}}16]{ScienceWigner2_2016}
Chen Wang, Yvonne~Y. Gao, Philip Reinhold, R.~W. Heeres, Nissim Ofek, Kevin Chou, Christopher Axline, Matthew Reagor, Jacob Blumoff, K.~M. Sliwa, L.~Frunzio, S.~M. Girvin, Liang Jiang, M.~Mirrahimi, M.~H. Devoret, and R.~J. Schoelkopf.
\newblock A schr{\"o}dinger cat living in two boxes.
\newblock {\em Science}, 352(6289):1087--1091, 2016.

\end{thebibliography}
\bibliographystyle{alpha}

\end{document}